\documentclass[11pt,a4paper]{article}
\usepackage{amssymb}
\usepackage{amsmath}
\usepackage{amsfonts}
\usepackage{dsfont}
\usepackage{amsthm}
\usepackage{mathrsfs}
\usepackage{hyperref}
\usepackage{color}
\usepackage[margin=2.41cm]{geometry}
\usepackage[all,cmtip]{xy}
\usepackage[utf8]{inputenc}
\usepackage{graphicx}
\usepackage{varwidth}
\usepackage{comment}

\usepackage{upgreek}
\usepackage{rotating}

\usepackage{tikz}
\usetikzlibrary{shapes.geometric}

\usepackage[shortlabels]{enumitem}

\definecolor{darkred}{rgb}{0.8,0.1,0.1}
\hypersetup{
     colorlinks=false,         
     linkcolor=darkred,
     citecolor=blue,
}

\theoremstyle{plain}
\newtheorem{theo}{Theorem}[section]

\newtheorem{propo}[theo]{Proposition}
\newtheorem{cor}[theo]{Corollary}

\theoremstyle{definition}
\newtheorem{defi}[theo]{Definition}

\newtheorem{exx}[theo]{Example}
\newenvironment{ex}
  {\pushQED{\qed}\exx}
  {\popQED\endexx}

\newtheorem{remm}[theo]{Remark}
\newenvironment{rem}
  {\pushQED{\qed}\remm}
  {\popQED\endremm}

\numberwithin{equation}{section}

\def\nn{\nonumber}

\def\bbK{\mathbb{K}}
\def\bbR{\mathbb{R}}

\def\bbZ{\mathbb{Z}}

\def\Hom{\mathrm{Hom}}
\def\hom{\underline{\mathrm{hom}}}
\def\End{\mathrm{End}}

\def\Der{\mathrm{Der}}
\def\Sym{\mathrm{Sym}}

\def\id{\mathrm{id}}

\def\dd{\mathrm{d}}

\def\dim{\mathrm{dim}}
\def\1{\mathbf{1}}
\def\oone{\mathds{1}}
\def\op{\mathrm{op}}

\def\2AQFT{\mathbf{2AQFT}}

\def\Alg{\mathbf{Alg}}
\def\CAlg{\mathbf{CAlg}}

\def\DGVec{\mathbf{DGVec}}
\def\DGdgVec{\mathbf{DGdgVec}}

\def\dgCAlg{\mathbf{dgCAlg}}
\def\DGdgCAlg{\mathbf{DGdgCAlg}}
\def\DGCat{\mathbf{DGCat}}

\def\dR{\mathrm{dR}}

\def\DD{\mathbf{D}}

\def\Mod{\mathbf{Mod}}

\def\Tot{\mathrm{Tot}}

\def\AAA{\mathfrak{A}}

\def\AAAA{\boldsymbol{\mathfrak{A}}}

\def\O{\mathcal{O}}
\def\P{\mathcal{P}}

\def\T{\mathsf{T}}

\def\g{\mathfrak{g}}

\def\spec{\mathrm{Spec}\,}

\def\holim{\mathrm{holim}}
\def\colim{\mathrm{colim}}

\def\CE{\mathrm{CE}}

\def\s{\mathsf{s}}
\def\t{\mathsf{t}}
\def\E{\mathsf{E}}
\def\V{\mathsf{V}}

\def\Con{\mathrm{Con}}
\def\Gau{\mathcal{G}}

\newcommand\und[1]{\underline{#1}}
\newcommand\ovr[1]{\overline{#1}}
\newcommand\mycom[2]{\genfrac{}{}{0pt}{}{#1}{#2}}
\newcommand\an[1]{\langle #1\rangle}

\DeclareMathOperator*{\Motimes}{\text{\raisebox{0.25ex}{\scalebox{0.8}{$\bigotimes$}}}}
\DeclareMathOperator*{\Moplus}{\text{\raisebox{0.25ex}{\scalebox{0.8}{$\bigoplus$}}}}

\def\sk{\vspace{1mm}}

\makeatletter
\let\@fnsymbol\@alph
\makeatother

\title{%
Quantization of derived cotangent stacks\\
and gauge theory on directed graphs
}

\author{%
Marco Benini$^{1,2,a}$,  Jonathan P.~Pridham$^{3,b}$\ and\ 
Alexander Schenkel$^{4,c}$\vspace{4mm}\\
{\small ${}^1$ Dipartimento di Matematica, Universit\`a di Genova,}\\
{\small Via Dodecaneso 35, 16146 Genova, Italy.}\vspace{2mm}\\
{\small ${}^2$ INFN, Sezione di Genova,}\\
{\small Via Dodecaneso 33, 16146 Genova, Italy.}\vspace{2mm}\\
{\small ${}^3$ School of Mathematics,  University of Edinburgh, }\\
{\small Edinburgh EH9 3FD,  United Kingdom.}\vspace{2mm}\\
{\small ${}^4$ School of Mathematical Sciences, University of Nottingham,}\\
{\small University Park, Nottingham NG7 2RD, United Kingdom.}\vspace{4mm}\\
{\small \begin{tabular}{ll}
Email: & ${}^a$~\texttt{marco.benini@unige.it}\\
& ${}^b$~\texttt{j.pridham@ed.ac.uk}\\
& ${}^c$~\texttt{alexander.schenkel@nottingham.ac.uk}\vspace{2mm}
\end{tabular}
}
}

\date{May 2023}

\begin{document}

\maketitle

\vspace{-7mm}

\begin{abstract}
\noindent We study the quantization of the canonical unshifted Poisson structure on the derived cotangent stack $T^\ast[X/G]$ of a quotient stack, where $X$ is a smooth affine scheme with an action of a (reductive) smooth affine group scheme $G$. This is achieved through an {\'e}tale resolution of $T^\ast[X/G]$ by stacky CDGAs that allows for an explicit description of the canonical Poisson structure on $T^\ast[X/G]$ and of the dg-category of modules quantizing it. These techniques are applied to construct a dg-category-valued prefactorization algebra that quantizes a gauge theory on directed graphs.
\end{abstract}

\vspace{-3mm}

\vspace{-3mm}

\vspace{-2mm}

\renewcommand{\baselinestretch}{0.75}\normalsize
\tableofcontents
\renewcommand{\baselinestretch}{1.0}\normalsize



\section{\label{sec:intro}Introduction and summary}
Derived algebraic geometry is a modern and powerful
geometric framework which plays an increasingly important
role both in the foundations of algebraic geometry and 
in mathematical physics. It introduces a refined concept 
of `space', the so-called {\em derived stacks}, that
is capable to describe correctly geometric situations
that are problematic in traditional approaches, such as 
non-transversal intersections and quotients by non-free group actions.
We refer the reader to \cite{Toen} for an excellent survey of this subject.
\sk

In the context of mathematical physics, derived algebraic
geometry provides the foundation and a natural home for 
the various cohomological methods that have been
developed over the past decades in the context of gauge field theory, 
such as the BRST, BV and BFV formalisms. The famous ghost
fields from these approaches can naturally be identified with the 
{\em stacky} structure of a derived stack, and the anti-fields 
give rise to a {\em derived} structure. Furthermore,
the anti-bracket admits a natural geometric interpretation
in terms of a general concept of shifted symplectic or Poisson 
structure on a derived stack, see \cite{DAG,DAG2,PridhamPoisson}.
Hence, it does not come as a surprise that ideas and concepts from
derived algebraic geometry play an essential role in modern mathematical 
formulations of quantum field theory, such as in the factorization algebra
program of Costello and Gwilliam \cite{CostelloGwilliam,CostelloGwilliam2} and in 
the homotopical algebraic quantum field theory program
\cite{BSWhomotopy,BSreview,LinearYM} that is being developed by two of us.
\sk

Most of these recent applications of derived algebraic geometry to 
mathematical physics however {\em do not} use the full power of 
this framework yet as they consider only formal (or, in a physical language, perturbative)
aspects of derived stacks. For instance, in the context of gauge theories,
they consider only the Lie algebra $\g$ of infinitesimal gauge transformations instead 
of the whole gauge group $G$, which often encodes non-trivial global features that are 
not visible at the level of Lie algebras. The framework of derived algebraic geometry is, 
in principle, fully capable to capture such non-perturbative and global features, 
which however requires mathematical techniques beyond 
the quite standard homological tools that are sufficient for perturbative studies.
This implies that one of the main challenges for, say, a mathematical physicist who 
would like to apply derived geometry to their work is to specialize and translate
these highly abstract concepts and techniques into a more concrete, and computationally accessible, 
form. This has already been achieved for a (rather limited) selection of concepts and constructions:
For example, the theory of shifted symplectic and Poisson structures on derived quotient stacks 
$[Y/G]$ has been made explicit in \cite{Yeung} and a global geometric generalization 
of the BV formalism (i.e.\ derived critical loci) to functions $f : [X/G]\to \bbR$
on quotient stacks has been worked out quite explicitly in \cite{BVgroup,AnelCalaque}.
\sk

The main purpose of this paper is to expand this rather short list of explicit constructions
in derived algebraic geometry by studying the global quantization of simple examples
of derived stacks. More concretely, we shall study the derived cotangent stack
$T^\ast[X/G]$ over a quotient stack, where $X$ is a smooth affine scheme with an action
of a (reductive) smooth affine group scheme $G$, and provide an explicit description of its
quantization along the canonical unshifted Poisson structure. 
This class of examples is not only motivated by its mathematical simplicity, but it
is also physically relevant as it contains the canonical phase spaces of mechanical systems with gauge symmetries.
(In physical language, $X$ is the configuration space, $G$ is the gauge group and the contangent bundle
$T^\ast$ introduces the canonical momenta.) 
Because such derived stacks are in general not affine, i.e.\ they cannot be described by a single differential
graded (dg) algebra, we have to construct a quantization in the sense of 
\cite{ToenQuant,PridhamUnshifted} of the dg-category of perfect modules
over $T^\ast[X/G]$. (Loosely speaking, one should think of this dg-category
as a kind of category of dg-vector bundles over $T^\ast[X/G]$.) 
Existence of such kinds of unshifted quantizations has been proven
via deformation theoretic techniques by one of us in \cite[Proposition 1.25]{PridhamUnshifted}.
The aim of the present paper is complementary to this abstract existence result 
in the sense that we will develop an explicit and computationally accessible model 
for the quantized dg-category associated with the derived stack $T^\ast[X/G]$ 
and its unshifted Poisson structure. The key technical tools that we use to compute
an explicit model for this quantized dg-category are the resolution techniques of 
derived quotient stacks by so-called {\em stacky CDGAs} that have 
been developed in \cite{PridhamPoisson,PridhamUnshifted}, and in \cite{DAG2}
under the name graded mixed CDGAs.
These tools allow us to phrase the quite abstract geometric 
concepts of Poisson structures and perfect modules on derived quotient stacks 
in terms of families of much simpler algebraic concepts on stacky CDGAs.
This will eventually allow us to break the global quantization problem in two steps: 
First, a family of local quantization problems at the level of stacky CDGAs is solved by 
constructing appropriate (noncommutative) dg-algebras of differential operators. 
Second, passing to the associated local quantized dg-categories 
and computing their homotopy limit yields the global quantized dg-category 
and hence solves the global quantization problem. 
\sk

The outline of the remainder of this paper is as follows: 
Section \ref{sec:prelim} provides a self-contained and (hopefully)
accessible introduction to the theory of stacky CDGAs and their relevance 
for computing {\'e}tale resolutions of derived quotient stacks.
We shall continuously emphasize the computational aspects of
this subject and provide fully explicit algebraic formulas wherever possible.
In Section \ref{sec:cotangent}, we specialize these techniques 
to a derived cotangent stack $T^\ast[X/G]$ and study its quantization along
the canonical unshifted Poisson structure. In more detail, we show in
Subsection \ref{subsec:sympred} that $T^\ast[X/G] \simeq [T^\ast X/\!\!/ G] = [\mu^{-1}(0)/G]$ 
can be computed in terms of a derived symplectic reduction (indicated by the symbol $/\!\!/$) 
and we spell out the resulting derived quotient stack $[\mu^{-1}(0)/G]$ in full detail. In Subsection 
\ref{subsec:cotangentresolution}, we specialize the 
resolution techniques from Subsection \ref{subsec:resolution}
to our example $T^\ast[X/G] \simeq [\mu^{-1}(0)/G]$ 
and  use this to describe the canonical Poisson structure on $T^\ast[X/G]$ 
in a fully explicit way. An explicit quantization in terms of differential operators
and $D$-modules of the individual stacky CDGAs that enter this resolution is 
described in Subsection \ref{subsec:localquantization}. In Subsection \ref{subsec:globalquantization},
we determine the resulting global quantization of $T^\ast[X/G]$ by computing explicitly
a certain homotopy limit of dg-categories over our resolution. The main result
is Proposition \ref{prop:quantdgCatquotientstack}, which provides a very explicit
and hands-on description of the quantization of the dg-category of perfect modules
over $T^\ast[X/G]$ along its canonical unshifted Poisson structure.
The aim of Section \ref{sec:lattice} is to apply our constructions and 
results to a problem arising from physics, namely the quantization of a 
certain gauge theory on directed graphs. The quantization of this theory
via traditional cohomological methods (that
treat only infinitesimal gauge symmetries) has been studied previously in \cite{Pflaum}, 
and the new aspect of our work is to provide a non-perturbative
and global (albeit still formal in $\hbar$) quantization 
of the relevant phase spaces. We further arrange the resulting quantized dg-categories associated with all directed graphs 
into a prefactorization algebra over a suitable category of directed graphs.
This provides an interesting toy-model for a quantum gauge theory in which the gauge symmetries 
are treated non-perturbatively.


\section{\label{sec:prelim}Preliminaries}
We recall some relevant definitions and constructions
from the theory of stacky CDGAs that are needed for this paper.
We refer the reader to \cite{PridhamPoisson} and \cite{PridhamUnshifted}
for the details, and also to \cite{PridhamOverview} for an overview. 
Throughout the paper, we shall fix
a field $\bbK$ of characteristic $0$. 

\subsection{\label{subsec:stackyCDGA}Stacky CDGAs}
Stacky CDGAs are bigraded objects carrying both derived and stacky structures.
Following the conventions in \cite{PridhamPoisson,PridhamUnshifted,PridhamOverview}, we 
shall use chain complexes (i.e.\ homological degree conventions) 
for the derived structures and cochain complexes (i.e.\ cohomological degree conventions)
for the stacky structures.
\begin{defi}\label{def:bicomplex}
A {\em chain-cochain bicomplex} $V^\bullet_\bullet$ is a
family $\{V^i_j\}^{}_{i,j\in\bbZ}$ of $\bbK$-vector spaces
together with a chain differential $\partial : V^i_j\to V^i_{j-1}$
and a cochain differential $\delta : V^{i}_j\to V^{i+1}_j$ satisfying
\begin{flalign}
\partial^2 = 0~~,\quad\delta^2 =0~~,\quad \partial\delta + \delta\partial=0\quad.
\end{flalign}
A morphism $f:V^\bullet_\bullet \to W^\bullet_\bullet$ 
of chain-cochain bicomplexes is a family 
$\{f^i_j : V^i_j\to W^i_j\}^{}_{i,j\in\bbZ}$ of $\bbK$-linear maps
that commutes with both differentials, i.e.\ $\partial\, f = f\,\partial$ and $\delta\,f = f\,\delta$.
We denote the category of chain-cochain bicomplexes by $\DGdgVec$.
\end{defi}
\begin{rem}\label{rem:DGdgVecSMcat}
We work with {\em anti}-commuting differentials
in order to simplify the description of the closed symmetric monoidal structure on $\DGdgVec$.
The tensor product of $V^\bullet_\bullet,W^\bullet_\bullet\in \DGdgVec$ is then given by
\begin{subequations}
\begin{flalign}
\big(V^\bullet_\bullet\otimes W^\bullet_\bullet\big)^k_l \,=\, \bigoplus_{i,j\in\bbZ} V^i_j\otimes W^{k-i}_{l-j}
\end{flalign}
together with the differentials
\begin{flalign}
\partial(v\otimes w)\,&=\,\partial v\otimes w + (-1)^{\vert v\vert} \,v\otimes\partial w\quad,\\
\delta(v\otimes w)\,&=\,\delta v\otimes w + (-1)^{\vert v\vert} \,v\otimes\delta w\quad,
\end{flalign}
\end{subequations}
where $\vert v\vert := i-j$ denotes the total degree of $v\in V^i_j$.
(Note that we count cochain degrees positively and chain degrees negatively.)
The monoidal unit is $\bbK$ concentrated in bidegree $\mycom{0}{0}$
and the symmetric braiding is given by the Koszul sign-rule with respect
to the total degrees, i.e.
\begin{flalign}
V^\bullet_\bullet \otimes W^{\bullet}_\bullet~\longrightarrow~
W^\bullet_\bullet \otimes V^{\bullet}_\bullet~~,\quad
v\otimes w~\longmapsto~(-1)^{\vert v\vert\,\vert w\vert}~w\otimes v\quad.
\end{flalign}
The internal hom object for $V^\bullet_\bullet,W^\bullet_\bullet\in \DGdgVec$
is given by
\begin{subequations}
\begin{flalign}
\hom(V^\bullet_\bullet,W^\bullet_\bullet)^k_l \,=\,\prod_{i,j\in\bbZ} 
\mathrm{Lin}\big(V^i_j,W^{i+k}_{j+l}\big)\quad,
\end{flalign}
where $\mathrm{Lin}$ denotes the vector space of linear maps, together with the `adjoint' differentials
\begin{flalign}
\partial (L)\,&=\,\big\{\partial\,L^i_j - (-1)^{\vert L\vert } \, L^{i}_{j-1}\,\partial \,:\,V^i_j\to W^{i+k}_{j+l-1} \big\}^{}_{i,j\in\bbZ}\quad,\\
\delta (L)\,&=\,\big\{\delta\,L^i_j - (-1)^{\vert L\vert } \, L^{i+1}_{j}\,\delta \,:\,V^i_j\to W^{i+k+1}_{j+l} \big\}^{}_{i,j\in\bbZ}\quad,
\end{flalign}
\end{subequations}
for all $L=\big\{L^i_j \,:\,V^i_j\to W^{i+k}_{j+l} \big\}^{}_{i,j\in\bbZ}
\in \hom(V^\bullet_\bullet,W^\bullet_\bullet)^k_l$, where $\vert L\vert =k-l$ denotes the total degree.
\end{rem}

The category $\DGdgVec$ of chain-cochain bicomplexes can be endowed with 
a cofibrantly generated model structure
in which the fibrations are bidegree-wise surjections and
the weak equivalences are level-wise quasi-isomorphisms,
i.e.\ morphisms $f : V^\bullet_\bullet \to W^\bullet_\bullet$ 
such that $f^i : V^i_\bullet \to W^i_\bullet$ is a quasi-isomorphism
of chain complexes, for each $i\in \bbZ$. See \cite[Lemma 3.4]{PridhamPoisson}.
Such weak equivalences are, in general, neither preserved by the sum-totalization
$\Tot^\oplus$ nor by the product-totalization $\Tot^{\Pi}$ functor.
However, they are preserved by the following construction.
\begin{defi}\label{def:totalization}
The semi-infinite totalization (or Tate realization) functor $\hat{\Tot} : \DGdgVec \to\DGVec$ assigns
to a chain-cochain bicomplex $V^\bullet_\bullet\in\DGdgVec$ the 
cochain sub-complex $\hat{\Tot}\,V^\bullet_\bullet \subseteq \Tot^\Pi\, V^\bullet_\bullet$
given by
\begin{flalign}\label{eqn:totalization}
(\hat{\Tot} \, V^\bullet_\bullet)^m\,:=\,\Big( \bigoplus_{i<0} V^i_{i-m}\Big)\oplus\Big(\prod_{i\geq 0} V^i_{i-m}\Big)\,\subseteq ~\prod_{i\in\bbZ} V^i_{i-m}
\end{flalign} 
and the total differential $\hat{\dd} := \partial + \delta$. 
\end{defi}
\begin{rem}
By definition, an element of $(\hat{\Tot} \, V^\bullet_\bullet)^m$ is given by a family
of elements $v=\{v^i\in V^i_{i-m}\}_{i\in\bbZ}^{}$ that is bounded from below, i.e.\
there exists $N\in\bbZ$ (depending on the element $v$) such that $v^i =0$ for all $i<N$.
This fact can be used to endow $\hat{\Tot} : \DGdgVec\to\DGVec$ 
with the structure of a Lax symmetric monoidal functor. More explicitly, 
the coherence maps
\begin{subequations}
\begin{flalign}
\big(\hat{\Tot}\,V^\bullet_\bullet\big) \,\otimes \,\big(\hat{\Tot}\,W^\bullet_\bullet\big)~\longrightarrow~
\hat{\Tot}\big(V^\bullet_\bullet \otimes W^\bullet_\bullet\big)
\end{flalign}
are given on homogeneous elements $v\in (\hat{\Tot} \, V^\bullet_\bullet)^m$ 
and  $w\in (\hat{\Tot} \, W^\bullet_\bullet)^n$ by
\begin{flalign}
\big\{v^i\in V^i_{i-m}\big\}_{i\in\bbZ}^{}\otimes \big\{  w^j \in W^j_{j-n}\big\}_{j\in\bbZ}^{}~\longmapsto~
\Big\{\sum_{i\in\bbZ} v^i\otimes w^{s-i}\Big\}_{s\in\bbZ}^{}\quad,
\end{flalign}
\end{subequations}
which is well-defined because both families $\{v^i\}$ and $\{w^j\}$ are bounded from below.
\end{rem}

\begin{defi}\label{def:stackyCDGA}
The category of {\em stacky CDGAs}
is defined as the category $\DGdgCAlg := \CAlg(\DGdgVec)$ of commutative and unital algebras 
in the symmetric monoidal category $\DGdgVec$. Explicitly,
a stacky CDGA consists of a chain-cochain bicomplex $A^\bullet_\bullet\in\DGdgVec$
together with two $\DGdgVec$-morphisms $\mu : A^\bullet_\bullet \otimes A^\bullet_\bullet\to A^\bullet_\bullet$
(called multiplication) and $\eta : \bbK\to A^\bullet_\bullet$ (called unit)
that satisfy the associativity, unitality and commutativity axioms.
A morphism of stacky CDGAs is a $\DGdgVec$-morphism $f : A^\bullet_\bullet\to B^\bullet_\bullet$
that preserves the multiplications and units.
\end{defi}
\begin{rem}
It is shown in \cite[Lemma 3.4]{PridhamPoisson} that the category
of stacky CDGAs can be endowed with a cofibrantly generated model
structure in which the fibrations are bidegree-wise surjections
and the weak equivalences are level-wise quasi-isomorphisms.
This model structure is obtained via transfer of the level-wise model structure on $\DGdgVec$
that we have mentioned above.
\end{rem}

\begin{ex}\label{ex:CEalgebra}
Let $B_\bullet\in \mathbf{dgCAlg}_{\geq 0}$ be a non-negatively 
graded chain CDGA together with an action $\rho : \g \to \Der(B_\bullet)$
of a finite-dimensional Lie algebra $\g$ in terms of derivations, i.e.\
$\rho(t)(b\,b^\prime) = \rho(t)(b)\,b^\prime + b\,\rho(t)(b^\prime)$
for all $t\in\g$ and $b,b^\prime\in B_\bullet$.
Taking Chevalley-Eilenberg cochains on $\g$ with coefficients in
$B_\bullet$ defines a stacky CDGA
\begin{flalign}
\CE^\bullet(\g,B_\bullet)\,\in\,\DGdgCAlg\quad.
\end{flalign}
In the context of derived algebraic geometry, one interprets
this stacky CDGA as (the function algebra of) the formal derived quotient stack
$[Y/\g]$ of the derived affine scheme $Y = \spec \,B_\bullet$ by the Lie algebra action,
i.e.\ $\O([Y/\g]) = \CE^\bullet(\g,B_\bullet)$.
\sk

Since this example will play a prominent role in our paper,
let us also spell out the stacky CDGA $\CE^\bullet(\g,B_\bullet)$ fully explicitly.
For convenience, let us choose a basis $\{t_a\in\g\}$ of the Lie algebra and 
denote by $\{\theta^a\in\g^\vee\}$ its dual basis.
The underlying bigraded commutative algebra (obtained by forgetting both differentials)
is given by the tensor product
\begin{flalign}
\CE^\sharp(\g,B_{\sharp})\,=\, \big(\Sym\, \g^{\vee[-1]}\big)^\sharp \otimes B_{\sharp}\quad,
\end{flalign}
where the degree shifting $^{[-1]}$ means that the non-trivial elements in $\g^{\vee[-1]}$ 
are of cochain degree $+1$. This bigraded algebra is generated by elements of the form
$b := \oone\otimes b\in \CE^0(\g,B_{\sharp})$ and $\theta^a:=\theta^a\otimes\oone\in \CE^1(\g,B_0)$.
With respect to our sign conventions, graded commutativity of the multiplication reads 
explicitly on such elements as
\begin{flalign}
b\,b^\prime \,=\, (-1)^{\vert b\vert \,\vert b^\prime\vert} ~b^\prime\,b~~,\quad
\theta^a\,b \,=\, (-1)^{\vert b\vert}\, b\,\theta^a~~,\quad
\theta^a\,\theta^b \,=\, -\theta^b\,\theta^a\quad.
\end{flalign}
The chain differential $\partial$ on $\CE^\bullet(\g,B_\bullet)$ is defined 
on the generators by
\begin{flalign}
\partial(b)\,=\,\partial_B^{}(b)~~,\quad
\partial(\theta^a)\, =\, 0\quad,
\end{flalign}
where $\partial_B^{}$ denotes the differential on the chain CDGA $B_\bullet$.
The cochain (i.e.\ Chevalley-Eilenberg) differential $\delta$ on $\CE^\bullet(\g,B_\bullet)$
is defined on the generators by
\begin{flalign}
\delta(b)\,=\, \theta^a\,\rho(t_a)(b)~~,\quad
\delta(\theta^a)\,=\, -\tfrac{1}{2} f^a_{bc}\,\theta^b\,\theta^c\quad,
\end{flalign}
where $[t_a,t_b] = f^{c}_{ab}\,t_c$ denote the structure constants of the Lie algebra
and summations over repeated indices are understood.
The two differentials $\partial$ and $\delta$ are 
extended to all of $\CE^\sharp(\g,B_{\sharp})$ via the graded Leibniz rule with respect to the total degree.
\end{ex}

\begin{rem}\label{rem:CEalgebrafunctoriality}
For later use, we would like to record the following functorial behavior 
of the Chevalley-Eilenberg stacky CDGAs. Let $\g$ and $\g^\prime$ be two 
finite-dimensional Lie algebras. Let further $B_\bullet \in\dgCAlg_{\geq 0}$ be a chain CDGA with a
$\g$-action $\rho : \g \to \Der(B_\bullet)$ and $B^\prime_\bullet \in\dgCAlg_{\geq 0}$
be a chain CDGA with a $\g^\prime$-action $\rho^{\prime} : \g^\prime \to \Der(B^\prime_\bullet)$.
Denote by $\CE^\bullet(\g,B_\bullet)$ and $\CE^{\bullet}(\g^\prime,B_\bullet^\prime)$
the associated Chevalley-Eilenberg stacky CDGAs from Example \ref{ex:CEalgebra}. 
\sk

Suppose that we are given a Lie algebra morphism $\tilde{\kappa} : \g^\prime\to \g$
and a CDGA morphism $\kappa^\ast : B_\bullet\to B^\prime_\bullet$ that is equivariant relative
to $\tilde{\kappa}$, i.e.\ $\rho^\prime(t^\prime)\big(\kappa^\ast(b)\big) 
= \kappa^\ast\big(\rho(\tilde{\kappa}(t^\prime))(b)\big)$
for all $b\in B_\bullet$ and $t^\prime\in \g^\prime$. Denoting by 
$\tilde{\kappa}^\ast : \g^\vee \to \g^{\prime\vee}$ the dual of the 
Lie algebra morphism $\tilde{\kappa}$, we can define a morphism
\begin{subequations}
\begin{flalign}
(\kappa^\ast,\tilde{\kappa}^\ast) \,:\, \CE^\bullet(\g,B_\bullet)~\longrightarrow~\CE^\bullet(\g^\prime,B_\bullet^\prime)
\end{flalign}
of stacky CDGAs by setting for the generators
\begin{flalign}
(\kappa^\ast,\tilde{\kappa}^\ast)(b)\,:=\,\kappa^\ast(b)~~,\quad
(\kappa^\ast,\tilde{\kappa}^\ast)(\theta) \,:=\,\tilde{\kappa}^\ast(\theta)\quad.
\end{flalign}
\end{subequations}
One easily checks that this is compatible with the chain and the cochain differentials.
\end{rem}

Next, we introduce a suitable 
concept of unshifted Poisson structures on stacky CDGAs. 
While \cite{PridhamPoisson} and \cite{PridhamUnshifted}
introduces {\em weak} Poisson structures, which loosely speaking are
Poisson brackets satisfying the Jacobi identity up to coherent homotopy,
it is sufficient for our paper to restrict ourselves to {\em strict} Poisson structures.
\begin{defi}\label{def:Poisson}
An unshifted {\em strict Poisson structure} on a stacky CDGA $A^\bullet_\bullet$
is an unshifted Poisson bracket on its totalization $\hat \Tot\, A^\bullet_\bullet$.
More explicitly, the latter is a morphism
\begin{flalign}
\{\,\cdot\,,\,\cdot\,\} \,:\,\hat \Tot\, A^\bullet_\bullet \otimes \hat \Tot\, A^\bullet_\bullet~\longrightarrow ~\hat \Tot\, A^\bullet_\bullet
\end{flalign}
of cochain complexes that satisfies the following properties:
\begin{itemize}
\item[(i)] {\em Antisymmetry:} For all homogeneous $a,a^\prime\in \hat{\Tot}\,A^\bullet_\bullet$,
\begin{flalign}
\{a,a^\prime\} \,=\, - (-1)^{\vert a\vert\,\vert a^\prime\vert}~\{a^\prime,a\}\quad.
\end{flalign}

\item[(ii)] {\em Derivation property:} For all homogeneous $a,a^\prime,a^{\prime\prime}\in 
\hat{\Tot}\,A^\bullet_\bullet$,
\begin{flalign}
\{a,a^\prime\,a^{\prime\prime}\}\,=\, \{a,a^\prime\} \,a^{\prime\prime} + (-1)^{\vert a\vert \,\vert a^\prime\vert}~
a^\prime\,\{a,a^{\prime\prime}\}\quad.
\end{flalign}

\item[(iii)] {\em Jacobi identity:} For all homogeneous $a,a^\prime,a^{\prime\prime}\in 
\hat{\Tot}\,A^\bullet_\bullet$,
\begin{flalign}
\resizebox{0.85\hsize}{!}{$
(-1)^{\vert a\vert \, \vert a^{\prime\prime}\vert}~\big\{a,\{a^\prime,a^{\prime\prime}\} \big\} + 
(-1)^{\vert a^\prime\vert \, \vert a\vert}~\big\{a^\prime,\{a^{\prime\prime},a\} \big\} + 
(-1)^{\vert a^{\prime\prime}\vert \, \vert a^\prime\vert}~\big\{a^{\prime\prime},\{a,a^{\prime}\} \big\} \,=\,0\quad.$}
\end{flalign}
\end{itemize}
\end{defi}

To every stacky CDGA $A^\bullet_\bullet\in\DGdgCAlg$
one can assign its category ${}_{A^\bullet_\bullet}\Mod$ of (say, left)
modules in the symmetric monoidal category $\DGdgVec$. An object in this category
consists of a chain-cochain bicomplex $M^\bullet_\bullet\in \DGdgVec$
together with a $\DGdgVec$-morphism $\ell : A^\bullet_\bullet \otimes M^\bullet_\bullet\to M^\bullet_\bullet$
(called left action) that satisfies the left module axioms with respect to the algebra structure on $A^\bullet_\bullet$.
The morphisms are $\DGdgVec$-morphisms $f: M^\bullet_\bullet\to N^\bullet_\bullet$ 
that preserve the left actions. This category admits
an enrichment over $\DGdgVec$ by introducing, for each pair of objects
$M^\bullet_\bullet, N^\bullet_\bullet \in {}_{A^\bullet_\bullet}\Mod$, the hom object
\begin{subequations}
\begin{flalign}
\hom_{A^\bullet_\bullet}^{}\big(M^\bullet_\bullet,N^\bullet_\bullet\big)\,\subseteq\,
\hom\big(M^\bullet_\bullet,N^\bullet_\bullet\big)\,\in\,\DGdgVec
\end{flalign}
given by the chain-cochain sub-bicomplex of all internal homs (see Remark \ref{rem:DGdgVecSMcat})
preserving the left actions. More explicitly, an element 
$L=\big\{L^i_j \,:\,M^i_j\to N^{i+k}_{j+l} \big\}^{}_{i,j\in\bbZ}
\in \hom(M^\bullet_\bullet,N^\bullet_\bullet)^k_l$ lies in this sub-bicomplex
if and only if, for all $m\in M^i_j$ and $a\in A^p_q$,
\begin{flalign}
L^{i+p}_{j+q}(a\cdot m) = (-1)^{\vert a\vert\, \vert L\vert }~a\cdot L^i_j(m)\quad,
\end{flalign}
\end{subequations}
where we abbreviated both left actions by a dot. 
(Recall that the total degrees are given by $\vert a\vert = p-q$ and $\vert L\vert =k-l$.)
The composition of morphisms in this enriched category is the one induced from the canonical
composition of internal homs.
The following definition is \cite[Definition 1.24]{PridhamUnshifted}.
\begin{defi}\label{def:perf}
The {\em dg-category of perfect modules} $\mathrm{per}(A^\bullet_\bullet)$ 
over a stacky CDGA $A^\bullet_\bullet$ is defined as follows: 
Its objects are all $M^\bullet_\bullet\in {}_{A^\bullet_\bullet}\Mod$
satisfying:
\begin{itemize}
\item[(i)] $M^\sharp_\bullet$ is cofibrant as a graded dg-module over the graded dg-algebra
$A^\sharp_\bullet$;

\item[(ii)] $M_\bullet^0$ is a perfect dg-module over the dg-algebra $A^0_\bullet$;

\item[(iii)] $A^\sharp_\bullet\otimes_{A^0_\bullet} M_\bullet^0 \longrightarrow M^\sharp_\bullet$ is a level-wise
quasi-isomorphism.
\end{itemize}
For each pair of objects, the cochain complex of morphisms is defined by
\begin{flalign}
\hat{\hom}_{A^\bullet_\bullet}^{}\big(M^\bullet_\bullet,N^\bullet_\bullet\big)\,:=\,\hat{\Tot}\, \Big(\hom_{A^\bullet_\bullet}^{}\big(M^\bullet_\bullet,N^\bullet_\bullet\big)\Big)\,\in\,\DGVec\quad.
\end{flalign}
The composition of morphisms is induced from the compositions on $\hom_{A^\bullet_\bullet}^{}$
and Lax monoidality of the functor $\hat{\Tot}$.
\end{defi}

\begin{ex}\label{ex:CEmodule}
Consider the stacky CDGA $ \CE^\bullet(\g,B_\bullet)$ from Example \ref{ex:CEalgebra}.
Given any left $B_\bullet$-dg-module $V_\bullet\in {}_{B_\bullet}\Mod$ with a compatible action
$\rho_V^{} : \g\to \End(V_\bullet)$, i.e.\ $\rho_V^{}(t)(b\cdot s) = \rho(t)(b)\cdot s + b\cdot \rho_V^{}(t)(s)$
for all $t\in\g$, $b\in B_\bullet$ and $s\in V_\bullet$, we can take the Chevalley-Eilenberg cochains
\begin{flalign}
\CE^\bullet(\g,V_\bullet)\,\in\,{}_{\CE^\bullet(\g,B_\bullet)}\Mod
\end{flalign}
and thereby obtain a left $\CE^\bullet(\g,B_\bullet)$-module. Similarly to Example \ref{ex:CEalgebra},
one can spell this out fully explicitly. The underlying bigraded module is given by
\begin{flalign}
\CE^\sharp(\g,V_{\sharp})\,=\, \big(\Sym\, \g^{\vee[-1]}\big)^\sharp \otimes V_{\sharp}
\end{flalign}
and the $\CE^\sharp(\g,B_\sharp)$-action
\begin{flalign}
(\xi\otimes b)\cdot (\xi^\prime\otimes s)\,=\, (-1)^{-i^\prime\,j}~(\xi\,\xi^\prime)\otimes (b\cdot s)\quad,
\end{flalign}
for all $\xi\otimes b\in \CE^i(\g,B_j)$ and $\xi^\prime\otimes s\in \CE^{i^\prime}(\g,V_{j^\prime})$.
Note that this module is free over $\big(\Sym\, \g^{\vee[-1]}\big)^\sharp$ with generators
concentrated in {\em co}chain degree $0$, hence it satisfies item (iii) of Definition \ref{def:perf}.
The chain differential is defined on elements $\xi\otimes s\in \CE^i(\g,V_j)$ by
\begin{flalign}
\partial(\xi\otimes s)\,=\,(-1)^{i}\,\xi\otimes \partial_V^{} (s)\quad,
\end{flalign}
where $\partial_V^{}$ denotes the differential on $V_\bullet$,
and the cochain (i.e.\ Chevalley-Eilenberg) differential
is defined on the generators by
\begin{flalign}
\delta(\oone\otimes s)\,=\,\theta^a\otimes \rho_V^{}(t_a)(s)\quad,
\end{flalign}
for all $s\in V_j$. One easily checks that item (i) of Definition \ref{def:perf}
holds true if and only if $V_\bullet\in {}_{B_\bullet}\Mod$ is a
cofibrant left $B_\bullet$-dg-module and item (ii) holds true if and only if
$V_\bullet\in {}_{B_\bullet}\Mod$ is perfect. The cofibrant dg-modules
can be characterized explicitly in terms of a semi-projectivity condition,
see e.g.\ \cite[Theorem 9.10 and Definition 9.1]{Riehl}. In particular,
$V_\bullet\in {}_{B_\bullet}\Mod$ is both cofibrant and perfect if
its underlying graded $B_\sharp$-module $V_\sharp$ is finitely generated and projective.
\sk

One can further check that, up to weak equivalence given by level-wise quasi-isomorphisms, 
every $\CE^\bullet(\g,B_\bullet)$-module
$M^\bullet_\bullet\in {}_{\CE^\bullet(\g,B_\bullet)}\Mod$ satisfying the conditions 
in Definition \ref{def:perf} is of this form. Indeed, item (iii) implies
that the underlying $\CE^\sharp(\g,B_\bullet)$-module $M^\sharp_\bullet$
(obtained by forgetting the cochain differential) is equivalent to
\begin{flalign}
\CE^\sharp(\g,B_\bullet)\otimes_{B_\bullet} M^0_\bullet\,\cong\,\big(\big(\Sym\, \g^{\vee[-1]}\big)^\sharp \otimes B_{\bullet}\big)\otimes_{B_\bullet} M_\bullet^0\,\cong\, \big(\Sym\, \g^{\vee[-1]}\big)^\sharp \otimes M_\bullet^0\quad.
\end{flalign}
The cochain differential $\delta$ is fixed by its action on the generators $\oone\otimes m$, 
with $m\in M^0_j$, which due to degree reasons has to be of the form 
$\delta(\oone\otimes m) =\theta^a\otimes \rho_M^{}(t_a)(m)$.
The square-zero condition $\delta^2=0$ is then equivalent to $\rho_M^{}$ being a Lie algebra action.
\end{ex}

Let us summarize and expand on the observations from Example \ref{ex:CEmodule}.
\begin{propo}\label{prop:CEmodulescharacterization}
The dg-category $\mathrm{per}\big(\CE^\bullet(\g,B_\bullet)\big)$ 
of perfect modules over the stacky CDGA $\CE^\bullet(\g,B_\bullet)$ 
is equivalent to the dg-category
whose objects are all cofibrant and perfect left $B_\bullet$-dg-modules $V_\bullet\in {}_{B_\bullet}\Mod$ 
that are endowed with a compatible Lie algebra action $\rho_V^{} : \g\to \End(V_\bullet)$ and 
whose cochain complexes of morphisms are given by
\begin{flalign}
\nn \hat{\hom}_{\CE^\bullet(\g,B_\bullet)}^{}\big(\CE^\bullet(\g,V_\bullet),\CE^\bullet(\g,W_\bullet)\big)\,&\cong\, \hom_{\hat{\Tot}\CE^\bullet(\g,B_\bullet)}^{}\Big(\hat{\Tot}\CE^\bullet(\g,V_\bullet),\hat{\Tot}\CE^\bullet(\g,W_\bullet)\Big)\\
\,&\cong\,\hat{\Tot}\CE^\bullet\big(\g,\hom_{B_\bullet}^{}(V_\bullet,W_\bullet)\big)\quad,\label{eqn:perfhomidentification}
\end{flalign}
for each pair of objects $V_\bullet, W_\bullet$.
\end{propo}
\begin{proof}
The characterization of the objects has been discussed in Example \ref{ex:CEmodule}.
The second statement that the totalized $\hat{\hom}$ can be identified with
the $\hom$ between the totalized objects follows easily from the fact that 
$\CE^\bullet(\g,B_\bullet)$, $\CE^\bullet(\g,V_\bullet)$ and $\CE^\bullet(\g,W_\bullet)$
are bounded from below and, because $\g$ is finite-dimensional, also from above
in the cochain degrees. Each element in 
$\hat{\hom}_{\CE^\bullet(\g,B_\bullet)}^{}\big(\CE^\bullet(\g,V_\bullet),\CE^\bullet(\g,W_\bullet)\big)^p$ 
can then be identified with a finite direct sum 
\begin{flalign}\label{eqn:tmpcomponentmaps}
\bigoplus_{i=0}^{\mathrm{dim}\g} \Big\{(L^i)_j \,:\, V_j \to \CE^i(\g,W_{j+i-p})\Big\}^{}_{j\in\bbZ}\quad,
\end{flalign}
which under \eqref{eqn:perfhomidentification} is mapped to
the element in $\hom_{\hat{\Tot}\CE^\bullet(\g,B_\bullet)}^{}\big(\hat{\Tot}\CE^\bullet(\g,V_\bullet),\hat{\Tot}\CE^\bullet(\g,W_\bullet)\big)^p$ that is specified by
\begin{flalign}
\Big\{\sum_{i=0}^{\mathrm{dim}\g} (L^i)_j\,:\,V_j \to \hat{\Tot}\CE^\bullet(\g,W_{\bullet})^{p-j}\Big\}_{j\in\bbZ}^{}\quad.
\end{flalign}
One easily checks that this assignment is an isomorphism of cochain complexes.
For the isomorphism in the second line of \eqref{eqn:perfhomidentification},
we use once more that $\g$ is finite-dimensional in order to interpret 
\eqref{eqn:tmpcomponentmaps} in terms of elements $\{(L^i)_j\}_{j\in\bbZ}\in 
\CE^i\big(\g,\hom_{B_\bullet}^{}(V_\bullet,W_\bullet)_{i-p}\big)$.
\end{proof}

\begin{rem}\label{rem:CEmodulescharacterization}
In other words, the above proposition states that the dg-category
of perfect modules
\begin{flalign}
\mathrm{per}\big(\CE^\bullet(\g,B_\bullet)\big)\,\subseteq\,{}_{\hat{\Tot}\CE^\bullet(\g,B_\bullet)}\Mod
\end{flalign}
over the stacky CDGA $\CE^\bullet(\g,B_\bullet)$ may be described as the full sub-dg-category of the 
dg-category of dg-modules over the totalized dg-algebra $\hat{\Tot}\CE^\bullet(\g,B_\bullet)$ on the objects
$\hat{\Tot}\CE^\bullet(\g,V_\bullet)$, for all cofibrant and perfect $V_\bullet\in {}_{B_\bullet}\Mod$
that are endowed with a compatible Lie algebra action $\rho_V^{} : \g\to \End(V_\bullet)$.
This observation will be useful later when we discuss quantizations.
We would like to stress that the totalized dg-module $\hat{\Tot}\CE^\bullet(\g,V_\bullet)$
is usually {\em not} cofibrant over $\hat{\Tot}\CE^\bullet(\g,B_\bullet)$, hence
our dg-category $\mathrm{per}\big(\CE^\bullet(\g,B_\bullet)\big)$ is {\em not}
the usual derived dg-category of $\hat{\Tot}\CE^\bullet(\g,B_\bullet)$-dg-modules.
\end{rem}

\begin{rem}\label{rem:perfunctoriality}
We would like to add a comment about the functoriality of 
the dg-categories of perfect modules. Consider a stacky CDGA morphism 
$(\kappa^\ast,\tilde{\kappa}^\ast) : \CE^\bullet(\g,B_\bullet)\to \CE^{\bullet}(\g^\prime,B_{\bullet}^\prime)$
as described in Remark \ref{rem:CEalgebrafunctoriality}. Then the corresponding induced module functor
\begin{flalign}\label{eqn:dgfunctorCE}
(\kappa^\ast,\tilde{\kappa}^\ast)_!^{}\,:=\,  \CE^{\bullet}(\g^\prime,B_{\bullet}^\prime)\otimes_{\CE^\bullet(\g,B_\bullet)}^{} (-) \,:\, \mathrm{per}\big(\CE^\bullet(\g,B_\bullet)\big)~\longrightarrow~
\mathrm{per}\big(\CE^\bullet(\g^\prime,B_\bullet^\prime)\big)
\end{flalign}
defines a dg-functor between the dg-categories of perfect modules.
Using the explicit description from Proposition 
\ref{prop:CEmodulescharacterization} and Example \ref{ex:CEmodule},
we can spell out the action of this dg-functor on objects very concretely:
Let $V_\bullet \in {}_{B_\bullet}\Mod$ be a cofibrant and perfect
left $B_\bullet$-dg-module with a compatible Lie algebra action $\rho_V^{}: \g\to\End(V_\bullet)$.
Forgetting the cochain differentials, we find at the level of graded dg-modules over 
$\CE^{\sharp}(\g^\prime,B_{\bullet}^\prime)$ that
\begin{flalign}
\CE^{\sharp}(\g^\prime,B_{\bullet}^\prime)\otimes_{\CE^\sharp(\g,B_\bullet)}^{} \CE^\sharp(\g,V_\bullet)~\cong\,
\big(\Sym\,\g^{\prime\vee[-1]}\big)^\sharp \otimes \big(B_{\bullet}^\prime\otimes_{B_\bullet}^{} V_\bullet\big)\quad,
\end{flalign}
where $B_{\bullet}^\prime\otimes_{B_\bullet}^{} V_\bullet\in {}_{B_\bullet^\prime}\Mod$
is the induced module along $\kappa^\ast : B_\bullet\to B_\bullet^\prime$. (This is 
cofibrant and perfect over $B_\bullet^\prime$.) Chasing through this isomorphism, 
one obtains that the cochain differential $\delta^\prime$ acts on the generators as
\begin{flalign}
\delta^\prime\big(\oone \otimes b^\prime\otimes_{B_\bullet } s \big)\,=\,
\theta^{\prime a}\otimes  \Big(\rho^\prime(t^\prime_a)(b^\prime)\otimes_{B_\bullet } s + 
b^\prime\otimes_{B_\bullet} \rho_V^{}(\tilde{\kappa}(t^\prime_a)) (s)\Big)\quad.
\end{flalign}
Comparing with the general construction in Example \ref{ex:CEmodule},
we observe that this is the cochain differential associated
with the $\g^\prime$-action $\rho_{V}^\prime :
\g^\prime\to \End\big(B_{\bullet}^\prime\otimes_{B_\bullet}^{} V_\bullet\big)$ defined
by $\rho^\prime_{V}(t^\prime)(b^\prime\otimes_{B_\bullet} s) =
\rho^\prime(t^\prime)(b^\prime)\otimes_{B_\bullet} s + 
b^\prime\otimes_{B_\bullet} \rho_V^{}(\tilde{\kappa}(t^\prime))(s)$.
Summing up, this means that the dg-functor \eqref{eqn:dgfunctorCE} 
acts on objects as $(V_\bullet,\rho_V^{})\mapsto 
\big(B_{\bullet}^\prime\otimes_{B_\bullet}^{} V_\bullet,\rho_V^\prime\big)$.
Note that, as typical for induced module functors, the assignment $(-)_!$
from stacky CDGA morphisms to dg-functors is not strictly functorial but 
only pseudo-functorial with respect to the obvious coherences for compositions and identities.
\end{rem}


\subsection{\label{subsec:resolution}Resolutions of derived quotient stacks}
Let $Y=\spec\, B_\bullet$ be a derived affine scheme, i.e.\ $B_\bullet\in\dgCAlg_{\geq 0}$
is a non-negatively graded chain CDGA, together with an action $Y\times G\to Y$ 
of a smooth affine group scheme $G = \spec\,H$, i.e.\ $H$ is a smooth commutative Hopf algebra.
The derived quotient stack $[Y/G]$ is defined as the colimit
\begin{flalign}\label{eqn:quotientstack}
[Y/G] \,:=\, \colim\Big(
\xymatrix@C=1em{
Y ~&~ \ar@<0.5ex>[l] \ar@<-0.5ex>[l] Y\times G~&~ \ar@<1ex>[l] \ar[l] \ar@<-1ex>[l] Y\times G^2 ~&~ \ar@<0.5ex>[l] \ar@<-0.5ex>[l]\ar@<1.5ex>[l] \ar@<-1.5ex>[l]\cdots
}\Big)
\end{flalign}
in the $\infty$-category of derived stacks
of the simplicial diagram associated with the $G$-action $Y\times G \to Y$ on $Y$. 
In contrast to the formal derived quotient stack $[Y/\g]$ in Example \ref{ex:CEalgebra},
the derived quotient stack $[Y/G]$ for a group action is in general {\em not} affine,
which means that it can {\em not} be described algebraically in terms of a CDGA or a stacky CDGA.
We will now recall, and then spell out fully explicitly, the construction in \cite[Example 3.6]{PridhamPoisson}
that provides an {\'e}tale resolution of a derived quotient stack $[Y/G]$ by stacky CDGAs.
This resolution allows us to apply the definitions and results for
stacky CDGAs from Subsection \ref{subsec:stackyCDGA} to derived quotient stacks,
and it eventually allows for the study of Poisson structures and quantizations of $[Y/G]$, 
see also \cite{PridhamUnshifted}.
\sk

Before spelling out the algebraic details in a self-contained way, 
let us provide first an informal overview of the resolution constructed in
\cite[Example 3.6]{PridhamPoisson}.
The basic idea is to consider the simplicial diagram
\begin{flalign}\label{eqn:quotientstackresolution}
\xymatrix@C=1em{
[Y/\g] ~&~ \ar@<0.5ex>[l] \ar@<-0.5ex>[l] \big[Y\times G/\g^{\oplus 2}\big]~&~ \ar@<1ex>[l] \ar[l] \ar@<-1ex>[l] 
\big[Y\times G^2/\g^{\oplus 3}\big] ~&~ \ar@<0.5ex>[l] \ar@<-0.5ex>[l]\ar@<1.5ex>[l] \ar@<-1.5ex>[l]\cdots
}
\end{flalign}
of formal derived quotient stacks instead of the simplicial diagram in 
\eqref{eqn:quotientstack}. One should think of this diagram as follows:
As a first order approximation, one approximates $[Y/G]$ by the formal derived quotient stack $[Y/\g]$,
and then one corrects for the difference by introducing the higher simplicial degrees.
The Lie algebra actions in \eqref{eqn:quotientstackresolution} are obtained from the group actions
\begin{flalign}\label{eqn:degreewisegroupaction}
\nn r\,:\,(Y\times G^n) \times G^{n+1}~&\longrightarrow~ Y\times G^n~~,\quad\\
\big((y,g_1,\dots, g_n),(g^\prime_0,\dots,g^\prime_n)\big)~&\longmapsto~
\big(y\,g^\prime_0 , g_0^{\prime -1} \,g_1\,g^\prime_1,\dots, g_{n-1}^{\prime -1}\,g_n\,g^\prime_n\big)\quad,
\end{flalign}
where the $G$-action on $Y$ is abbreviated by $y\,g^\prime_0$.
Before taking formal quotients, the face maps are defined by
\begin{flalign}\label{eqn:facemap}
\nn d_i\,:\,Y\times G^n~&\longrightarrow~ Y\times G^{n-1}~~,\quad\\
(y,g_1,\dots, g_n)~&\longmapsto~
\begin{cases}
\big(y\,g_1,g_2,\dots, g_n\big)&~~\text{for }i=0\quad,\\
\big(y,g_1,\dots,g_i\,g_{i+1},\dots, g_n\big)&~~\text{for } i=1,\dots,n-1\quad,\\
\big(y,g_1,\dots,g_{n-1}\big)&~~\text{for } i=n\quad.
\end{cases}
\end{flalign}
and the degeneracy maps are defined by
\begin{flalign}\label{eqn:degeneracymap}
s_i\,:\, Y\times G^n~\longrightarrow~ Y\times G^{n+1}~~,\quad
(y,g_1,\dots, g_n)~\longmapsto~\big(y,g_1, \dots,
g_i,e,g_{i+1},\dots,g_n\big)\quad,
\end{flalign}
for $i=0,\dots,n$, where $e\in G$ denotes the identity element.
The face and degeneracy maps are equivariant under the group actions
\eqref{eqn:degreewisegroupaction} relative to the group homomorphisms
\begin{flalign}\label{eqn:facemaptilde}
\tilde{d}_i\,:\,G^{n+1}~\longrightarrow~ G^n~~,\quad
(g_0,\dots,g_n)~\longmapsto~\big(g_0,\dots, \check{g_i}, \dots,g_n\big)\quad,
\end{flalign}
where $\check{-}$ means to omit the corresponding factor, and
\begin{flalign}\label{eqn:degeneracymaptilde}
\tilde{s}_i\,:\, G^{n+1}~\longrightarrow~G^{n+2}~~,\quad
(g_0,\dots,g_n)~\longmapsto~\big(g_0,\dots,g_i,g_i,\dots,g_n\big)\quad.
\end{flalign}
The latter means that the diagrams 
\begin{flalign}\label{eqn:compatibilitydiagram}
\xymatrix@C=1em{
\ar[d]_-{d_i\times\tilde{d}_i} 
(Y\times G^n) \times G^{n+1} \ar[r]^-{r}~&~ Y\times G^n\ar[d]^-{d_i} ~&~~&~
\ar[d]_-{s_i\times\tilde{s}_i}(Y\times G^n) \times G^{n+1} \ar[r]^-{r}~&~ Y\times G^n\ar[d]^-{s_i} \\
(Y\times G^{n-1}) \times G^{n}\ar[r]_-{r}~&~Y\times G^{n-1} ~&~~&~ 
(Y\times G^{n+1}) \times G^{n+2}\ar[r]_-{r}~&~Y\times G^{n+1}
}
\end{flalign}
commute. In particular, this implies that the face and degeneracy maps induce 
to the formal derived quotient stacks in \eqref{eqn:quotientstackresolution}.
\sk

Let us now rewrite the construction above in the dual algebraic language 
in order to connect to the theory of stacky CDGAs. The derived affine scheme
$Y=\spec\, B_\bullet$ is described by its function CDGA $\O(Y)=B_\bullet\in\dgCAlg_{\geq 0}$
and the smooth affine group scheme $G = \spec\,H$ by its (smooth)
function Hopf algebra $\O(G)=H$. We denote the comultiplication by $\Delta : H\to H\otimes H$,
the counit by $\varepsilon : H\to\bbK$ and the antipode by $S : H\to H$. 
The $G$-action $Y\times G\to Y$ on $Y$ is algebraically described by an $H$-comodule
structure $\rho : B_\bullet\to B_\bullet\otimes H $ on the CDGA $B_\bullet$.
Throughout the whole paper, we shall use the standard Sweedler notation
\begin{subequations}
\begin{flalign}
\Delta^n (h)\,&=\, h_{\und{1}}\otimes h_{\und{2}}\otimes \cdots\otimes h_{\und{n+1}}\,\in\, H^{\otimes n+1}\qquad\text{(summation understood)}\quad,\\
\rho^n (b)\,&=\, b_{\und{0}}\otimes b_{\und{1}}\otimes\cdots \otimes b_{\und{n}}\,\in\,B_\bullet\otimes H^{\otimes n}
\qquad \text{(summation understood)}
\end{flalign}
\end{subequations}
for (iterated) applications of the comultiplication and the coaction. In particular, this entails that
$\Delta (h) = h_{\und{1}}\otimes h_{\und{2}}$ and $\rho (b) = b_{\und{0}}\otimes b_{\und{1}}$.
\sk

The group actions \eqref{eqn:degreewisegroupaction} are algebraically given by
$H^{\otimes n+1}$-comodule structures on the CDGAs $\O(Y\times G^n) = B_\bullet\otimes H^{\otimes n}
\in \dgCAlg_{\geq 0}$, which we shall collectively denote by the same symbol
$\rho : B_\bullet\otimes H^{\otimes n} \to B_\bullet\otimes H^{\otimes n}\otimes H^{\otimes n+1}$.
In order to spell this out, we use that $B_\bullet\otimes H^{\otimes n}$
is generated by elements of the form
\begin{flalign}\label{eqn:elementconvention}
b\,:=\,b\otimes\Motimes_{i=1}^n \oone\,\in\, B_\bullet\otimes H^{\otimes n}~~,\quad
h^{\an{j}}\,:=\,\oone_B^{}\otimes\Motimes_{i=1}^{j-1}\oone \otimes h\otimes\Motimes_{i=j+1}^{n}\oone \,\in\, B_\bullet\otimes H^{\otimes n}\quad,
\end{flalign}
for $j=1,\dots, n$. The coaction is defined on these generators by
\begin{flalign}\label{eqn:coaction}
\rho(b)\,=\, b_{\und{0}}^{}\otimes \Big(b_{\und{1}}\otimes\Motimes_{k=1}^n\oone\Big)~~,\quad
\rho(h^{\an{j}})\,=\,h^{\an{j}}_{\und{2}}\otimes\Big(\Motimes_{k=0}^{j-2}\oone\otimes S(h_{\und{1}}^{})\otimes
h_{\und{3}}^{}\otimes\Motimes_{k=j+1}^n\oone\Big)\quad,
\end{flalign}
where we used parentheses to indicate the factor in $H^{\otimes n+1}$. 
The face maps \eqref{eqn:facemap} are algebraically given by the CDGA-morphisms
\begin{subequations}\label{eqn:facemapdual}
\begin{flalign}
d_i^\ast \,:\, B_\bullet\otimes H^{\otimes n-1} ~&\longrightarrow~ B_\bullet\otimes H^{\otimes n}
\end{flalign}
that are defined on the generators by
\begin{flalign}
d_i^\ast(b)\,&=\,\begin{cases}
b_{\und{0}}^{}\, b_{\und{1}}^{\an{1}}&~~\text{for }i=0\quad,\\
b&~~\text{for } i=1,\dots n\quad,
\end{cases}\quad\\
d_i^{\ast}\big(h^{\an{j}}\big)\,&=\,\begin{cases}
h^{\an{j+1}}&~~\text{for }i=0,\dots,j-1\quad,\\
h_{\und{1}}^{\an{i}}\, h_{\und{2}}^{\an{i+1}}&~~\text{for }i=j\quad,\\
h^{\an{j}}&~~\text{for }i=j+1,\dots,n\quad.
\end{cases}
\end{flalign}
\end{subequations}
The degeneracy maps \eqref{eqn:degeneracymap} are algebraically given by
the CDGA-morphisms
\begin{subequations}\label{eqn:degeneracymapdual}
\begin{flalign}
s_i^\ast\,:\, B_\bullet\otimes H^{\otimes n+1} ~&\longrightarrow~B_\bullet\otimes H^{\otimes n}
\end{flalign}
that are defined on the generators by
\begin{flalign}
s_i^\ast(b)\,=\,b~~,\quad
s_i^\ast\big(h^{\an{j}}\big)\,=\,\begin{cases}
h^{\an{j-1}}&~~\text{for }i=0,\dots,j-2\quad,\\
\varepsilon(h) &~~\text{for }i=j-1\quad,\\
h^{\an{j}}&~~\text{for }i=j,\dots,n\quad.
\end{cases}
\end{flalign}
\end{subequations}
The group homomorphisms \eqref{eqn:facemaptilde} and \eqref{eqn:degeneracymaptilde} 
are algebraically given by Hopf algebra morphisms.
Similarly to \eqref{eqn:elementconvention}, we will use that $H^{\otimes n+1}$
is generated by elements of the form
\begin{flalign}\label{eqn:elementconvention2}
h^{\an{j}}\,:=\,\Motimes_{k=0}^{j-1}\oone\otimes h\otimes\Motimes_{k=j+1}^n\oone\,\in\, H^{\otimes n+1}\quad,
\end{flalign}
where $j=0,\dots,n$ runs from $0$ to $n$ in order to match the labeling conventions in \eqref{eqn:facemaptilde}
and \eqref{eqn:degeneracymaptilde}. The Hopf algebra morphism 
\begin{subequations}\label{eqn:facemaptildedual}
\begin{flalign}
\tilde{d}_i^\ast\,:\, H^{\otimes n}~&\longrightarrow~H^{\otimes n+1}
\end{flalign}
associated with \eqref{eqn:facemaptilde} is defined on the generators by
\begin{flalign}
\tilde{d}_i^\ast\big(h^{\an{j}}\big)\,=\,
\begin{cases}
h^{\an{j+1}} &~~\text{for }i=0,\dots,j\quad,\\
h^{\an{j}} &~~\text{for }i=j+1,\dots,n\quad.
\end{cases}
\end{flalign}
\end{subequations}
The Hopf algebra morphism
\begin{subequations}\label{eqn:degeneracymaptildedual}
\begin{flalign}
\tilde{s}^\ast_i\,:\, H^{\otimes n+2}~&\longrightarrow~H^{\otimes n+1}
\end{flalign}
associated with \eqref{eqn:degeneracymaptilde} is defined on the generators by
\begin{flalign}
\tilde{s}^\ast_i\big(h^{\an{j}}\big)\,=\,\begin{cases}
h^{\an{j-1}}&~~\text{for }i=0,\dots,j-1\quad,\\
h^{\an{j}} &~~\text{for }i=j,\dots, n\quad.
\end{cases}
\end{flalign}
\end{subequations}
The algebraic analog of the commutative diagrams in \eqref{eqn:compatibilitydiagram}
is that the diagrams
\begin{flalign}\label{eqn:compatibilitydiagramdual}
\xymatrix@C=1em{ 
B_\bullet \otimes H^{\otimes n} \otimes  H^{\otimes n+1} ~&~ 
\ar[l]_-{\rho}B_\bullet\otimes H^{\otimes n} ~~&~~
B_\bullet\otimes H^{\otimes n}\otimes H^{\otimes n+1}~&~ 
\ar[l]_-{\rho}B_\bullet\otimes H^{\otimes n} \\
\ar[u]^-{d_i^\ast\otimes \tilde{d}_i^\ast}B_\bullet\otimes H^{\otimes n-1}\otimes H^{\otimes n}~&~
\ar[l]^-{\rho}\ar[u]_-{d_i^\ast} B_\bullet\otimes H^{\otimes n-1}~~&~~
\ar[u]^-{s_i^\ast\otimes \tilde{s}_i^\ast}B_\bullet\otimes H^{\otimes n+1}\otimes H^{\otimes n+2}~&~
\ar[l]^-{\rho}\ar[u]_-{s_i^\ast}B_\bullet\otimes H^{\otimes n+1}
}
\end{flalign}
commute or, in other words, that $d_i^\ast$ is equivariant relative
to the Hopf algebra morphism $\tilde{d}_i^\ast$ and that $s_i^\ast$
is equivariant relative to $\tilde{s}_i^\ast$. (Note that these are Hopf algebra
analogs of the Lie-equivariance conditions from Remark \ref{rem:CEalgebrafunctoriality}.)
\sk

As a last preparation, we describe the Lie algebra action that is induced by the coaction
\eqref{eqn:coaction} and the morphisms of dual Lie algebras that are 
induced by the Hopf algebra morphisms \eqref{eqn:facemaptildedual} and \eqref{eqn:degeneracymaptildedual}.
Recall that the dual Lie algebra of $G=\spec\,H$
is defined by the quotient 
\begin{flalign}
\g^\vee \,:=\, H^+\big/ H^{+2}
\end{flalign}
of the augmentation ideal $H^+ := \ker(\varepsilon : H\to\bbK)\subseteq H$ by its square.
The Lie algebra $\g$ is then defined as the dual of $\g^\vee$
and it can be described explicitly as
\begin{flalign}
\g \,:=\, \mathrm{Lin}(\g^\vee,\bbK)\,\cong\, \Der_{\varepsilon}(H,\bbK)\quad,
\end{flalign}
where $\Der_{\varepsilon}(H,\bbK)$ denotes the vector space of derivations
relative to the counit $\varepsilon: H\to\bbK$, i.e.\ an element $t\in \Der_{\varepsilon}(H,\bbK)$
is a linear map $t:H\to\bbK$ satisfying $t(h\,h^\prime) = t(h)\,\varepsilon(h^\prime) + \varepsilon(h)\,t(h^\prime)$,
for all $h,h^\prime\in H$. More generally, the dual Lie algebra
of $G^{n+1} = \spec H^{\otimes n+1}$ can be identified with $\g^{\vee\oplus n+1}$
and its Lie algebra with $\g^{\oplus n+1}$. Similarly to \eqref{eqn:elementconvention2},
we shall write
\begin{flalign}\label{eqn:elementconvention3}
\theta^{\an{j}} \,:=\,\Moplus_{k=0}^{j-1}0 \oplus \theta\oplus \Moplus_{k=j+1}^n 0 \,\in\,\g^{\vee \oplus n+1}~~,\quad
t^{\an{j}}\,:=\,\Moplus_{k=0}^{j-1}0\oplus t\oplus\Moplus_{k=j+1}^n 0\,\in\,\g^{\oplus n+1}\quad,
\end{flalign}
for $j=0,\dots,n$.
The $H^{\otimes n+1}$-coaction \eqref{eqn:coaction} 
induces the Lie algebra action $\rho : \g^{\oplus n+1}\to \Der(B_\bullet\otimes H^{\otimes n})$
(denoted with abuse of notation by the same symbol) that is given on the generators by
\begin{subequations}\label{eqn:inducedLieactions}
\begin{flalign}\label{eqn:inducedLieactions1}
\rho(t^{\an{j}})(b)\,&=\,\begin{cases}
\rho_B^{}(t)(b)&~~\text{for }j=0\quad,\\
0 &~~\text{else}\quad,
\end{cases}\\
\rho(t^{\an{j}})(h^{\an{k}})\,&=\, \begin{cases}
\rho^{\mathrm{L}}(t)(h)^{\an{k}}&~~\text{for } j=k\quad,\\
\rho^{\mathrm{R}}(t)(h)^{\an{k}}&~~\text{for } j+1=k\quad,\\
0&~~\text{else}\quad,
\end{cases}
\end{flalign}
where
\begin{flalign}\label{eqn:inducedLieactions2}
\rho_B^{}(t)(b) \,:=\, b_{\und{0}}^{}\, t(b_{\und{1}}^{})~~,\quad
\rho^{\mathrm{R}}(t)(h)\,:=\, t(S(h_{\und{1}}^{}))\,h_{\und{2}}^{}~~,\quad
\rho^{\mathrm{L}}(t)(h)\,:=\, h_{\und{1}}^{}\,t(h_{\und{2}}^{})
\end{flalign}
\end{subequations}
are respectively the induced $\g$-action on $B_\bullet$, the $\g$-action on $H$
in terms of right invariant derivations and the $\g$-action on $H$ in terms of left invariant derivations.
The Hopf algebra morphisms \eqref{eqn:facemaptildedual} and \eqref{eqn:degeneracymaptildedual} 
induce to the dual Lie algebras as
\begin{flalign}\label{eqn:facemaptildeLie} 
\tilde{d}_i^\ast\,:\,\g^{\vee \oplus n}~\longrightarrow~\g^{\vee \oplus n+1}~~,\quad
\theta^{\an{j}}~\longmapsto~\begin{cases}
\theta^{\an{j+1}} &~~\text{for }i=0,\dots, j\quad,\\
\theta^{\an{j}} &~~\text{for }i=j+1,\dots,n\quad,
\end{cases}
\end{flalign}
and
\begin{flalign}\label{eqn:degeneracymaptildeLie} 
\tilde{s}_i^\ast\,:\,\g^{\vee\oplus n+2} ~\longrightarrow~\g^{\vee\oplus n+1}~~,\quad
\theta^{\an{j}}~\longmapsto~\begin{cases}
\theta^{\an{j-1}} &~~\text{for }i=0,\dots,j-1\quad,\\
\theta^{\an{j}} &~~\text{for } i=j,\dots,n\quad.
\end{cases}
\end{flalign}
Their duals $\tilde{d}_i : \g^{\oplus n+1}\to \g^{\oplus n}$ and 
$\tilde{s}_i :  \g^{\oplus n+1}\to \g^{\oplus n+2}$ are morphisms of
Lie algebras.
\sk

We are now ready to provide a precise and fully explicit algebraic description
of the simplicial diagram \eqref{eqn:quotientstackresolution} of formal derived quotient stacks
in terms of the cosimplicial diagram
\begin{flalign}\label{eqn:cosimplicialstackyCDGA}
\xymatrix@C=1em{
\CE^\bullet(\g,B_\bullet) \ar@<0.5ex>[r] \ar@<-0.5ex>[r]~&~ 
\CE^\bullet\big(\g^{\oplus 2}, B_\bullet\otimes H\big) \ar@<1ex>[r] \ar[r] \ar@<-1ex>[r] ~&~ 
\CE^\bullet\big(\g^{\oplus 3}, B_\bullet\otimes H^{\otimes 2}\big)\ar@<0.5ex>[r] \ar@<-0.5ex>[r]\ar@<1.5ex>[r] \ar@<-1.5ex>[r]~&~\cdots
}
\end{flalign}
of stacky CDGAs. The stacky CDGAs $\CE^\bullet\big(\g^{\oplus n+1}, B_\bullet\otimes H^{\otimes n}\big)$
can be described precisely as in Example \ref{ex:CEalgebra}: One just has to replace 
the Lie algebra $\g$ by $\g^{\oplus n+1}$, the chain CDGA $B_\bullet$ by $B_\bullet\otimes H^{\otimes n}$
and use the Lie algebra action defined in \eqref{eqn:inducedLieactions}.
As basis for $\g^{\oplus n+1}$ one can take $\big\{t_a^{\an{j}}\in \g^{\oplus n+1}\big\}$,
where $\{t_a\in\g\}$ is a basis for $\g$, and the dual basis then reads as
$\big\{\theta^{a\an{j}}\in \g^{\vee \oplus n+1}\big\}$, where $\{\theta^a\in\g^{\vee}\}$
is the dual basis for $\g^{\vee}$. With respect to this basis,
the structure constants of $\g^{\oplus n+1}$ are given by
$\big[t_a^{\an{j}},t_b^{\an{k}}\big] = \delta^{jk}\,f^c_{ab}\,t_{c}^{\an{j}}$.
The coface (respectively, codegeneracy) maps in \eqref{eqn:cosimplicialstackyCDGA} are the stacky 
CDGA morphisms obtained from  \eqref{eqn:facemapdual} and \eqref{eqn:facemaptildeLie} (respectively,
\eqref{eqn:degeneracymapdual} and \eqref{eqn:degeneracymaptildeLie}) by the construction
explained in Remark \ref{rem:CEalgebrafunctoriality}. Explicitly,
the $i$-th coface map is given by
\begin{flalign}\label{eqn:coface}
d^i := (d_i^\ast,\tilde{d}_i^\ast)\,:\,
\CE^\bullet\big(\g^{\oplus n},B_\bullet\otimes H^{\otimes n-1}\big)~\longrightarrow~
\CE^\bullet\big(\g^{\oplus n+1},B_\bullet\otimes H^{\otimes n}\big)
\end{flalign}
and the $i$-th codegeneracy map is given by
\begin{flalign}\label{eqn:codegeneracy}
s^i := (s_i^\ast,\tilde{s}_i^\ast) 
\,:\,\CE^\bullet\big(\g^{\oplus n+2}, B_\bullet \otimes H^{\otimes n+1}\big)~\longrightarrow~
\CE^\bullet\big(\g^{\oplus n+1}, B_\bullet \otimes H^{\otimes n}\big)\quad.
\end{flalign}

To conclude, we shall globalize the concepts of strict Poisson structures 
(see Definition \ref{def:Poisson}) and
of perfect modules (see Definition \ref{def:perf}) from stacky CDGAs to derived quotient stacks.
\begin{defi}\label{def:Poissonresolution}
An unshifted {\em strict Poisson structure} on a derived quotient stack $[Y/G]= [\spec\,B_\bullet / \spec\,H]$
is a family, indexed by $n\in\bbZ_{\geq 0}$, consisting of an 
unshifted Poisson structure $\{\,\cdot\,,\,\cdot\,\}_n^{}$
in the sense of Definition \ref{def:Poisson} on each stacky CDGA 
$\CE^\bullet\big(\g^{\oplus n+1}, B_\bullet\otimes H^{\otimes n}\big)$.
The Poisson brackets must be preserved by all totalized coface maps \eqref{eqn:coface} and by all
totalized codegeneracy maps \eqref{eqn:codegeneracy}.
\end{defi}

\begin{defi}\label{def:perquotientstack}
The {\em dg-category of perfect modules} over a derived quotient stack $[Y/G]= [\spec\,B_\bullet / \spec\,H]$
is defined as the homotopy limit
\begin{flalign}
\mathrm{per}\big([Y/G]\big) \,:=\,\holim\Big(
\xymatrix@C=1em{
\mathrm{per}\big(\CE^\bullet(\g,B_\bullet)\big) \ar@<0.5ex>[r] \ar@<-0.5ex>[r]~&~ 
\mathrm{per}\big(\CE^\bullet\big(\g^{\oplus 2}, B_\bullet\otimes H\big)\big) \ar@<1ex>[r] \ar[r] \ar@<-1ex>[r] ~&~ 
\cdots
}
\Big)
\end{flalign}
in the model category of dg-categories of the cosimplicial diagram
that is obtained by applying Definition \ref{def:perf} object-wise to the cosimplicial
diagram \eqref{eqn:cosimplicialstackyCDGA} of stacky CDGAs. The coface and 
codegeneracy maps are obtained from \eqref{eqn:coface} and \eqref{eqn:codegeneracy}
via the induced module construction from Remark \ref{rem:perfunctoriality}.
\end{defi}

\begin{propo}\label{prop:dgCatquotientstack}
Suppose that $G=\spec\, H$ is reductive.
Then the dg-category of perfect modules $\mathrm{per}\big([Y/G]\big)$
from Definition \ref{def:perquotientstack} admits the following explicit model: 
Its objects are all cofibrant and perfect left $B_\bullet$-dg-modules
$V_\bullet\in {}_{B_\bullet}\Mod$ that are endowed with a compatible
$H$-coaction $\rho_V^{} : V_\bullet\to V_\bullet\otimes H$, i.e.\
$\rho_V^{}(b\cdot s) = \rho(b)\cdot \rho_V^{}(s)$ for all $b\in B_\bullet$ and $s\in V_\bullet$.
For two objects $V_\bullet$ and $W_\bullet$, the cochain complex of morphisms
is given by the sub-complex
\begin{flalign}
\hom_{B_\bullet}^{H}\big(V_\bullet,W_\bullet\big)\,\subseteq\,
\hom_{B_\bullet}^{~}(V_\bullet,W_\bullet)
\end{flalign}
of left $B_\bullet$-module morphisms that are strictly equivariant with respect to the $H$-coactions, i.e.\
$L\in\hom_{B_\bullet}^{~}(V_\bullet,W_\bullet)$ lies in $\hom_{B_\bullet}^{H}\big(V_\bullet,W_\bullet\big)$ 
if and only if $\rho_{W}^{}\circ L = (L\otimes\id )\circ \rho_V^{}$.
\end{propo}
\begin{proof} 
By the proof of \cite[Proposition 3.11]{PridhamUnshifted}, one can calculate 
$\mathrm{per}\big([Y/G]\big)$ using any flat hypercover of the stack $[Y/G]$.
Using as in \eqref{eqn:quotientstack} the hypercover given by the simplicial affine
scheme $Y \Leftarrow Y \times G \Lleftarrow Y \times G^2 \cdots$, 
we obtain an equivalence
\begin{flalign}
\mathrm{per}\big([Y/G]\big) \,\simeq \,\holim\Big(
\xymatrix@C=1em{
\mathrm{per}\big(B_\bullet\big) \ar@<0.5ex>[r] \ar@<-0.5ex>[r]~&~ 
\mathrm{per}\big( B_\bullet\otimes H\big) \ar@<1ex>[r] \ar[r] \ar@<-1ex>[r] ~&~ 
\cdots
}
\Big)\quad,
\end{flalign}
which we will use to simplify the proof.
(A direct proof using the description of Definition \ref{def:perquotientstack} also exists.
See Proposition \ref{prop:quantdgCatquotientstack} and take the classical limit.)
\sk

Let us denote by $\mathcal{C}$ the dg-category described in the statement of
this proposition. Forgetting the $H$-coaction gives us a dg-functor 
$\mathcal{C} \to \mathrm{per}\big(B_\bullet\big)$. Composing this with the coface maps
$d_0^\ast,d_1^\ast$ gives us dg-functors $\mathcal{C} \to \mathrm{per}\big( B_\bullet\otimes H\big)$ 
which are not equal, but are naturally isomorphic via the $H$-coaction. 
This natural isomorphism $\theta$ satisfies the 
cocycle condition $d_1^\ast\theta =d_2^\ast\theta \circ d_0^\ast\theta$ in 
$\mathrm{per}\big( B_\bullet\otimes H^{\otimes 2}\big)$, which means that we have constructed a 
dg-functor from $\mathcal{C}$ to the $2$-categorical limit of the diagram above. Since the 
$2$-categorical limit maps to the homotopy limit, this gives us a dg-functor 
$\mathcal{C} \to \mathrm{per}\big([Y/G]\big)$, so it remains to show that this is a quasi-equivalence. 
\sk

On morphisms, the dg-functor is given by
\begin{flalign}
\resizebox{0.9\hsize}{!}{$\hom_{B_\bullet}^{H}(V_\bullet,W_\bullet) ~\longrightarrow~ \holim\Big(
\xymatrix@C=1em{
\hom_{B_\bullet}^{~}(V_\bullet,W_\bullet)\ar@<0.5ex>[r] \ar@<-0.5ex>[r]~&~
\hom_{B_\bullet\otimes H}^{~}(V_\bullet\otimes H,W_\bullet\otimes H)
 \ar@<1ex>[r] \ar[r] \ar@<-1ex>[r] ~&~ 
\cdots
}
\Big)\quad.$}
\end{flalign}
The homotopy limit in the target can be computed
by the normalized group cohomology complex $\hat{\Tot} N^\bullet(G,\hom_{B_\bullet}(V_\bullet,W_\bullet))$.
Since $G$ is reductive, the higher group cohomology groups vanish, hence we obtained the
desired quasi-isomorphism. 
\sk

It remains to establish essential surjectivity.
The functor sending a chain CDGA $B_\bullet$ to the category of all left  
$B_\bullet$-dg-modules is a left Quillen hypersheaf in the sense of 
\cite[\S 17]{HirschowitzSimpson}. By the strictification theorem \cite[Corollary 18.7]{HirschowitzSimpson}, 
adapted to a diagram of full subcategories (specifically, pseudo-model categories) as in 
\cite[Corollary 1.3.7.4]{hag2}, we can thus characterize $\mathrm{per}\big([Y/G]\big)$ as the 
dg-category $\mathrm{per}\big([Y/G]\big)_{\mathrm{cart}}$ of those left modules 
\begin{flalign}
 \xymatrix@C=1em{
M^0_{\bullet} \ar@<0.5ex>[r] \ar@<-0.5ex>[r]~&~ 
M^1_{\bullet} \ar@<1ex>[r] \ar[r] \ar@<-1ex>[r] ~&~ 
M^2_{\bullet} \ar@<1.5ex>[r] \ar@<0.5ex>[r] \ar@<-0.5ex>[r]\ar@<-1.5ex>[r] ~&~ 
\cdots
}
\end{flalign}
in cosimplicial chain complexes over the cosimplicial CDGA
\begin{flalign}\label{eqn:cosimplicialCDGA}
\xymatrix@C=1em{
B_\bullet \ar@<0.5ex>[r] \ar@<-0.5ex>[r]~&~ 
 B_\bullet\otimes H \ar@<1ex>[r] \ar[r] \ar@<-1ex>[r] ~&~ 
 B_\bullet\otimes H^{\otimes 2}  \ar@<1.5ex>[r] \ar@<0.5ex>[r] \ar@<-0.5ex>[r]\ar@<-1.5ex>[r] ~&~ 
\cdots
}
\end{flalign}
which
\begin{enumerate}
\item are cofibrant in the projective model structure on cosimplicial dg-modules over \eqref{eqn:cosimplicialCDGA},
\item have $M^0_{\bullet}$ perfect as a $B_{\bullet}$-dg-module, and
\item are homotopy-Cartesian in the sense that the morphisms
\begin{flalign}
 d^j \,:\, \big(B_{\bullet}\otimes H^{\otimes n+1}\big)  \otimes_{(B_{\bullet}\otimes H^{\otimes n})}^{}
 M_\bullet^n~\longrightarrow~ M_\bullet^{n+1}
\end{flalign}
from the induced modules are all quasi-isomorphisms.
\end{enumerate}
We can think of these as $B_{\bullet}$-dg-modules equipped with a strong homotopy $H$-coaction,
which we will now strictify.
\sk

Given an object $V_\bullet \in \mathcal{C}$, the associated  object of 
$\mathrm{per}\big([Y/G]\big)_{\mathrm{cart}}$ is a cofibrant replacement of the homotopy-Cartesian 
cosimplicial module  $\big\{C^n(G,V_{\bullet})\big\}_{n\geq 0}$ given by setting 
$C^n(G,V_{\bullet}):=V_{\bullet}\otimes H^{\otimes n}$, with the coface 
maps $d^i$ given by the coaction $\rho$, coproducts and units,
and the codegeneracy maps $s^i$ by counits. In other words, $ C^\bullet(G,V_\bullet)$ 
is the complex calculating algebraic group cohomology of $G$ with coefficients in $V_\bullet$.
For any homotopy-Cartesian module $M_{\bullet}^{\bullet} \in \mathrm{per}\big([Y/G]\big)_{\mathrm{cart}}$, 
the homology groups $\mathrm{H}_i\big(M^{\bullet}_\bullet)$ form a Cartesian cosimplicial module, 
in the sense that the morphisms
\begin{flalign}
 d^j \,:\, 
 \big(\mathrm{H}_0(B_\bullet)\otimes H^{\otimes n+1}\big) \otimes_{(\mathrm{H}_0(B_\bullet) \otimes H^{\otimes n})}^{}  \mathrm{H}_i\big(M^n_\bullet\big)  ~\longrightarrow~ \mathrm{H}_i\big(M^{n+1}_\bullet\big)
\end{flalign}
are all isomorphisms. These give, and are 
determined by, $H$-coactions on the $\mathrm{H}_0(B_\bullet)$-modules 
$\mathrm{H}_i\big(M^0_\bullet\big)$. By taking projective resolutions and cones, we can then show by 
induction on the lowest non-zero homology group of $M^0_\bullet$ that every object is homotopy 
equivalent to one in the image of $\mathcal{C}$.
\end{proof}

\begin{rem}\label{rem:reductive}
Our hypothesis that $G$ is reductive is crucial to obtain
the simple model for the homotopy limit in Proposition \ref{prop:dgCatquotientstack}.
Dropping this hypothesis would lead to homotopy coherent $H$-comodule structures, see \cite{holimdgCat},
which are much harder to work with. 
We would like to recall that many of the standard examples of groups are 
reductive, such as $\mathrm{GL}(n)$, $\mathrm{SL}(n)$, $\mathrm{O}(n)$, 
$\mathrm{SO}(n)$, $\mathrm{Sp}(n)$ and tori.
\end{rem}


\section{\label{sec:cotangent}Quantization of derived cotangent stacks}
The goal of this section is to study Poisson structures 
and their quantizations for an explicit class of examples of derived quotient stacks.
Consider an ordinary smooth affine scheme $X = \spec\,A$, i.e.\ $A\in \CAlg$ is a 
smooth commutative algebra, together with an action
$X\times G\to X$  of a smooth affine group scheme $G=\spec\,H$, i.e.\
$H$ is a smooth commutative Hopf algebra. The corresponding
quotient stack $[X/G]$ is defined via \eqref{eqn:quotientstack}
as a colimit in the $\infty$-category of derived stacks. We are interested
in the derived cotangent stack $T^\ast[X/G]$, which by \cite{CalaqueTangent} carries
a canonical unshifted symplectic structure and hence also an unshifted Poisson structure.
This section is divided into four subsections: In Subsection \ref{subsec:sympred},
we compute $T^\ast[X/G]\simeq [T^\ast X/\!\!/G]$ explicitly as a derived symplectic reduction,
see e.g.\ \cite{SafronovImplosion} and also \cite[Section 4]{BVgroup} for more details.
In particular, this provides a presentation of the derived cotangent stack $T^\ast[X/G]$ 
as a derived quotient stack.
In Subsection \ref{subsec:cotangentresolution}, we apply the resolution techniques
from Subsection \ref{subsec:resolution} to the example $T^\ast[X/G]\simeq [T^\ast X/\!\!/G]$
and spell out its canonical unshifted Poisson structure fully explicitly.
In Subsection \ref{subsec:localquantization}, we construct quantizations
of the dg-categories of perfect modules over the individual stacky CDGAs
entering our resolution, which are then globalized in Subsection \ref{subsec:globalquantization}
to a quantization of the dg-category of perfect modules
over $T^\ast[X/G]\simeq [T^\ast X/\!\!/G]$.


\subsection{\label{subsec:sympred}Symplectic reduction}
We recall from \cite{SafronovImplosion} that 
the derived cotangent stack $T^\ast[X/G]\simeq [T^\ast X/\!\!/G]$
is equivalent to the derived symplectic reduction of the Hamiltonian
$G$-space $(T^\ast X,\omega,\mu)$ given by 
\begin{itemize}
\item the cotangent bundle
$T^\ast X$ over $X=\spec\,A$, 
\item the canonical symplectic structure $\omega := \dd^{\dR}\lambda\in\Omega^2(T^\ast X)$, where
$\lambda\in \Omega^1(T^\ast X)$ is the tautological $1$-form, and 
\item the moment map
$\mu : T^\ast X \to \g^\vee$ defined by $\langle\mu,t\rangle := -\iota_{\rho(t)}\lambda$, 
for all $t\in\g$, i.e.\ the contraction
of $-\lambda$ against the vector field $\rho(t)$ obtained from the Lie algebra action on $T^\ast X$.
\end{itemize}
Symplectic reduction is a two-step construction: First, one takes
the fiber product
\begin{flalign}
\xymatrix{
\ar@{-->}[d]\mu^{-1}(0)\ar@{-->}[r]~&~\mathrm{pt} \ar[d]^-{0}\\
T^\ast X \ar[r]_-{\mu}~&~\g^\vee
}
\end{flalign}
in the $\infty$-category of derived stacks, which determines the derived zero locus
of the moment map, and then one forms the derived quotient stack
\begin{flalign}
[T^\ast X/\!\!/G]\,:=\, [\mu^{-1}(0)/G]
\end{flalign}
using the inherited $G$-action on $\mu^{-1}(0)$. Since the 
derived stack $\mu^{-1}(0)$ is affine, the derived
cotangent stack $T^\ast[X/G]\simeq [T^\ast X/\!\!/G]= [\mu^{-1}(0)/G]$
is a special instance of the class of examples $[Y/G]$ discussed in Subsection \ref{subsec:resolution}.
\sk

Using the calculation of homotopy pushouts in \cite[Section 4]{BVgroup}, 
we can spell out fully explicitly the chain CDGA 
\begin{flalign}
B_\bullet \,:=\, \O(\mu^{-1}(0))_\bullet \,\in\,\dgCAlg_{\geq 0}
\end{flalign}
of functions on $\mu^{-1}(0)$ and its $H$-comodule structure
$\rho :  B_\bullet\to B_\bullet \otimes H$.
This however requires some preparations and introducing some notations.
Let us start by recalling that, as input data, we are given a 
smooth commutative algebra $A=\O(X)\in\CAlg$
together with an $H$-comodule structure $\rho_A^{} : A\to A\otimes H$.
The function algebra on the cotangent bundle 
\begin{flalign}
\O(T^\ast X)\,=\, \Sym_{A}^{} \T_{\! A}\,\in\,\CAlg\quad
\end{flalign}
is given by the symmetric algebra over $A$ of the $A$-module 
$\T_{\! A} = \Der(A)\in {}_A\Mod$ of derivations of $A$, i.e.\ 
an element $v\in \T_{\! A}$ is a linear map $v:A\to A$ satisfying
the Leibniz rule $v(a\,a^\prime) = v(a)\,a^\prime + a\,v(a^\prime)$, for all $a,a^\prime \in A$.
We endow the $A$-module $\T_{\! A}=\Der(A)$ with the adjoint $H$-coaction
$\rho_{\T_{\! A}}^{} : \T_{\! A} \to \T_{\! A}\otimes H$ defined by
\begin{flalign}
\rho_{\T_{\! A}}^{}(v) (a)\,=\, \big(v(a^{}_{\und{0}})\big)_{\und{0}}^{}\otimes \, \big(v(a^{}_{\und{0}})\big)_{\und{1}}^{}\, S(a_{\und{1}}^{})\quad,
\end{flalign}
for all $v\in \T_{\! A}$ and $a\in A$, which  equivalently can be defined as the dual of the 
canonical coaction $\rho_{\Omega^1_A}: \Omega^1_A\to \Omega^1_A\otimes H $ on K\"ahler $1$-forms.
Observe that the evaluation map $\T_{\! A} \otimes A \to A\,,~v\otimes a\mapsto v(a)$ is $H$-equivariant.
Using tensor product coactions, this yields an $H$-coaction $\rho_{\Sym_A\T_{\! A}}^{} : 
\Sym_A\T_{\! A}\to \Sym_A \T_{\! A}\otimes H$ on the symmetric algebra.
We further endow the dual Lie algebra $\g^\vee$ with the adjoint $H$-coaction
$\rho_{\g^{\vee}}^{} : \g^\vee\to \g^\vee\otimes H$ given by
\begin{flalign}
\rho_{\g^{\vee}}^{}(\theta)\,=\, \theta_{\und{2}}^{}\otimes S(\theta_{\und{1}}^{})\,\theta_{\und{3}}^{}\quad,
\end{flalign}
for all $\theta \in \g^\vee = H^+/H^{+2}$, and the Lie algebra
$\g$ with the dual  $\rho_{\g}^{}:\g\to\g\otimes H$ of this coaction.
By construction, the evaluation map $\langle\,\cdot\, , \,\cdot\,\rangle :
\g^\vee\otimes \g \to\bbK $ is then $H$-equivariant.
The tautological $1$-form $\lambda \in \Omega^1(T^\ast X) = \Omega^1_{\Sym_A\T_{\! A}}$
can be defined via the coevaluation map $\mathrm{coev} : A\to \T_{\! A}\otimes \Omega^1_A$
for the dual pair $(\T_{\! A},\Omega^1_A)$ of $A$-modules given by the derivations
$\T_{\! A}=\Der(A)$ and the K\"ahler $1$-forms $\Omega^1_A$.
Explicitly, we have that
\begin{flalign}
\lambda\,:=\,\mathrm{coev}(\oone)\,\in  \T_{\! A}\otimes \Omega^1_A\,\subseteq \, \Sym_A \T_{\! A} \otimes \Omega^1_A\,\subseteq\, \Omega^1_{\Sym_A\T_{\! A}}\quad.
\end{flalign}
The moment map $\mu : T^\ast X\to \g^\vee$ is algebraically given by
the algebra map
\begin{subequations}\label{eqn:momentmap}
\begin{flalign}
\mu^\ast\,:\ \Sym\,\g ~\longrightarrow~\Sym_A\,\T_{\! A}
\end{flalign}
that is defined on the generators $t\in \g$ by
\begin{flalign}
\mu^\ast(t)\,=\,-\iota_{\rho_{\Sym_A\T_{\! A}}^{}(t)}\lambda\, \in\, \T_{\! A}\, \subseteq\, \Sym_A\T_{\! A}\quad,
\end{flalign}
\end{subequations}
where $\iota : \T_{\Sym_A \T_{\! A}} \otimes \Omega^1_{\Sym_A\T_{\! A}}\to\Sym_A\T_{\! A}$ 
denotes the contraction between the derivations and the K\"ahler $1$-forms on $\Sym_A\T_{\! A}$,
and $\rho_{\Sym_A\T_{\! A}}^{}(t)\in \T_{\Sym_A \T_{\! A}}= \Der(\Sym_A \T_{\! A})$ is the Lie algebra action that is
induced by the $H$-coaction $\rho_{\Sym_A\T_{\! A}}^{} :  \Sym_A\T_{\! A}\to \Sym_A \T_{\! A}\otimes H$.
We are now in the position to extract from \cite[Section 4]{BVgroup} that
\begin{flalign}\label{eqn:Bbulletsym}
B_{\bullet}\,:=\,\O(\mu^{-1}(0))_\bullet\,=\,\Sym_A\Big(
\xymatrix@C=1.5em{
\T_{\! A} ~&~\ar[l]_-{\mu^\ast} A\otimes \g_{[-1]}^{}
} \Big)
\end{flalign}
is the symmetric algebra over $A$ of a chain complex of $A$-modules (concentrated in degrees $0$ and $1$)
with differential given by the $A$-linear extension 
$\mu^\ast : A\otimes \g_{[-1]}^{} \to \T_{\! A}\,,~a\otimes t \mapsto a\,\mu^\ast(t)$
of the moment map \eqref{eqn:momentmap}. In other words, this means that $B_\bullet$
is generated by elements of the form $a\in A\subseteq B_0$,
$v\in \T_{\! A}\subseteq B_0$ and $t\in \g_{[-1]}^{}\subseteq B_{1}$
and that the chain differential $\partial$ is given on the generators by
\begin{flalign}\label{eqn:chaindiffmu}
\partial (a)\,=\,0~~,\quad \partial(v)\,=\,0~~,\quad \partial(t)\,=\, \mu^\ast(t) \quad,
\end{flalign}
with $\mu^\ast$ given explicitly in \eqref{eqn:momentmap}.
The $H$-coaction $\rho: B_\bullet\to B_\bullet\otimes H$
is defined on the generators by the $H$-coactions introduced above, i.e.\
\begin{flalign}\label{eqn:rhosymp}
\rho(a)\, =\,\rho_A^{}(a)~~,\quad \rho(v)\,=\,\rho_{\T_{\! A}}^{}(v)~~,\quad\rho(t)\,=\,\rho_{\g}^{}(t)\quad.
\end{flalign}
Note that $H$-equivariance of the differential $\partial$ is a consequence
of $H$-equivariance of the moment map.


\subsection{\label{subsec:cotangentresolution}Resolution by stacky CDGAs and Poisson structure}
We have shown in the previous subsection that the derived
cotangent stack $T^\ast[X/G]\simeq [T^\ast X/\!\!/G]= [\mu^{-1}(0)/G]$
is a derived quotient stack, hence the resolution techniques from Subsection 
\ref{subsec:resolution} can be applied. The aim of this subsection is
to spell out this resolution by stacky CDGAs fully explicitly
and to describe the canonical unshifted Poisson structure on $T^\ast[X/G]$
in terms of this resolution.
\sk

The $n$-th level of the cosimplicial diagram \eqref{eqn:cosimplicialstackyCDGA}
is given by the Chevalley-Eilenberg stacky CDGA $\CE^\bullet\big(\g^{\oplus n+1},B_\bullet\otimes H^{\otimes n}\big)$
with $B_\bullet\in \dgCAlg_{\geq 0}$ given explicitly in \eqref{eqn:Bbulletsym}.
The underlying bigraded commutative algebra (obtained by forgetting both 
differentials) is given by
\begin{flalign}\label{eqn:CEdoublesharp}
\CE^\sharp\big(\g^{\oplus n+1},B_\sharp\otimes H^{\otimes n}\big)\,\cong\, 
\big(\Sym\,\g_{}^{\vee \oplus n+1 [-1]}\big)^{\sharp}\otimes \big(\Sym\, \g_{[-1]}^{}\big)_{\sharp}\otimes 
\Sym_A \T_{\! A}\otimes H^{\otimes n}\quad.
\end{flalign}
This bigraded algebra is generated by the following elements: 
The generators of bidegree $\mycom{0}{0}$ are 
$a\in A$, $v\in \T_{\! A}$ and $h^{\an{k}}\in H^{\otimes n}$, for $k=1,\dots,n$.
(Recall the notation introduced in \eqref{eqn:elementconvention}.)
The generators of bidegree $\mycom{0}{1}$ are
$t\in \g_{[-1]}^{}$ and 
the generators of bidegree $\mycom{1}{0}$ are
$\theta^{\an{j}}\in \g_{}^{\vee \oplus n+1 [-1]}$, for $j=0,\dots, n$. 
(Recall the notation introduced in \eqref{eqn:elementconvention3}.)
Recalling also \eqref{eqn:chaindiffmu},
we observe that the chain differential acts on the generators as
\begin{flalign}\label{eqn:chaincotangent}
\partial(a)\,=\,0~~,\quad \partial(v)\,=\,0~~,\quad
\partial(t) \,=\, \mu^\ast(t)~~,\quad
\partial\big(h^{\an{k}}\big)\,=\, 0~~,\quad
\partial\big(\theta^{\an{j}}\big)\,=\,0\quad.
\end{flalign}
Using \eqref{eqn:inducedLieactions},
we find that the cochain differential acts on the generators as
\begin{flalign}\label{eqn:cochaincotangent}
\nn \delta(a)\,&=\, \theta^{b \an{0}}\,\rho_{A}^{}(t_{b})(a)~~,\quad
\delta(v)\,=\,\theta^{b \an{0}}\,\rho_{\T_{\! A}}^{}(t_{b})(v)~~,\quad
\delta(t)\,=\,\theta^{b \an{0}}\,\rho_{\g}^{}(t_{b})(t)\quad,\\
\delta\big(h^{\an{k}}\big)\,&=\,\theta^{b \an{k-1}}\,\rho^{\mathrm{R}}(t_{b})(h)^{\an{k}}
+ \theta^{b \an{k}}\,\rho^{\mathrm{L}}(t_{b})(h)^{\an{k}}~~,
\quad  \delta\big(\theta^{a \an{j}}\big)\,=\,-\tfrac{1}{2} f^{a}_{bc}~\theta^{b\an{j}}\,\theta^{c\an{j}}\quad,
\end{flalign}
where we recall that $\{t_a\in\g\}$ and $\{\theta^a\in\g^\vee\}$ denotes a pair of dual bases.
\sk

The canonical unshifted Poisson structure on the formal derived 
cotangent stack $T^\ast[X/\g]$ is described by the following Poisson bracket
on the totalized stacky CDGA $\hat{\Tot}\CE^\bullet(\g,B_\bullet)$:
For the generators of total degree $0$, we set
\begin{subequations}\label{eqn:Poissonlevel0}
\begin{flalign}
\{a,a^\prime\}_0^{}\,=\,0~~,\quad
\{v,a\}_0^{}\,=\, v(a)\,=\, -\{a,v\}_{0}^{}~~,\quad
\{v,v^\prime\}_{0}^{}\,=\,[v,v^\prime]\quad,
\end{flalign}
where $v(a)\in A$ denotes the evaluation of the derivation $v\in\T_{\! A}=\Der(A)$ on $a\in A$
and $[v,v^\prime]\in\T_{\! A}$ denotes the Lie bracket (i.e.\ the commutator) of derivations. 
For the generators $t\in\g_{[-1]}$ in total degree $-1$ 
and the generators $\theta\in \g_{}^{\vee [-1]}$ 
in total degree $+1$, we set
\begin{flalign}
\{t, \theta \}_0^{}\,=\, -\langle\theta,t\rangle \,=\,\{\theta, t\}_0^{} \quad,
\end{flalign}
\end{subequations}
where $\langle\,\cdot\,,\,\cdot\,\rangle : \g^\vee \otimes\g \to \bbK$
denotes the duality pairing. (In a pair of dual bases, the latter Poisson brackets
reads as $\{t_a,\theta^b\}_0^{} = -\delta_a^b = \{\theta^b,t_a\}_0^{}$.)
All other Poisson brackets between generators are taken to be trivial.
It is easy to prove that $\{\,\cdot\,,\,\cdot\,\}_0^{}$ satisfies the properties
listed in Definition \ref{def:Poisson}, hence it defines an unshifted 
Poisson structure on the stacky CDGA $\CE^\bullet\big(\g,B_\bullet\big)$.
\sk

To obtain an unshifted strict Poisson structure on the derived 
cotangent stack $T^\ast[X/G]\simeq [\mu^{-1}(0)/G]$,
we have to construct a compatible family $\{\,\cdot\, , \,\cdot\,\}_n^{}$ of unshifted
Poisson structures on $\CE^\bullet\big(\g^{\oplus n+1},B_\bullet\otimes H^{\otimes n}\big)$ 
that extends \eqref{eqn:Poissonlevel0}, see Definition \ref{def:Poissonresolution}.
Let us present a concrete construction.
From \eqref{eqn:facemapdual} and \eqref{eqn:facemaptildeLie},
we find that, for all $m\in\bbZ_{\geq 1}$, the last of the totalized coface maps 
$d^m := (d^{\ast}_m, \tilde{d}^\ast_m): 
\hat{\Tot}\CE^\bullet\big(\g^{\oplus m},B_\bullet\otimes H^{\otimes m-1}\big) \to 
\hat{\Tot}\CE^\bullet\big(\g^{\oplus m+1},B_\bullet\otimes H^{\otimes m}\big)$ acts as the identity
on all generators. Hence, from the family of compatibility conditions 
$d^m \{\,\cdot\,,\,\cdot\,\}_{m-1}^{} = \{d^m(\,\cdot\,),d^m(\,\cdot\,)\}_m^{}$,
one finds that $\{\,\cdot\, , \,\cdot\,\}_{n}^{}$ must act as in \eqref{eqn:Poissonlevel0} among generators
of the type $a,v,t,\theta^{\an{0}}\in \hat{\Tot}\CE^\bullet\big(\g^{\oplus n+1},B_\bullet\otimes H^{\otimes n}\big)$,
i.e.\
\begin{flalign}\label{eqn:Poissonleveln1}
\{a,a^\prime\}_n^{}\,=\,0~~,\quad
\{v,a\}_n^{}\,=\, v(a)~~,\quad
\{v,v^\prime\}_{n}^{}\,=\,[v,v^\prime]~~,\quad
\big\{t , \theta^{\an{0}}\big\}_n^{}\,=\, -\langle\theta,t\rangle\quad.
\end{flalign}
Demanding that the Poisson brackets of the form $\{\,\cdot\, , h^{\an{k}}\}_n^{}=0$
are trivial, for all $k=1,\dots,n$, it remains to determine
$\{\,\cdot\, , \theta^{\an{j}}\}_n^{}$, for all $j=1,\dots,n$.
Writing again $d^i := (d^{\ast}_i, \tilde{d}^\ast_i)$
for the totalized coface maps, we obtain from \eqref{eqn:facemaptildeLie} that 
$\theta^{\an{j}} = d^{n}\cdots d^{j+1}(d^0)^j (\theta) \in 
\hat{\Tot}\CE^\bullet\big(\g^{\oplus n+1},B_\bullet\otimes H^{\otimes n}\big) $,
with $\theta\in \hat{\Tot}\CE^\bullet(\g,B_\bullet)$. Hence, the required 
compatibility conditions enforce that
\begin{flalign}\label{tmp:compatibility1}
\big\{d^{n}\cdots d^{j+1}(d^0)^j(\,\cdot\,) , \theta^{\an{j}}\big\}_n^{}
\,=\,d^{n}\cdots d^{j+1}(d^0)^j \{\,\cdot\, , \theta\}_0^{}\quad.
\end{flalign}
Writing $b\in B_\bullet$ for one of the generators
$a,v,t\in \hat{\Tot}\CE^\bullet\big(\g^{\oplus n+1},B_\bullet\otimes H^{\otimes n}\big)$,
we compute using \eqref{eqn:facemapdual} that
\begin{flalign}\label{tmp:compatibility2}
\big\{d^{n}\cdots d^{j+1}(d^0)^j(b) , \theta^{\an{j}}\big\}_n^{}
\,=\, \big\{ b_{\und{0}}^{}\,b_{\und{1}}^{\an{1}}\,\cdots\,b_{\und{j}}^{\an{j}} , \theta^{\an{j}} \big\}_n^{}
\,=\,\big\{ b_{\und{0}}^{} , \theta^{\an{j}} \big\}_n^{}~\,b_{\und{1}}^{\an{1}}\,\cdots\,b_{\und{j}}^{\an{j}}\quad,
\end{flalign}
where in the last step we have used that $\{\,\cdot\,, h^{\an{k}} \}_n^{}=0$.
Combining \eqref{tmp:compatibility1} and \eqref{tmp:compatibility2},
we obtain
\begin{flalign}
\big\{b, \theta^{\an{j}} \big\}_n^{} \,=\, \big(d^{n}\cdots d^{j+1}(d^0)^j 
\{b_{\und{0}}^{} , \theta\}_0^{}\big)~S(b_{\und{1}}^{})^{\an{j}}\,S(b_{\und{2}}^{})^{\an{j-1}}\cdots
S(b_{\und{j}}^{})^{\an{1}}\quad.
\end{flalign}
Recalling also \eqref{eqn:Poissonlevel0}, we obtain that the only non-vanishing Poisson
brackets of this type are
\begin{flalign}\label{eqn:Poissonleveln2}
\big\{t , \theta^{\an{j}}\big\}_n^{}\,=\, - \langle\theta,t_{\und{0}^{}}\rangle ~
S(t_{\und{1}}^{})^{\an{j}}\,S(t_{\und{2}}^{})^{\an{j-1}}\cdots
S(t_{\und{j}}^{})^{\an{1}}\quad,
\end{flalign}
for all $j=1,\dots,n$.
\begin{propo}\label{prop:unshiftedcotangent}
For each $n\in \bbZ_{\geq 1}$, 
the map $\{\,\cdot\,,\,\cdot\,\}_n^{}$ given by \eqref{eqn:Poissonleveln1},
\eqref{eqn:Poissonleveln2} and zero for all other combinations of generators 
defines an unshifted Poisson structure on the stacky CDGA
$\CE^\bullet\big(\g^{\oplus n+1},B_\bullet\otimes H^{\otimes n}\big)$.
Together with $\{\,\cdot\, ,\,\cdot\, \}_0^{}$ given in \eqref{eqn:Poissonlevel0},
this family defines an unshifted Poisson structure on the derived cotangent stack
$T^\ast[X/G]\simeq [\mu^{-1}(0)/G]$ (see Definition \ref{def:Poissonresolution}).
\end{propo}
\begin{proof}
One immediately checks that each $\{\,\cdot\,,\,\cdot\,\}_n^{}$ satisfies the three conditions
(i), (ii) and (iii) in Definition \ref{def:Poisson}. Showing that $\{\,\cdot\, , \,\cdot\,\}_n^{}$ is a 
cochain map with respect to the total differential $\hat{\dd} = \partial + \delta$ obtained from 
\eqref{eqn:chaincotangent} and \eqref{eqn:cochaincotangent} is a slightly lengthy but straightforward
calculation at the level of the various combinations of generators. 
\sk

The compatibility conditions $d^i\{\,\cdot\,,\,\cdot\,\}_{n-1}^{} = 
\{d^i(\,\cdot\,),d^i(\,\cdot\,)\}_n^{}$ for 
the coface maps \eqref{eqn:coface} can be checked using 
\eqref{eqn:facemapdual} and \eqref{eqn:facemaptildeLie}.
The compatibility conditions $s^i\{\,\cdot\,,\,\cdot\,\}_{n+1}^{} 
= \{s^i(\,\cdot\,),s^{i}(\,\cdot\,)\}_n$ for the codegeneracy maps
\eqref{eqn:codegeneracy} can be checked using 
\eqref{eqn:degeneracymapdual} and \eqref{eqn:degeneracymaptildeLie}.
\end{proof}

\begin{rem}
The unshifted Poisson structure on the derived cotangent stack 
$T^\ast[X/G]\simeq [\mu^{-1}(0)/G]$ constructed in Proposition 
\ref{prop:unshiftedcotangent} above is compatible with
the canonical symplectic structure on derived cotangent stacks from \cite{CalaqueTangent}. 
The comparison between symplectic and Poisson structures 
can be performed via \cite[Theorem 3.33]{PridhamPoisson}
or \cite[Theorem 3.2.4]{DAG2}.
\end{rem}


\subsection{\label{subsec:localquantization}Local quantizations by differential operators and $D$-modules}
We shall now explain how to quantize, for each non-negative integer $n\in \bbZ_{\geq 0}$, 
the totalized stacky CDGAs 
$\hat{\Tot}\CE^\bullet(\g^{\oplus n+1},B_\bullet \otimes H^{\otimes n})$
along the unshifted Poisson bracket $\{\,\cdot\, , \,\cdot\,\}_n^{}$ which was
defined for $n=0$ in \eqref{eqn:Poissonlevel0} and for $n\geq 1$
in \eqref{eqn:Poissonleveln1} and \eqref{eqn:Poissonleveln2}.
\sk

In order to illustrate our construction, 
let us start with the simplest case $n=0$. Note that the underlying
graded algebra of $\hat{\Tot}\CE^\bullet(\g,B_\bullet)$ 
can be identified with the symmetric algebra 
\begin{subequations}\label{eqn:symderivations}
\begin{flalign}
\hat{\Tot}\CE^\bullet(\g,B_\bullet)^\sharp \,\cong\, \Sym_{(\Sym \,\g^{\vee[-1]})^\sharp \otimes A}^{} 
\T_{(\Sym\, \g^{\vee[-1]})^\sharp \otimes A}^{}
\end{flalign}
of the $(\Sym\, \g^{\vee[-1]})^\sharp \otimes A$-module
of derivations 
\begin{flalign}
\T_{(\Sym \,\g^{\vee[-1]})^\sharp \otimes A}^{} \,=\, \Der\big((\Sym \,\g^{\vee[-1]})^\sharp \otimes A\big)\quad.
\end{flalign}
\end{subequations}
This means that we may interpret the generators $t\in \g_{[-1]}$ and $v\in \T_{\!A}$
as derivations on the graded algebra $(\Sym\, \g^{\vee[-1]})^\sharp\otimes A$ generated
by $\theta \in \g^{\vee[-1]}$ and $a\in A$: The derivation associated
with $v\in \T_{\! A} = \Der(A)$ is given by the evaluation $v(a)\in A$
and $v(\theta)=0$, and the derivation associated with $t\in\g_{[-1]}$
is given by $t(a)=0$ and the negative duality pairing $t(\theta) = -\langle\theta,t\rangle$.
Through this identification, the total differential $\hat{\dd} = \partial + \delta$ 
remains unaltered, i.e.\ $\partial$ and $\delta$ are still
defined on the generators by \eqref{eqn:chaincotangent} and \eqref{eqn:cochaincotangent},
and the Poisson bracket \eqref{eqn:Poissonlevel0} can be understood
in terms of evaluations of derivations and their Lie brackets.
\sk

The symmetric algebra \eqref{eqn:symderivations} of the module of derivations 
admits a standard quantization along the canonical unshifted Poisson 
structure, which can be described very explicitly 
via differential operators. (See e.g.\ \cite[Example 1.20]{PridhamUnshifted}.) 
Let $\hbar$ denote a formal deformation parameter. 
We define the {\em noncommutative} graded algebra of differential operators
\begin{subequations}\label{eqn:AAAhbar}
\begin{flalign}
\AAA^{0\sharp}_\hbar\,:=\, \mathrm{DiffOp}^\sharp_\hbar\big((\Sym \,\g^{\vee[-1]})^\sharp \otimes A[[\hbar]]\big)
\end{flalign}
in terms of an explicit presentation by generators and relations:
It is generated over the formal power series extension
$(\Sym \,\g^{\vee[-1]})^\sharp \otimes A[[\hbar]]$ 
by the differential operators $\hat{v}$ and $\hat{t}$, 
for all $v\in\T_{\! A}$ and $t\in\g_{[-1]}$, that act as
\begin{flalign}
\hat{v}(a) \,=\, \hbar\,v(a)~~,\quad
\hat{v}(\theta)\,=\,0~~,\quad
\hat{t}(a)\,=\,0~~,\quad
\hat{t}(\theta) \,=\,-\hbar\,\langle \theta,t\rangle\quad.
\end{flalign}
Denoting by $\hat{a}$ and $\hat{\theta}$ the differential operators
that act by left multiplication on $(\Sym \,\g^{\vee[-1]})^\sharp \otimes A[[\hbar]]$,
we obtain the following commutation relations
\begin{flalign}
\nn \hat{a}\,\hat{a}^\prime -\hat{a}^\prime\, \hat{a} \,&=\, 0~~,\quad
\hat{v}\,\hat{a} - \hat{a}\,\hat{v} \,=\, \hbar\, \widehat{v(a)} ~~,\quad 
\hat{a}\,\hat{t} - \hat{t}\,\hat{a}  \,=\,0~~,\quad 
\hat{a}\,\hat{\theta} - \hat{\theta}\,\hat{a} \,=\,0~~,\quad \\
\nn 
\hat{v}\,\hat{v}^\prime - \hat{v}^\prime\,\hat{v} \,&=\, \hbar\,\widehat{[v,v^\prime]}~~,\quad 
\hat{v}\,\hat{t} - \hat{t}\,\hat{v} \,=\,0~~,\quad
\hat{v}\,\hat{\theta} - \hat{\theta}\,\hat{v} \,=\, 0~~,\quad\\
\hat{t}\,\hat{t}^\prime + \hat{t}^\prime\,\hat{t} \,&=\,0~~,\quad
\hat{t}\,\hat{\theta} + \hat{\theta}\,\hat{t} \,=\, -\hbar\,\langle\theta,t\rangle ~~,\quad 
\hat{\theta}\,\hat{\theta}^\prime + \hat{\theta}^\prime\,\hat{\theta} \,=\,0\label{eqn:ideal0}
\end{flalign}
\end{subequations}
in $\AAA^{0\sharp}_\hbar$ that quantize the Poisson bracket \eqref{eqn:Poissonlevel0}.
To quantize the classical total differential
$\hat{\dd}=\partial+\delta$ to a differential 
$\hat{\dd}_\hbar$ on $\AAA^{0\sharp}_\hbar$, we proceed in two steps:
First, we observe using \eqref{eqn:chaincotangent} and \eqref{eqn:cochaincotangent}
that the classical differential $\hat{\dd}$ closes on 
$(\Sym \,\g^{\vee[-1]})^\sharp \otimes A[[\hbar]]$ and it is explicitly given by
\begin{flalign}
\hat{\dd}(a) \,=\,\theta^b\,\rho_A^{}(t_b)(a)~~,\quad
\hat{\dd}(\theta^a)\,=\,-\tfrac{1}{2}f^{a}_{bc}\,\theta^b\,\theta^c\quad.
\end{flalign}
Second, regarding the differential $\hat{\dd}$ on $(\Sym \,\g^{\vee[-1]})^\sharp \otimes A[[\hbar]]$
as a $\bbK[[\hbar]]$-linear endomorphism, we can define the quantized differential on
the differential operators through the graded commutator
\begin{flalign}\label{eqn:dquantcommutator}
\hat{\dd}_\hbar\,:=\, [\hat{\dd},\,\cdot\,]
\end{flalign}
of endomorphisms of $(\Sym \,\g^{\vee[-1]})^\sharp \otimes A[[\hbar]]$. 
A short calculation using \eqref{eqn:chaincotangent} and 
\eqref{eqn:cochaincotangent} shows that
\begin{flalign}
\nn \hat{\dd}_\hbar(\hat{a}) \,&=\,\hat{\theta}^b\,\widehat{\rho_A^{}(t_b)(a)}~~,\quad
\hat{\dd}_\hbar(\hat{\theta}^a)\,=\,-\tfrac{1}{2}f^{a}_{bc}\,\hat{\theta}^b\,\hat{\theta}^c\quad,\\
\hat{\dd}_\hbar(\hat{v})\,&=\, \hat{\theta}^b\,\widehat{\rho_{\T_{\! A}}^{}(t_b)(v)}~~,\quad
\hat{\dd}_\hbar(\hat{t})\,=\, \widehat{\mu^\ast(t)} + \hat{\theta}^b\,\widehat{\rho_{\g}^{}(t_b)(t)}\quad.
\end{flalign}
Note that the quantized differential takes formally the same form as 
its classical counterpart in \eqref{eqn:chaincotangent} and \eqref{eqn:cochaincotangent}, 
but this has to be interpreted with some care: In the classical case, the 
underlying graded algebra is commutative, while in the quantum case it is noncommutative.
This means that one faces questions about operator orderings 
when trying naively to promote the classical differential $\hat{\dd}$
in \eqref{eqn:chaincotangent} and \eqref{eqn:cochaincotangent} to a quantized differential $\hat{\dd}_\hbar$.
Our construction above provides a systematic fix for these operator ordering issues.
Summing up, we obtain a noncommutative cochain dg-algebra 
\begin{flalign}\label{eqn:AAA0}
\AAA^{0\bullet}_\hbar \,:=\,\big(\AAA_\hbar^{0\sharp}, \hat{\dd}_\hbar\big)
\end{flalign}
that we interpret as a quantization of $\hat{\Tot}\CE^\bullet(\g,B_\bullet)$
along the Poisson structure \eqref{eqn:Poissonlevel0}.
\begin{rem}
The noncommutative cochain dg-algebra $\AAA^{0\bullet}_\hbar$ is a quantization 
in the following precise sense: Taking the quotient 
$\AAA^{0\bullet}_\hbar/(\hbar)\cong \hat{\Tot}\CE^\bullet(\g,B_\bullet)$ by the ideal
generated by $\hbar$ recovers the classical cochain dg-algebra $\hat{\Tot}\CE^\bullet(\g,B_\bullet)$
on which the normalized commutator $\frac{1}{\hbar}[\,\cdot\,,\,\cdot\,]$ descends to the 
Poisson bracket $\{\,\cdot\,,\,\cdot\,\}_0^{}$ given in \eqref{eqn:Poissonlevel0}.
Note that such quantizations by differential operators are commonly used in quantum mechanics:
For instance, consider a classical mechanical system on the real line $\bbR$ that
is described by the phase space $T^\ast \bbR$ equipped with its canonical Poisson structure.
The corresponding classical algebra is generated by the position $x$ and momentum $p$ variables,
whose Poisson bracket is $\{p,x\} = \oone$. Quantization can be performed by 
introducing the differential operators $\hat{p} = \hbar \frac{d}{dx}$ and $\hat{x}=x$,
whose commutation relation $\hat{p}\,\hat{x} - \hat{x}\,\hat{p} = \hbar$ 
is of the same form as the relations in \eqref{eqn:ideal0}.
\end{rem}

Similarly to our construction above, one can quantize, for each $n\geq 1$,
the cochain dg-algebra $\hat{\Tot}\CE^\bullet(\g^{\oplus n+1},B_\bullet\otimes H^{\otimes n})$
along the unshifted Poisson structure $\{\,\cdot\, ,\,\cdot\,\}_{n}^{}$
given in \eqref{eqn:Poissonleveln1} and \eqref{eqn:Poissonleveln2}. 
Again, this is achieved by interpreting
$v\in\T_{\! A}$ and $t\in \g_{[-1]}$ in terms of derivations on the graded
algebra $(\Sym \,\g^{\vee \oplus n+1 [-1]})^\sharp \otimes A\otimes H^{\otimes n}$
and then defining as in \eqref{eqn:AAAhbar} a quantization $\AAA_\hbar^{n\bullet}$ 
in terms of differential operators. For completeness, let us note that 
the relevant commutation relations 
generalizing \eqref{eqn:ideal0} to the case $n\geq 1$ are given by
\begin{flalign}
\nn \hat{a}\,\hat{a}^\prime - \hat{a}^\prime\,\hat{a}\,&=\, 0~~,\quad 
\hat{v}\,\hat{a} - \hat{a}\,\hat{v} \,=\, \hbar\, \widehat{v(a)} ~~,\quad  
\hat{a}\,\hat{h}^{\an{k}} - \hat{h}^{\an{k}} \,\hat{a}\,=\,0~~,\quad 
\hat{a}\,\hat{t} - \hat{t}\,\hat{a}  \,=\,0~~,\quad \\ \nn 
\hat{a}\,\hat{\theta} - \hat{\theta}\,\hat{a}  \,&=\,0~~,\quad
\hat{v}\,\hat{v}^\prime - \hat{v}^\prime\,\hat{v} \,=\, \hbar \,\widehat{[v,v^\prime]}~~,\quad 
\hat{v}\,\hat{h}^{\an{k}} - \hat{h}^{\an{k}} \,\hat{v}\,=\,0~~,\quad\\ \nn
\hat{v}\,\hat{t} - \hat{t}\,\hat{v} \,&=\,0~~,\quad 
\hat{v}\,\hat{\theta}^{\an{j}} - \hat{\theta}^{\an{j}}\,\hat{v} \,=\, 0~~,\quad
\hat{h}^{\an{k}}\,\hat{h}^{\prime\an{k^\prime}} - \hat{h}^{\prime \an{k^\prime}}\,\hat{h}^{\an{k}}\, =\, 0~~,\quad \\ \nn
\hat{h}^{\an{k}}\, \hat{t} - \hat{t}\, \hat{h}^{\an{k}} \,&=\,0~~,\quad
\hat{h}^{\an{k}}\, \hat{\theta}^{\an{j}} - \hat{\theta}^{\an{j}} \, \hat{h}^{\an{k}} \,=\,0~~,\quad
\hat{t}\,\hat{t}^\prime + \hat{t}^\prime\,\hat{t} \,=\,0~~,\quad \\ 
\hat{t}\,\hat{\theta}^{\an{j}} + \hat{\theta}^{\an{j}}\,\hat{t} \,&=\, -\hbar\,
\langle\theta,t_{\und{0}}^{}\rangle
~\widehat{S(t_{\und{1}}^{})}^{\an{j}}\cdots \widehat{S(t_{\und{j}}^{})}^{\an{1}}~~,\quad
\hat{\theta}^{\an{j}}\, \hat{\theta}^{\prime \an{j^\prime}} + \hat{\theta}^{\prime \an{j^\prime}} \, \hat{\theta}^{\an{j}} \,=\,0\quad, \label{eqn:idealn}
\end{flalign}
for all $j,j^\prime \in\{ 0,\dots,n\}$ and $k,k^\prime\in \{1,\dots,n\}$. To obtain a quantized differential
$\hat{\dd}_\hbar$ on $\AAA_\hbar^{n\bullet}$, we note again
that the classical total differential $\hat{\dd}=\partial+\delta$ 
given in \eqref{eqn:chaincotangent} and \eqref{eqn:cochaincotangent}
closes on $(\Sym \,\g^{\vee \oplus n+1 [-1]})^\sharp \otimes A\otimes H^{\otimes n}[[\hbar]]$.
The quantized differential $\hat{\dd}_\hbar$ is then defined by the graded commutator
\eqref{eqn:dquantcommutator} of endomorphisms and one explicitly finds that
\begin{flalign}
\nn \hat{\dd}_\hbar(\hat{a}) \,&=\,\hat{\theta}^{b\an{0}}\,\widehat{\rho_A^{}(t_b)(a)}~~,\quad
\hat{\dd}_\hbar\big(\hat{\theta}^{a\an{j}}\big)\,=\,-\tfrac{1}{2}f^{a}_{bc}\,\hat{\theta}^{b\an{j}}\,\hat{\theta}^{c\an{j}}~~,\quad\\
\nn \hat{\dd}_\hbar\big(\hat{h}^{\an{k}}\big)\,&=\, \hat{\theta}^{b \an{k-1}}\,\widehat{\rho^{\mathrm{R}}(t_{b})(h)}^{\an{k}}
+ \hat{\theta}^{b \an{k}}\,\widehat{\rho^{\mathrm{L}}(t_{b})(h)}^{\an{k}}\quad,\\
\hat{\dd}_\hbar(\hat{v})\,&=\, \hat{\theta}^{b\an{0}}\,\widehat{\rho_{\T_{\! A}}^{}(t_b)(v)}~~,\quad
\hat{\dd}_\hbar(\hat{t})\,=\, \widehat{\mu^\ast(t)} +\hat{\theta}^{b\an{0}}\,\widehat{\rho_{\g}^{}(t_b)(t)}\quad.
\end{flalign}
This defines a noncommutative cochain dg-algebra
\begin{flalign}\label{eqn:AAAn}
\AAA_\hbar^{n\bullet}\,:=\,\big(\AAA_{\hbar}^{n\sharp}, \hat{\dd}_\hbar\big)
\end{flalign}
quantizing $\hat{\Tot}\CE^\bullet(\g^{\oplus n+1},B_\bullet\otimes H^{\otimes n})$
along the unshifted Poisson structure $\{\,\cdot\, ,\,\cdot\,\}_{n}^{}$.
To obtain quantized analogs of the classical coface and codegeneracy maps
in \eqref{eqn:cosimplicialstackyCDGA}, we use for the generators
the same formulas as in the classical case, see \eqref{eqn:facemapdual}, \eqref{eqn:facemaptildeLie},
\eqref{eqn:degeneracymapdual} and \eqref{eqn:degeneracymaptildeLie}. (Note that there
are no operator ordering ambiguities in defining these maps on generators.)
One easily checks that this is compatible with the commutation relations in
\eqref{eqn:idealn} and with the quantized differentials $\hat{\dd}_\hbar$. 
Summing up, we obtain
\begin{propo}\label{propo:cosimplicialquantum}
The quantizations described in this subsection yield a cosimplicial diagram
\begin{flalign}\label{eqn:cosimplicialquantum}
\xymatrix@C=1em{
\AAA_\hbar^{0\bullet} \ar@<0.5ex>[r] \ar@<-0.5ex>[r]~&~ 
\AAA_\hbar^{1\bullet} \ar@<1ex>[r] \ar[r] \ar@<-1ex>[r] ~&~ 
\AAA_\hbar^{2\bullet}\ar@<0.5ex>[r] \ar@<-0.5ex>[r]\ar@<1.5ex>[r] \ar@<-1.5ex>[r]~&~\cdots
}
\end{flalign}
of noncommutative cochain dg-algebras.
\end{propo}

In the remaining part of this subsection we discuss how to quantize
the dg-categories of perfect modules over the stacky CDGAs 
$\CE^\bullet(\g^{\oplus n+1},B_\bullet\otimes H^{\otimes n})$.
Our approach consists of working out an explicit model for the 
deformation theoretic construction in \cite[Proposition 1.25]{PridhamUnshifted} 
in our example of interest. 
We start again with spelling out the simplest case $n=0$ fully explicitly, 
from which one can easily infer the general construction for all $n\geq 0$. 
Recall from Proposition \ref{prop:CEmodulescharacterization}
that an object in $\mathrm{per}(\CE^\bullet(\g,B_\bullet))$ 
is a left $B_\bullet$-dg-module $V_\bullet\in {}_{B_\bullet}\Mod$
with a compatible $\g$-action $\rho_V^{}:\g\to\End(V_\bullet)$. (We ignore
for the moment the cofibrancy and perfectness properties and will come back to them later.)
In the present case, we have that $B_\bullet$ is the symmetric algebra over $A$ of a chain complex
of $A$-modules (see \eqref{eqn:Bbulletsym}), which allows us to 
decompose the $B_\bullet$-module structure on $V_\bullet$ into the following data:
\begin{itemize}
\item[(1)] a left $A$-dg-module $V_\bullet\in {}_A\Mod$ with a compatible $\g$-action $\rho_V^{}:\g\to\End(V_\bullet)$,
\item[(2)] a $\g$-equivariant left $A$-dg-module map 
$\Theta: V_\bullet \to \Omega^1_A\otimes_A^{} V_\bullet$, and
\item[(3)] a $\g$-equivariant graded left $A$-module map $\Psi : \g_{[-1]}\otimes V_\sharp\to V_\sharp$.
\end{itemize}
These data have to satisfy the following properties:
\begin{itemize}
\item[(i)] For all $v,v^\prime\in\T_{\! A}$ and $t,t^\prime\in\g_{[-1]}$,
\begin{flalign}
\Theta_v \circ \Theta_{v^\prime} \,=\,\Theta_{v^\prime} \circ \Theta_{v}~~,\quad
\Theta_{v}\circ \Psi_{t} \,=\, \Psi_{t}\circ \Theta_{v}~~,\quad
\Psi_t\circ \Psi_{t^\prime} \,=\,- \Psi_{t^\prime}\circ \Psi_{t}\quad,
\end{flalign}
where  $\Theta_v := \iota_v\Theta$ is defined through contraction and $\Psi_t := \Psi(t\otimes -)$. 
\item[(ii)] For all $t\in \g_{[-1]}$,
\begin{flalign}
\partial\circ \Psi_t + \Psi_t \circ \partial \,=\, \Theta_{\mu^\ast(t)}\quad,
\end{flalign}
where $\partial$ denotes the differential on $V_\bullet$ and the moment map 
$\mu^\ast(t)$ arises from \eqref{eqn:Bbulletsym}.
\end{itemize}
Let us rephrase this decomposition 
of structure in the context of Example \ref{ex:CEmodule} and Remark \ref{rem:CEmodulescharacterization}:
The structure maps $\Theta$ and $\Psi$, and their compatibility conditions (i) and (ii),
are precisely what is needed to endow the totalized Chevalley-Eilenberg module 
$\hat{\Tot}\CE^\bullet(\g,V_\bullet)$ associated to $V_\bullet\in {}_A\Mod$
with a left dg-module structure over the totalized Chevalley-Eilenberg algebra
$\hat{\Tot}\CE^\bullet(\g,B_{\bullet})$. 
\sk

In order to quantize the dg-category 
$\mathrm{per}(\CE^\bullet(\g,B_{\bullet}))$, we have to find quantum analogs
of the structure maps $\Theta$ and $\Psi$, and their compatibility conditions,
such that $\hat{\Tot}\CE^\bullet(\g,V_\bullet)$ can be endowed
with a left dg-module structure over the noncommutative cochain dg-algebra
$\AAA_\hbar^{0\bullet}$ quantizing $\hat{\Tot}\CE^\bullet(\g,B_{\bullet})$.
Since $\AAA_\hbar^{0\bullet}$ is an algebra of differential operators, it is not surprising
that this is related to $D$-modules, i.e.\ we expect that $\Theta$ will be quantized to a 
flat connection $\nabla$ on $V_\bullet$. Working out the details, one finds that the 
relevant data are given by:
\begin{itemize}
\item[($1_{\hbar}$)] a left $A[[\hbar]]$-dg-module $V_{\bullet}\in {}_{A[[\hbar]]}\Mod$ 
with a compatible $\g$-action $\rho_V^{}:\g\to\End(V_{\bullet})$,

\item[($2_{\hbar}$)] a $\g$-equivariant dg-connection $\nabla : V_{\bullet}\to 
\Omega^1_{A}[[\hbar]] \otimes_{A[[\hbar]]}^{} V_{\bullet}$ with respect to 
the differential $\hbar\,\dd^{\dR}$, i.e.\
$\nabla$ satisfies the Leibniz rule $\nabla(a\cdot s)=\hbar \,\dd^{\dR}a\otimes_{{A[[\hbar]]}} s + 
a\cdot \nabla(s)$, for all $a\in A[[\hbar]]$ and $s\in V_{\bullet}$,
and it commutes $(\id\otimes \partial)\circ \nabla = \nabla\circ\partial $ with the 
differential $\partial$ on $V_{\bullet}$, and

\item[($3_{\hbar}$)] a $\g$-equivariant graded left $A[[\hbar]]$-module map 
$\Psi : \g_{[-1]}\otimes V_{\sharp}\to V_{\sharp}$.
\end{itemize}
These data have to satisfy the following properties:
\begin{itemize}
\item[($\text{i}_{\hbar}$)] For all $v,v^\prime\in\T_{\! A}$ and $t,t^\prime\in\g_{[-1]}$,
\begin{flalign}
\nabla_v\circ \nabla_{v^\prime} - \nabla_{v^\prime}\circ \nabla_{v} \,=\, \hbar\, \nabla_{[v,v^\prime]}~~,\quad
\nabla_v\circ \Psi_t\,=\, \Psi_t\circ \nabla_v~~,\quad \Psi_t\circ\Psi_{t^\prime} \,=\, - 
\Psi_{t^\prime}\circ\Psi_{t}\quad,
\end{flalign}
where $\nabla_v := \iota_v\nabla$ denotes the covariant derivative along $v\in \T_{\! A}$ that is
associated with the connection $\nabla$. Note that the first condition states that $\nabla$ is a flat connection.

\item[($\text{ii}_{\hbar}$)] For all $t\in\g_{[-1]}$,
\begin{flalign}
\partial\circ \Psi_t + \Psi_t \circ \partial \,=\, \nabla_{\mu^\ast(t)} +\hbar\, \rho_V^{}(t)\quad.
\end{flalign}
\end{itemize}
\begin{rem}
Note that the classical limit $(V_\bullet,\nabla,\Psi)\vert_{\hbar=0}$ 
of such triple $(V_\bullet,\nabla,\Psi)$, which is 
obtained by quotienting by the ideal $(\hbar)\subseteq\bbK[[\hbar]]$,
is precisely the classical data listed in the itemization above.
\end{rem}

Based on these insights, we propose
the following natural quantization of the 
dg-category $\mathrm{per}(\CE^\bullet(\g,B_\bullet))$
associated to the formal derived cotangent stack $T^\ast[X/\g]$. 
It corresponds to the pre-triangulated envelope of the deformation of 
$\mathrm{per}(\CE^\bullet(\g,B_\bullet))$ associated by \cite[Proposition 1.25]{PridhamUnshifted} 
to the quantization $\AAA_\hbar^{0\bullet}$ of $\hat{\Tot}\CE^\bullet(\g,B_{\bullet})$.
\begin{defi}
We define the quantized dg-category $\mathrm{per}_\hbar(\CE^\bullet(\g,B_\bullet))$
to be the dg-category over $\bbK[[\hbar]]$ whose objects are all triples
$(V_\bullet,\nabla,\Psi)$ as introduced above, such that
the classical limit $(V_\bullet,\nabla,\Psi)\vert_{\hbar=0}$ 
defines an object in $\mathrm{per}(\CE^\bullet(\g,B_\bullet))$
and additionally $V_\bullet \cong \varprojlim_n V_\bullet/\hbar^nV_\bullet$
with $V_{\sharp}/\hbar^nV_\sharp$ projective as a graded
$\mathrm{DiffOp}_\hbar (A[[\hbar]])/\hbar^n \otimes \big(\Sym\, \g_{[-1]}^{}\big)_{\sharp}$-module, for all $n\geq 0$.
The cochain complexes (over $\bbK[[\hbar]]$) of morphisms are given by
\begin{flalign}\label{eqn:quantummapcomplexlocal}
\hom_{\AAA_\hbar^{0\bullet}}^{}\big(\hat{\Tot}\CE^\bullet(\g,V_\bullet),\hat{\Tot}\CE^\bullet(\g,V_\bullet^\prime)\big)\quad,
\end{flalign}
for all pairs of objects $(V_\bullet,\nabla,\Psi)$ and $(V_\bullet^\prime,\nabla^\prime,\Psi^\prime)$.
\end{defi}
\begin{rem}
The role of the additional condition on the module $V_\bullet$ 
is to rule out pathological objects in the quantized dg-category 
$\mathrm{per}_\hbar(\CE^\bullet(\g,B_\bullet))$, such as those where
$\hbar$ acts as $0$ or as $1$.
To satisfy all conditions it is sufficient, but not necessary, 
for $V_\sharp$ to be a finitely generated projective 
$\mathrm{DiffOp}_\hbar (A[[\hbar]]) \otimes \big(\Sym\, \g_{[-1]}^{}\big)_{\sharp}$-module.
\end{rem}
\begin{rem}\label{rem:CEmodulequantum}
In analogy to the classical case (see Remark \ref{rem:CEmodulescharacterization}),
we have that
\begin{flalign}
\mathrm{per}_\hbar(\CE^\bullet(\g,B_\bullet))\,\subseteq\, {}_{\AAA_\hbar^{0\bullet}}\Mod
\end{flalign}
can be regarded (by construction) as a full sub-dg-category of the dg-category of dg-modules over the 
quantized algebra $\AAA_\hbar^{0\bullet}$. Each object in this sub-dg-category
is of the form $\hat{\Tot}\CE^\bullet(\g,V_\bullet)$, with the left $\AAA_\hbar^{0\bullet}$-module
structure determined from the data of a triple $(V_\bullet,\nabla,\Psi)$.
\end{rem}

\begin{rem}\label{rem:alternativequantmorphisms}
There exists an alternative description of the cochain complexes
of morphisms in $\mathrm{per}_\hbar(\CE^\bullet(\g,B_\bullet))$ that generalizes
the second isomorphic description in \eqref{eqn:perfhomidentification} to the quantum case
and that emphasizes better our description of objects by triples. By an analogous rewriting exercise
as in the proof of Proposition \ref{prop:CEmodulescharacterization}, one observes
that the cochain complex \eqref{eqn:quantummapcomplexlocal} may be regarded as a sub-complex of
$\hat{\Tot}\CE^\bullet(\g,\hom_{A[[\hbar]]}^{}(V_\bullet,V^\prime_\bullet))$. This sub-complex is characterized
by the following compatibility conditions with $\nabla$ and $\Psi$: An element
\begin{subequations}
\begin{flalign}
L\, =\, \sum_{j=0}^{\dim\,\g} \frac{1}{j!}\,\hat{\theta}^{a_1}\cdots \hat{\theta}^{a_j}\, L_{a_1\cdots a_j}\,\in\,
\hat{\Tot}\CE^\bullet\big(\g,\hom_{A[[\hbar]]}^{}(V_\bullet,V^\prime_\bullet)\big)^p\quad,
\end{flalign}
where $L_{a_1\cdots a_j}\in \hom_{A[[\hbar]]}^{}(V_\bullet,V^\prime_\bullet)_{j-p}$,
defines an element of degree $p$ in  \eqref{eqn:quantummapcomplexlocal} if and only if
\begin{flalign}
\nabla_v^\prime\circ L_{a_1\cdots a_j}  \,=\, L_{a_1\cdots a_j}\circ \nabla_v^{}\quad,
\end{flalign}
for all $v\in\T_{\! A}$, and
\begin{flalign}
\Psi^\prime_{t_b} \circ L_{a_1\cdots a_j}\,=\, (-1)^{p-j}\, L_{a_1\cdots a_j}\circ \Psi_{t_b} + \hbar \, L_{a_1\cdots a_j b}\quad,
\end{flalign}
\end{subequations}
for all basis elements $t_b\in\g_{[-1]}$.
Note that the $\hbar$-contribution to the second condition arises from the graded commutators
$\hat{t}\,\hat{\theta} + \hat{\theta}\,\hat{t} = -\hbar\,\langle\theta,t\rangle$.
This alternative description will become useful in Subsection \ref{subsec:globalquantization}.
\end{rem}

\begin{ex}\label{ex:rank1module}
The rank-one module $\AAA_\hbar^{0\bullet}\in {}_{\AAA_\hbar^{0\bullet}}\Mod$
is (isomorphic to) an object in the quantized dg-category $\mathrm{per}_\hbar(\CE^\bullet(\g,B_\bullet))$.
Let us spell out explicitly a triple $(B_{\hbar\bullet},\nabla,\Psi)$ that represents this object.
From \eqref{eqn:AAAhbar}, we observe that there exists
an isomorphism of graded $\bbK[[\hbar]]$-modules
\begin{flalign}\label{eqn:tmporderingpointing}
\AAA_\hbar^{0\sharp}\,\cong\, \hat{\Tot}\Big(\Sym\,\g^{\vee[-1]} \otimes \mathrm{DiffOp}_\hbar(A[[\hbar]])\otimes \Sym\, \g_{[-1]}\Big)^\sharp
\end{flalign}
given by using the commutation relations in \eqref{eqn:ideal0} to bring
each differential operator to the displayed form, i.e.\ the $\hat{\theta}$'s to the left,
composites of the $\hat{a}$'s and $\hat{v}$'s in the middle, and the $\hat{t}$'s to the right. 
Let us define the graded left $A[[\hbar]]$-module
\begin{flalign}\label{eqn:Bhbarsharp}
B_{\hbar\sharp}\,:=\,\mathrm{DiffOp}_\hbar(A[[\hbar]])\otimes \big(\Sym\, \g_{[-1]}\big)_{\sharp}
\end{flalign}
and endow it with the $\g$-action determined 
by the given $\g$-actions on the individual tensor factors. 
Demanding that \eqref{eqn:tmporderingpointing} defines
an isomorphism $\AAA_\hbar^{0\bullet}\cong \hat{\Tot}\CE^\bullet(\g,B_{\hbar\bullet})$ 
of left $A[[\hbar]]$-dg-modules fixes uniquely a chain differential
$\partial$ on $B_{\hbar\sharp}$. Explicitly, one finds at the level
of elements $\hat{D}\otimes \hat{t}_{1}\cdots \hat{t}_n\in \mathrm{DiffOp}_\hbar(A[[\hbar]])
\otimes\big( \Sym\, \g_{[-1]}\big)_\sharp$ that
\begin{flalign}\label{eqn:Bhbardifferential}
\partial\big(\hat{D}\otimes \hat{t}_1\cdots \hat{t}_n\big)\,=\, \sum_{i=1}^n (-1)^{i-1}\Big( 
\hat{D}\,\widehat{\mu^{\ast}(t_i)}\otimes \hat{t}_1\cdots \check{t_{i}}\cdots \hat{t}_n 
+ \hat{D}\otimes\big[\hat{t}_1\cdots \hat{t}_{i-1}, \hat{\theta}^b\big]\,\widehat{\rho_\g^{}(t_b)(t_i)}\,\cdots \hat{t}_n\Big)\quad,
\end{flalign}
where $\check{-}$ denotes omission of the factor and 
the bracket $[\,\cdot\,,\,\cdot\,]$ denotes the graded commutator of differential operators,
see \eqref{eqn:ideal0} for the relevant commutation relations.
One further checks that $\AAA_\hbar^{0\bullet}\cong \hat{\Tot}\CE^\bullet(\g,B_{\hbar\bullet})$  
becomes an isomorphism of left $\AAA_\hbar^{0\bullet}$-dg-modules
if we equip $\hat{\Tot}\CE^\bullet(\g,B_{\hbar\bullet})$ with the module structure
resulting from the connection $\nabla_v(\hat{D}\otimes
\hat{t}_1\cdots \hat{t}_n)= \hat{v}\,\hat{D}\otimes \hat{t}_1\cdots \hat{t}_n$ 
and the map $\Psi_t(\hat{D}\otimes \hat{t}_1\cdots \hat{t}_n) = \hat{D}\otimes\hat{t}\,\hat{t}_1\cdots \hat{t}_n$ given
by multiplications of differential operators. This defines
a triple $(B_{\hbar\bullet},\nabla,\Psi)$ representing, up to the isomorphism
constructed in this example, the rank-one module $\AAA_\hbar^{0\bullet}\in {}_{\AAA_\hbar^{0\bullet}}\Mod$.
\end{ex}

The construction above generalizes in the obvious way to all $n\geq 0$. 
This allows us to define quantized dg-categories
$\mathrm{per}_\hbar\big(\CE^\bullet(\g^{\oplus n+1},B_\bullet\otimes H^{\otimes n})\big)$
over $\bbK[[\hbar]]$ which, in analogy to Remark \ref{rem:CEmodulequantum}, can be presented
as full sub-dg-categories 
\begin{flalign}
\mathrm{per}_\hbar\big(\CE^\bullet(\g^{\oplus n+1},B_\bullet\otimes H^{\otimes n})\big)
\,\subseteq\, {}_{\AAA_\hbar^{n\bullet}}\Mod
\end{flalign}
of the dg-categories of dg-modules over the quantized algebras $\AAA_\hbar^{n\bullet}$.
Via the induced module functors associated with the coface and the codegeneracy maps
in \eqref{eqn:cosimplicialquantum}, we obtain a cosimplicial diagram
\begin{flalign}\label{eqn:cosimplicialdgCat}
\xymatrix@C=1em{
\mathrm{per}_\hbar\big(\CE^\bullet(\g,B_\bullet)\big) \ar@<0.5ex>[r] \ar@<-0.5ex>[r]~&~ 
\mathrm{per}_\hbar\big(\CE^\bullet\big(\g^{\oplus 2}, B_\bullet\otimes H\big)\big) \ar@<1ex>[r] \ar[r] \ar@<-1ex>[r]~&~ 
\mathrm{per}_\hbar\big(\CE^\bullet\big(\g^{\oplus 3}, B_\bullet\otimes H^{\otimes 2}\big)\big)\ar@<0.5ex>[r] \ar@<-0.5ex>[r]\ar@<1.5ex>[r] \ar@<-1.5ex>[r]~&~ 
\cdots
}
\end{flalign}
of dg-categories over $\bbK[[\hbar]]$, which quantizes the cosimplicial diagram in 
Definition \ref{def:perquotientstack}.


\subsection{\label{subsec:globalquantization}Global quantization}
The quantization of the dg-category $\mathrm{per}(T^\ast[X/G])$ of perfect modules
over the derived cotangent stack $T^\ast[X/G]$ is defined as the homotopy limit
\begin{flalign}\label{eqn:quantizeddgCATholim}
\mathrm{per}_\hbar(T^\ast[X/G])\,:=\,\holim\Big(
\xymatrix@C=1em{
\mathrm{per}_\hbar\big(\CE^\bullet(\g,B_\bullet)\big) \ar@<0.5ex>[r] \ar@<-0.5ex>[r]~&~ 
\mathrm{per}_\hbar\big(\CE^\bullet\big(\g^{\oplus 2}, B_\bullet\otimes H\big)\big) \ar@<1ex>[r] \ar[r] \ar@<-1ex>[r]~&~ 
\cdots
}\Big)
\end{flalign}
of the cosimplicial diagram \eqref{eqn:cosimplicialdgCat} of local quantizations
that we have obtained in Subsection \ref{subsec:localquantization} in terms of 
a resolution by stacky CDGAs. 
\sk

Before we provide an explicit model for 
this dg-category over $\bbK[[\hbar]]$, it will be instructive to 
describe in some detail the classical dg-category $\mathrm{per}(T^\ast[X/G])$.
For this we assume that $G=\spec\, H$ is reductive, see also Remark \ref{rem:reductive}.
Specializing Proposition \ref{prop:dgCatquotientstack} to the present case,
we obtain, in complete analogy to the decomposition of data explained in the
previous subsection, that the data specifying an object in $\mathrm{per}(T^\ast[X/G])$
may be decomposed as follows:
\begin{itemize}
\item[(1)] a left $A$-dg-module $V_\bullet\in {}_A\Mod$ with a compatible $H$-coaction $\rho_V^{}:V_\bullet \to V_\bullet \otimes H$,
\item[(2)] an $H$-equivariant left $A$-dg-module map 
$\Theta: V_\bullet \to \Omega^1_A\otimes_A^{} V_\bullet$, and
\item[(3)] an $H$-equivariant graded left $A$-module map $\Psi : \g_{[-1]}\otimes V_\sharp\to V_\sharp$.
\end{itemize}
These data have to satisfy the following properties:
\begin{itemize}
\item[(i)] For all $v,v^\prime\in\T_{\! A}$ and $t,t^\prime\in\g_{[-1]}$,
\begin{flalign}
\Theta_v \circ \Theta_{v^\prime} \,=\,\Theta_{v^\prime} \circ \Theta_{v}~~,\quad
\Theta_{v}\circ \Psi_{t} \,=\, \Psi_{t}\circ \Theta_{v}~~,\quad
\Psi_t\circ \Psi_{t^\prime} \,=\,- \Psi_{t^\prime}\circ \Psi_{t}\quad.
\end{flalign}
\item[(ii)] For all $t\in \g_{[-1]}$,
\begin{flalign}
\partial\circ \Psi_t + \Psi_t \circ \partial \,=\, \Theta_{\mu^\ast(t)}\quad.
\end{flalign}
\end{itemize}
The left $B_\bullet$-dg-module that is canonically associated 
with a triple $(V_\bullet,\Theta,\Psi)$ is further required to be cofibrant and perfect.
For two objects $(V_\bullet,\Theta,\Psi)$ and $(V^\prime_\bullet,\Theta^\prime,\Psi^\prime)$,
the cochain complex of morphisms in $\mathrm{per}(T^\ast[X/G])$ is given by
the sub-complex
\begin{subequations}
\begin{flalign}
\hom_A^{H,\Theta,\Psi}\big(V_\bullet,V^\prime_\bullet\big)\,\subseteq\, \hom_A^{}(V_\bullet,V^\prime_\bullet)
\end{flalign}
of left $A$-module morphisms that are strictly $H$-equivariant
and that commute with $\Theta$ and $\Psi$, i.e.\
$L\in\hom_A^{}(V_\bullet,V^\prime_\bullet)$ lies in $\hom_A^{H,\Theta,\Psi}\big(V_\bullet,V^\prime_\bullet\big)$ 
if and only if
\begin{flalign}\label{eqn:classicalcompatibilityHOM}
\rho_{V^\prime}^{}\circ L \,=\, (L\otimes\id )\circ \rho_V^{}~~,\quad
\Theta_v^\prime \circ L \,=\,  L\circ \Theta_v~~,\quad
\Psi_t^\prime \circ L \,=\, (-1)^{\vert L\vert}\, L \circ \Psi_t\quad,
\end{flalign}
\end{subequations}
for all $v\in\T_{\! A}$ and $t\in\g_{[-1]}$.
\sk

Inspired by our local quantizations from Subsection \ref{subsec:localquantization},
it is natural to expect that such triples quantize to the following data:
\begin{itemize}
\item[($1_{\hbar}$)] a left $A[[\hbar]]$-dg-module $V_{\bullet}\in {}_{A[[\hbar]]}\Mod$ 
with a compatible $H$-coaction $\rho_V^{}:V_\bullet \to V_\bullet\otimes H$,

\item[($2_{\hbar}$)] an $H$-equivariant dg-connection $\nabla : V_{\bullet}\to 
\Omega^1_{A}[[\hbar]] \otimes_{A[[\hbar]]}^{} V_{\bullet}$ with respect to $\hbar\,\dd^{\dR}$, and

\item[($3_{\hbar}$)] an $H$-equivariant graded left $A[[\hbar]]$-module map 
$\Psi : \g_{[-1]}\otimes V_{\sharp}\to V_{\sharp}$.
\end{itemize}
These data have to satisfy the following properties:
\begin{itemize}
\item[($\text{i}_{\hbar}$)] For all $v,v^\prime \in\T_{\! A}$ and $t,t^\prime\in\g_{[-1]}$,
\begin{flalign}
\nabla_v\circ \nabla_{v^\prime} - \nabla_{v^\prime}\circ \nabla_{v} \,=\, \hbar\, \nabla_{[v,v^\prime]}~~,\quad
\nabla_v\circ \Psi_t\,=\, \Psi_t\circ \nabla_v~~,\quad \Psi_t\circ\Psi_{t^\prime} \,=\, - 
\Psi_{t^\prime}\circ\Psi_{t}\quad.
\end{flalign}

\item[($\text{ii}_{\hbar}$)] For all $t\in\g_{[-1]}$,
\begin{flalign}
\partial\circ \Psi_t + \Psi_t \circ \partial \,=\, \nabla_{\mu^\ast(t)} +\hbar\, \rho_V^{}(t)\quad,
\end{flalign}
where $\rho_V^{}(t)$ denotes the $\g$-action induced by the $H$-coaction $\rho_V^{}$. 
(Explicitly, $\rho_V^{}(t)(s):= s_{\und{0}} \, t(s_{\und{1}})$, for all $s\in V_\bullet$.)
\end{itemize}

\begin{rem}
Writing $\Omega^{\bullet}_{\hbar}(A):=(\Omega^\sharp_A[[\hbar]], \hbar \,\dd^{\dR})$ 
and similarly $\Omega^{\bullet}_{\hbar}(H):=(\Omega^\sharp_H[[\hbar]], \hbar\, \dd^{\dR})$, 
observe that, for any object $(V_\bullet,\nabla,\Psi)$ of $\mathrm{per}_\hbar(T^\ast[X/G])$, 
we have an $\Omega^{\bullet}_{\hbar}(A)$-module
\begin{flalign}
\Omega^{\bullet}_{\hbar}(A,V)_\nabla \,:=\, \big(\hat{\Tot}\big(\Omega^\bullet_\hbar(A)\otimes_{A[[\hbar]]}^{} V_{\bullet}\big)^\sharp, 
\hat{\dd} + \nabla\big) \quad,
\end{flalign}
whose differential combines the total differential $\hat{\dd} =  \partial + \hbar \,\dd^\dR$ 
and the flat connection $\nabla$, that may be equipped with a compatible $\Omega^{\bullet}_{\hbar}(H)$-coaction
\begin{flalign}
\Omega^{\bullet}_{\hbar}(A,V)_\nabla ~\longrightarrow ~ \Omega^{\bullet}_{\hbar}(A,V)_\nabla 
\otimes_{\bbK[[\hbar]]} \Omega^\bullet_\hbar(H) 
\end{flalign}
constructed out of $\rho$ and $\Psi$. 
This provides a concise characterization of the conditions ($\text{i}_{\hbar}$) and ($\text{ii}_{\hbar}$).
\end{rem}
\begin{rem}\label{rem:hodgestack}
On inverting $\hbar$, the resulting $A((\hbar))$-module $V_\bullet[\hbar^{-1}]$ 
has a flat connection $\hbar^{-1}\nabla$, and hence a $\mathcal{D}_A((\hbar))$-module structure. 
The description above then immediately gives $V_\bullet[\hbar^{-1}]$ the structure of a $\bbK((\hbar))$-linear 
$\mathcal{D}$-module on $[X/G]$ in the sense of \cite[\S 6.2.2]{DrinfeldGaitsgoryFinAlgStacks}. 
However, inverting $\hbar$ destroys a lot of information (see Example \ref{ex:Gmex} below); 
whereas $\mathcal{D}$-modules correspond to quasi-coherent sheaves on the de Rham stack 
$[X/G]_{\mathrm{DR}}$ as in \cite{GaitsgoryRozenblyumCrystals,GaitsgoryRozenblyumBook}, our structure on $V_\bullet$ 
corresponds to that of a quasi-coherent sheaf on the $2$-stack 
$[X/G]_{\mathrm{Hod}}:=[X_{\mathrm{Hod}}/G_{\mathrm{Hod}}]$, where $X_{\mathrm{Hod}}$ 
is the Hodge stack  of \cite[\S 7]{simpsonhodgefil}.    
\end{rem}

Recalling also Remark \ref{rem:alternativequantmorphisms}, we propose the following
candidate for the quantized dg-category.
\begin{defi}\label{def:quantizeddgCat}
We define $\mathrm{per}_\hbar(T^\ast[X/G])$ to be the dg-category over $\bbK[[\hbar]]$
whose objects are all triples $(V_\bullet,\nabla,\Psi)$ as introduced above, such that
the classical limit $(V_\bullet,\nabla,\Psi)\vert_{\hbar=0}$ 
defines an object in $\mathrm{per}(T^\ast[X/G])$
and additionally $V_\bullet \cong \varprojlim_n V_\bullet/\hbar^nV_\bullet$
with $V_{\sharp}/\hbar^nV_\sharp$ projective as a graded
$\mathrm{DiffOp}_\hbar (A[[\hbar]])/\hbar^n \otimes \big(\Sym\, \g_{[-1]}^{}\big)_{\sharp}$-module, for all $n\geq 0$.
For two objects $(V_\bullet,\nabla,\Psi)$ and $(V^\prime_\bullet,\nabla^\prime,\Psi^\prime)$,
the cochain complex (over $\bbK[[\hbar]]$) of morphisms is given by the sub-complex
\begin{subequations}
\begin{flalign}
\hom_{A[[\hbar]]}^{H,\nabla,\Psi}\big(V_\bullet,V^\prime_\bullet\big)\,\subseteq\, \hom_{A[[\hbar]]}^{}(V_\bullet,V^\prime_\bullet)
\end{flalign}
of left $A[[\hbar]]$-module morphisms that are strictly $H$-equivariant
and that commute with $\nabla$ and $\Psi$, i.e.\
$L\in\hom_{A[[\hbar]]}^{}(V_\bullet,V^\prime_\bullet)$ lies in 
$\hom_{A[[\hbar]]}^{H,\nabla,\Psi}\big(V_\bullet,V^\prime_\bullet\big)$ 
if and only if
\begin{flalign}
\rho_{V^\prime}^{}\circ L \,=\, (L\otimes\id )\circ \rho_V^{}~~,\quad
\nabla_v^\prime \circ L \,=\,  L\circ \nabla_v~~,\quad
\Psi_t^\prime \circ L \,=\, (-1)^{\vert L\vert}\, L \circ \Psi_t\quad,
\end{flalign}
\end{subequations}
for all $v\in\T_{\! A}$ and $t\in\g_{[-1]}$.
\end{defi}

\begin{propo}\label{prop:quantdgCatquotientstack}
Suppose that $G = \spec H$ is reductive.
Then the dg-category over $\bbK[[\hbar]]$ introduced in Definition \ref{def:quantizeddgCat}
is a model for the homotopy limit in \eqref{eqn:quantizeddgCATholim}.
\end{propo}
\begin{proof}
We adapt the proof of Proposition \ref{prop:dgCatquotientstack} to 
the homotopy limit \eqref{eqn:quantizeddgCATholim}. Let us temporarily denote 
the dg-category from Definition \ref{def:quantizeddgCat} by $\mathcal{C}$, and the
analogous dg-categories defined via modules over the truncated rings $\bbK[\hbar]/\hbar^m$ by $\mathcal{C}_m$.
\sk

Given any object  $(V_\bullet,\nabla,\Psi) \in \mathcal{C}$,  we
observe that we have $\hat{\Tot}\CE^\bullet(\g,V_\bullet)  \in \mathrm{per}_\hbar\big(\CE^\bullet(\g,B_\bullet)\big)$, 
together with an isomorphism  $\theta : d_0^\ast(\hat{\Tot}\CE^\bullet(\g,V_\bullet)) \cong d_1^\ast(\hat{\Tot}\CE^\bullet(\g,V_\bullet))$ 
in $\mathrm{per}_\hbar\big(\CE^\bullet\big(\g^{\oplus 2}, B_\bullet\otimes H\big)\big)$ satisfying the 
cocycle condition $d_1^\ast\theta \cong d_2^\ast\theta \circ d_0^\ast\theta$. 
This gives us a dg-functor from $\mathcal{C}$ to the $2$-categorical limit of the cosimplicial diagram
in \eqref{eqn:quantizeddgCATholim} and hence  a dg-functor $\mathcal{C} \to \mathrm{per}_\hbar(T^\ast[X/G])$
to the homotopy limit.
\sk

To see that this is fully faithful, observe that on morphisms we are looking at the map 
\begin{flalign}
\resizebox{1\hsize}{!}{$\hom_{A[[\hbar]]}^{H,\nabla,\Psi}\big(V_\bullet,V^\prime_\bullet\big) ~\to~ 
\holim\Big(
\xymatrix@C=1em{\hat{\Tot}\CE^\bullet\big(\g,  \hom_{A[[\hbar]]}^{\nabla,\Psi}(V_\bullet,V_\bullet^\prime)\big) \ar@<0.5ex>[r] \ar@<-0.5ex>[r]~&~ 
\hat{\Tot}\CE^\bullet\big(\g^{\oplus 2}, \hom_{A[[\hbar]]}^{\nabla,\Psi}(V_\bullet, V^\prime_\bullet)\otimes H\big)  
~\cdots
}\Big)\quad,$}
\end{flalign}
where we have used the description of the morphism complexes from Remark \ref{rem:alternativequantmorphisms}.
We can calculate the homotopy limit as a normalized total complex, 
and the factors involving $\g$ make an acyclic contribution, giving us a 
quasi-isomorphism from the homotopy limit to 
$ \hat\Tot N^\bullet(G,\hom_{A[[\hbar]]}^{\nabla,\Psi}(V_\bullet, V^\prime_\bullet))$. 
Since $G$ is reductive, this is quasi-isomorphic to its subcomplex 
$\hom_{A[[\hbar]]}^{H,\nabla,\Psi}(V_\bullet, V^\prime_\bullet)$ of $G$-invariants, 
giving the required quasi-isomorphism on $\hom$-complexes.
\sk

We now turn our attention to essential surjectivity. 
In order to use the same conventions as in the proof
of Proposition \ref{prop:dgCatquotientstack}, let us introduce the 
{\em chain} dg-algebras $\AAA_{\hbar\bullet}^{n}$ 
via the usual degree-reflection of the {\em cochain} dg-algebras
$\AAA_{\hbar}^{n\bullet}$, i.e.\ $(\AAA_{\hbar}^{n})_i := (\AAA_{\hbar}^{n})^{-i}$.
In order to appeal to the strictification 
theorem \cite[Corollary 18.7]{HirschowitzSimpson}, we use the projective model structure of 
the second kind, as in \cite[Theorem 8.3(a)]{positselski}, applied to the categories of 
$\AAA_{\hbar\bullet}^{n}/\hbar^m$-modules. 
This model structure has the crucial property that modules are cofibrant 
whenever the underlying graded module is projective.  
Beware that the notion of weak equivalence 
is stronger than quasi-isomorphism, but for cofibrant objects it coincides with 
homotopy equivalence, or in our case with morphisms $M_\bullet \to N_{\bullet}$ which induce quasi-isomorphisms 
\begin{flalign}\label{eqn:tmpinducedmod}
\big( B_\bullet \otimes H^{\otimes n}\big)\otimes_{\AAA_{\hbar\bullet}^{n}/\hbar}  \big(M_\bullet/\hbar\big)
~\longrightarrow~\big(B_\bullet \otimes H^{\otimes n}\big)\otimes_{\AAA_{\hbar\bullet}^{n}/\hbar} \big(N_{\bullet}/\hbar\big)\quad.
\end{flalign}
The dg-algebra map
$\AAA_{\hbar\bullet}^{n}/\hbar = \hat{\Tot}\CE^\bullet\big(\g^{\oplus n+ 1}, B_\bullet \otimes H^{\otimes n}\big)\to
B_\bullet \otimes H^{\otimes n}$ entering the induced module construction in 
\eqref{eqn:tmpinducedmod} is given by sending all Chevalley-Eilenberg 
generators to zero.
\sk

Replacing modules over  $\bbK[[\hbar]]$ with modules over $\bbK[\hbar]/\hbar^m$ in the homotopy limit 
\eqref{eqn:quantizeddgCATholim} gives us a dg-category  $\mathrm{per}_\hbar(T^\ast[X/G]/\hbar^m)$, 
which by strictification can be characterized as 
the dg-category $\mathrm{per}_\hbar(T^\ast[X/G]/\hbar^m)_{\mathrm{cart}}$ of those left modules 
\begin{flalign}
 \xymatrix@C=1em{
M^0_{\bullet} \ar@<0.5ex>[r] \ar@<-0.5ex>[r]~&~ 
M^1_{\bullet} \ar@<1ex>[r] \ar[r] \ar@<-1ex>[r] ~&~ 
M^2_{\bullet} \ar@<1.5ex>[r] \ar@<0.5ex>[r]\ar@<-0.5ex>[r] \ar@<-1.5ex>[r] ~&~ 
\cdots
}
\end{flalign}
in cosimplicial complexes over the cosimplicial noncommutative dg-algebra
\begin{flalign}
\xymatrix@C=1em{\AAA_{\hbar\bullet}^{0}/\hbar^m
\ar@<0.5ex>[r] \ar@<-0.5ex>[r]~&~ 
\AAA_{\hbar\bullet}^{1}/\hbar^m\ar@<1ex>[r] \ar[r] \ar@<-1ex>[r] ~&~ 
\AAA^2_{\hbar \bullet}/\hbar^m \ar@<1.5ex>[r] \ar@<0.5ex>[r]\ar@<-0.5ex>[r] \ar@<-1.5ex>[r] ~&~ 
\cdots
}
\end{flalign}
which
\begin{enumerate}
\item are projective as cosimplicial graded modules,
\item have  $B_{\bullet}\otimes_{\AAA_{\hbar\bullet}^{0}/\hbar}  (M^0_\bullet/\hbar)$ 
cofibrant and perfect as a left $B_{\bullet}$-dg-module, and
\item are homotopy-Cartesian in the sense that the morphisms
\begin{flalign}
 d^j \,:\, \AAA_{\hbar\bullet}^{n+1}/\hbar^m \otimes_{ \AAA_{\hbar\bullet}^{n}/\hbar^m} M_\bullet^n ~\longrightarrow M_\bullet^{n+1}
\end{flalign}
from the induced modules are all homotopy equivalences.
\end{enumerate}

For any homotopy-Cartesian module 
$M_{\bullet}^{\bullet} \in \mathrm{per}_\hbar(T^\ast[X/G]/\hbar^m)_{\mathrm{cart}}$, 
the homology groups 
$\mathrm{H}_i\big((B_\bullet \otimes H^{\otimes n})\otimes_{\AAA_{\hbar\bullet}^{n}/\hbar} ( M^n_\bullet/\hbar) \big)$ 
form a Cartesian cosimplicial module, so for $n=0$ this gives a 
$G$-equivariant $B_\bullet$-module which determines the whole diagram.  
By lifting $G$-equivariant projective modules from $B_\bullet $ 
to $ \AAA_{\hbar\bullet}^{0}/\hbar^m$ 
and taking cones, we can then show by induction on the lowest non-zero homology group of 
$B_\bullet \otimes_{\AAA_{\hbar\bullet}^{0}/\hbar}  (M^0_\bullet/\hbar)$ 
that every object is homotopy equivalent to one in the essential image of $\mathcal{C}_m$.
\sk

Finally, observe that $\mathcal{C}=\varprojlim_m \mathcal{C}_m$, and that the 
morphisms $\mathcal{C}_m \to \mathcal{C}_{m-1}$ are all fibrations of 
dg-categories in the model structure of \cite{tabuadaMCdgcat}, so the limit is a 
homotopy limit. 
Similarly, we have 
$\mathrm{per}_\hbar\big(\CE^\bullet(\g^{\oplus n+1},B_\bullet\otimes H^{\otimes n})\big) 
\simeq \holim_m\mathrm{per}_\hbar\big(\CE^\bullet(\g^{\oplus n+1},B_\bullet\otimes H^{\otimes n})\big/\hbar^m)$, 
so we have shown
\begin{flalign}
\nn  \mathcal{C} \,&\simeq\, \holim_m^{}\,\mathcal{C}_m \\
\nn  \,&\simeq\,  \holim_m^{}\holim_{n \in \Delta}^{}\, \mathrm{per}_\hbar\big(\CE^\bullet(\g^{\oplus n+1},B_\bullet\otimes H^{\otimes n})\big/\hbar^m)\\
\nn  \,&\simeq\, \holim_{n \in \Delta}^{}\, \mathrm{per}_\hbar\big(\CE^\bullet(\g^{\oplus n+1},B_\bullet\otimes H^{\otimes n})\big)\\
\,&\simeq\, \mathrm{per}_\hbar(T^\ast[X/G])\quad,
\end{flalign}
as required.
\end{proof}

\begin{ex}\label{ex:globalrank1module}
The object $(B_{\hbar\bullet},\nabla,\Psi)\in \mathrm{per}_\hbar(\CE^\bullet(\g,B_\bullet))$ 
constructed in Example \ref{ex:rank1module} can be upgraded to an object in $\mathrm{per}_\hbar(T^\ast[X/G])$
by endowing $B_{\hbar\bullet}$ given in \eqref{eqn:Bhbarsharp} with the tensor product $H$-coaction.
The resulting object $(B_{\hbar\bullet},\nabla,\Psi)\in \mathrm{per}_\hbar(T^\ast[X/G])$ plays the role of a pointing 
(i.e.\ an $E_0$-monoidal structure) on the dg-category $\mathrm{per}_\hbar(T^\ast[X/G])$ that quantizes the 
symmetric monoidal structure on the classical dg-category $ \mathrm{per}(T^\ast[X/G])$. In fact,
the classical limit $(B_{\hbar\bullet},\nabla,\Psi)\vert_{\hbar =0}$ is the monoidal unit of $ \mathrm{per}(T^\ast[X/G])$.
\end{ex}

\begin{ex}\label{ex:Gmex}
As a simple illustration, let us sketch the quantized dg-category
$\mathrm{per}_\hbar(T^\ast[X/G])$ for the case where $X = \mathrm{pt} = \spec\,\bbK$ 
is a point and $G = \mathbb{G}_m = \spec\,\bbK[x^{\pm1}]$ is the $1$-dimensional torus.
The $\bbK[x^{\pm1}]$-coaction on the $\bbK[[\hbar]]$-dg-module $V_\bullet$ 
defines a weight grading $V_\bullet = \bigoplus_{n\in \bbZ} V_\bullet^{(n)}$
such that $\rho(s) = s\otimes x^n$, for all $s\in V_\bullet^{(n)}$. 
Since $\Omega^1_{\bbK}[[\hbar]] = 0$,
the datum of a connection $\nabla$ is necessarily $0$ in this case. Furthermore,
since $\g = \bbK$, the $\Psi$-map specializes to a $\bbK[[\hbar]]$-linear 
map $\Psi : V_\sharp \to V_{\sharp+1}$ of degree $1$ that
preserves the weight grading and squares $\Psi^2 =0$ to zero by ($\text{i}_{\hbar}$). 
The condition ($\text{ii}_{\hbar}$) then reads on elements $s\in V_\bullet^{(n)}$
of weight $n\in\bbZ$ as $\partial\Psi(s) + \Psi(\partial s)  =\hbar\, n\,s$.
\sk

This dg-category is generated by the following objects:
For each $n\in \bbZ$, consider $(V_\bullet,0,\Psi)\in\mathrm{per}_\hbar(T^\ast[X/G])$
whose underlying graded $\bbK[[\hbar]]$-module is
$\big(\Sym\,\g_{[-1]}\big)[[\hbar]]_\sharp \cong
\bbK[[\hbar]] \oplus t\,\bbK[[\hbar]]$ in weight $n$ with differential $\partial(a+ t\,b) = \hbar\, n\, b$
and $\Psi(a + t\,b) = t\,a$, for all $a,b\in\bbK[[\hbar]]$.
Note that $n=0$ gives the pointing from Example \ref{ex:globalrank1module}.
The morphism complexes between any two generating objects 
with different weights is $0$, and the endomorphism complex
of the generating object of weight $n$ is the CDGA
$\big(\Sym\,\g_{[-1]}\big)[[\hbar]]_\sharp$ with differential
$\partial(a+ t\,b) = \hbar\, n\, b$, for all $a,b\in\bbK[[\hbar]]$.
\sk

Note that on inverting $\hbar$ as in Remark \ref{rem:hodgestack}, 
these endomorphism complexes become acyclic for all $n\ne 0$, 
so $\mathrm{per}_\hbar(T^\ast[X/G])$ has a far richer structure than
the category of $\mathcal{D}$-modules on $[X/G]$.
\end{ex}


\section{\label{sec:lattice}Gauge theory on directed graphs}
We apply the constructions and results of the previous section
to study the quantization of gauge theories on directed graphs.
In Subsection \ref{subsec:model}, we introduce our gauge theoretic
model of interest. On each directed graph, the phase space of this model is 
a derived cotangent stack, together with its canonical unshifted Poisson structure.
Hence, it can be quantized by our techniques developed in Section \ref{sec:cotangent}
and we spell out this quantization explicitly in Subsection \ref{subsec:quantization}.
In Subsection \ref{subsec:factorizationproducts}, we show that these quantizations can be endowed
with a prefactorization algebra structure associated with suitable pairwise disjoint
embeddings of directed graphs.


\subsection{\label{subsec:model}The model}
Recall that a (finite) {\em directed graph} $U := (\s,\t: \E \rightrightarrows \V)$ consists
of a finite set $\V$ of vertices, a finite set $\E$ of edges and two maps $\s,\t : \E\rightrightarrows \V$
assigning to each edge its source and its target vertex. One may visualize
an edge $e\in\E$ by an arrow $\s(e)\stackrel{e}{\longrightarrow}\t(e)$.
We interpret a directed graph $U = (\s,\t: \E \rightrightarrows \V)$ as a discrete approximation
of a manifold: The vertices $v\in\V$ correspond to points and the edges
$e\in\E$ correspond to paths between points. 
\sk

One can decorate each directed graph $U = (\s,\t: \E \rightrightarrows \V)$ 
with gauge theoretic data, see e.g.\ \cite{Baez}. 
Let $G=\spec\, H$ be a smooth affine group scheme.
The space of {\em $G$-connections} on $U$ is defined as the product
\begin{flalign}\label{eqn:connections}
\Con_G^{}(U)\,:=\,\prod_{e\in \E} G
\end{flalign}
of affine schemes over the set of edges $\E$.
This is interpreted as assigning to each edge the datum of a parallel transport. 
The {\em gauge group} on $U$ is defined as the product
\begin{flalign}\label{eqn:gaugegroup}
\Gau(U)\,:=\, \prod_{v\in \V} G
\end{flalign}
of affine group schemes, i.e.\ we assign to each vertex a copy of $G$.
The action of the gauge group on the space of $G$-connections is defined by
\begin{flalign}\label{eqn:gaugeaction}
\Con_G^{}(U)\times \Gau(U)~\longrightarrow~\Con_G^{}(U)~~,\quad
\big(\{A_e\}_{e\in\E}, \{g_v\}_{v\in\V}\big)~\longmapsto~
\big\{g_{\t(e)}^{-1}\,A_e\,g_{\s(e)}^{}\big\}_{e\in \E}\quad,
\end{flalign}
where $g_{\t(e)}^{-1}\,A_e\,g_{\s(e)}^{}$ denotes group multiplications.
This is precisely the way in which parallel transports transform under gauge transformations.
\sk

The canonical phase space of the gauge theory is given by the derived cotangent stack
\begin{flalign}\label{eqm:phasespace}
T^\ast\big[\Con_G^{}(U)/\Gau(U)\big]
\end{flalign}
of the quotient stack of $G$-connections modulo gauge transformations on $U$,
together with its canonical unshifted Poisson structure.
As explained in Subsection \ref{subsec:sympred}, this derived stack
admits a concrete description in terms of derived symplectic reduction,
which we shall spell out now explicitly in the algebraic language. 
First, let us note that the function Hopf 
algebra of the gauge group \eqref{eqn:gaugegroup} is the tensor Hopf algebra
\begin{flalign}\label{eqn:Ogaugegroup}
\O(\Gau(U))\,=\,\bigotimes_{v\in\V }H \quad,
\end{flalign}
where $H=\O(G)$ is the function Hopf algebra on $G=\spec\, H$.
The Lie algebra of $\Gau(U)$ and its dual can be identified with direct sums 
\begin{flalign}\label{eqn:gaugeLiealgebra}
\g(U)\,\cong\,\bigoplus_{v\in\V} \g\quad,\qquad \g^\vee(U)\,\cong\ \bigoplus_{v\in \V} \g^\vee
\end{flalign}
of the Lie algebra $\g$ of $G=\spec\, H$ and its dual $\g^\vee$.
The cotangent bundle of the space of $G$-connections can be written as
\begin{flalign}
T^\ast \Con_G^{}(U) \,=\,\prod_{e\in\E} T^\ast G\,\cong\, \prod_{e\in \E}\big( \g^\vee \times G\big)\quad,
\end{flalign}
hence its function algebra is given by the tensor algebra
\begin{flalign}\label{eqn:OTCon}
\O(T^\ast \Con_G^{}(U))\,\cong\, \bigotimes_{e\in\E } \big( \Sym\,\g \otimes H\big)\quad.
\end{flalign}
Note that these isomorphisms require a choice of trivialization
of the cotangent bundle $T^\ast G$, or equivalently in the algebraic language
an isomorphism of $H$-modules $H\otimes \g \stackrel{\cong}{\to} \T_{\! H} = \Der(H)$.
We choose to work with the isomorphism 
\begin{flalign}\label{tmp:isoleftinv}
H\otimes \g ~\longrightarrow~ \T_{\! H} ~~,\quad h^\prime \otimes t ~\longmapsto~h^\prime \,\rho^{\mathrm{L}}(t)
\end{flalign}
given by the assignment of left invariant derivations $\rho^{\mathrm{L}}(t)(h)
= h_{\und{1}} \,t(h_{\und{2}})$, for all $t\in\g$ and $h\in H$.
To specify the comodule structure $\rho : \O(T^\ast \Con_G^{}(U))\to \O(T^\ast \Con_G^{}(U))\otimes 
\O(\Gau(U))$ induced by \eqref{eqn:gaugeaction}, it is convenient to introduce the following
notations: For a vertex $v\in\V$, we denote by 
$h_v\in \O(\Gau(U))$ the element that is $h\in H$ on the $v$-th
tensor factor of \eqref{eqn:Ogaugegroup} and the unit element $\oone$ on all other factors.
We further denote by $t_v\in \g(U)$ and $\theta_v\in \g^\vee(U)$ 
the elements that are, respectively, $t\in \g$ or $\theta\in\g^\vee$ 
on the $v$-th direct summand of \eqref{eqn:gaugeLiealgebra} and zero on all other summands.
For \eqref{eqn:OTCon} we use a similar convention and write $t_e\in \O(T^\ast \Con_G^{}(U))$
and $a_e\in \O(T^\ast \Con_G^{}(U))$ for the generators associated with
placing, respectively, $t\otimes \oone \in \g\otimes H$ or $\oone \otimes a\in\Sym\,\g\otimes H$ 
on the $e$-th tensor factor and unit elements everywhere else. The $\O(\Gau(U))$-coaction on $\O(T^\ast \Con_G^{}(U))$
induced by \eqref{eqn:gaugeaction} then reads on these generators as
\begin{flalign}\label{eqn:gaugerhocoaction}
\rho(a_e)\,=\,a_{\und{2}\,e}\otimes S(a_{\und{1}})_{\t(e)}\,a_{\und{3}\,\s(e)}~~,\quad
\rho(t_e)\,=\,t_{\und{0}\,e}\otimes t_{\und{1}\,\s(e)}\quad, 
\end{flalign}
where $t_{\und{0}}\otimes t_{\und{1}}=\rho_{\g}^{}(t)$ is the adjoint coaction on $\g$.
Note that the coaction on $t_{e}$ does only depend on the source $\s(e)$ but not on the target
$\t(e)$ of the edge, which is due to the fact that we have chosen to work with 
left invariant derivations in \eqref{tmp:isoleftinv}.
The moment map $\mu_U^{} : T^\ast \Con_G^{}(U) \to \g^\vee(U)$ is algebraically described
by the algebra map
\begin{subequations}\label{eqn:gaugemoment}
\begin{flalign}
\mu_U^{\ast}\,:\Sym\,\g(U)~\longrightarrow~\O(T^\ast \Con_G^{}(U))
\end{flalign}
that is given on the generators $t_v\in \g(U)$ by
\begin{flalign}
\mu_U^\ast(t_v)\,=\,-\!\! \sum_{e\in \s^{-1}(v)}\!\! t_e + \!\! \sum_{e\in \t^{-1}(v)}\!\! t_{\und{0}\,e} \, S(t_{\und{1}})_{e}\quad.
\end{flalign}
\end{subequations}
Note that the moment map is a discrete analog of the Gauss constraint: 
For each vertex $v\in\V$, it compares the total outgoing canonical 
momentum ($\s(e)=v$) with the appropriately parallel transported 
total incoming canonical momentum ($\t(e)=v$).
\sk

With these preparations we can now compute the derived zero locus
$\mu_U^{-1}(0)$ of the moment map \eqref{eqn:gaugemoment},
which by \eqref{eqn:Bbulletsym} is the derived affine scheme 
defined by the following chain CDGA $\O(\mu_{U}^{-1}(0))_\bullet\in\dgCAlg_{\geq 0}$: 
The underlying graded algebra is
\begin{subequations}\label{eqn:OmuU}
\begin{flalign}\label{eqn:OmuU1}
\O(\mu_U^{-1}(0))_\sharp \,=\,  \Big(\bigotimes_{v\in\V}\Sym\,\g_{[-1]}\Big)_{\sharp} 
\otimes \bigotimes_{e\in\E } \big( \Sym\,\g \otimes H\big)
\end{flalign}
and the chain differential $\partial$ is defined on the generators by
\begin{flalign}\label{eqn:OmuU2}
\partial(a_e)\,=\,0~~,\quad
\partial(t_e)\,=\, 0~~,\quad
\partial(t_v)\,=\, \mu^{\ast}_U(t_v)\,=\, -\!\! \sum_{e\in \s^{-1}(v)}\!\! t_e + \!\! \sum_{e\in \t^{-1}(v)}\!\! t_{\und{0}\,e} \, S(t_{\und{1}})_{e}\quad.
\end{flalign}
The $\O(\Gau(U))$-coaction on $\O(\mu_{U}^{-1}(0))_\bullet$ is 
given by \eqref{eqn:gaugerhocoaction} and the adjoint coaction on $t_v\in\g(U)$.
For completeness, let us spell out the coaction on the generators
\begin{flalign}\label{eqn:OmuU3}
\rho(a_e)\,=\,a_{\und{2}\,e}\otimes S(a_{\und{1}})_{\t(e)}\,a_{\und{3}\,\s(e)}~~,\quad
\rho(t_e)\,=\,t_{\und{0}\,e}\otimes t_{\und{1}\,\s(e)}~~,\quad
\rho(t_v)\,=\, t_{\und{0}\,v}\otimes t_{\und{1}\,v} \quad.
\end{flalign}
\end{subequations}
Summing up, we have shown that the canonical phase space of our gauge theory model
on the directed graph $U = (\s,\t:\E\rightrightarrows\V)$
is given by the derived quotient stack
\begin{flalign}\label{eqn:phasespace}
\mathcal{S}(U)\,:=\, \big[\mu_U^{-1}(0)/\Gau(U)\big]\,\simeq\, T^\ast\big[\Con_G(U)/\Gau(U)\big]\quad.
\end{flalign}

The assignment $U \mapsto \mathcal{S}(U)$ of the canonical phase spaces 
to directed graphs is functorial with respect to a (rather limited) 
class of graph embeddings.
\begin{defi}\label{def:digraph}
The category $\DD$ is defined as follows:
\begin{itemize}
\item An object in $\DD$ is a directed graph 
$U = (\s,\t: \E \rightrightarrows \V)$.

\item A $\DD$-morphism $f: U\to U^\prime$
from $U = (\s,\t: \E \rightrightarrows \V)$ 
to $U^\prime = (\s^\prime,\t^\prime : \E^\prime \rightrightarrows \V)$ is a pair 
$(f_{\V}: \V \to \V^\prime , f_{\E} : \E\to\E^\prime)$ of injective maps that satisfies the following properties:
\begin{itemize}
\item[(i)] The two diagrams
\begin{flalign}
\xymatrix@R=2em@C=2em{
\ar[d]_-{\s}\E \ar[r]^-{f_{\E}} & \ar[d]^-{\s^\prime}\E^\prime && \ar[d]_-{\t}\E \ar[r]^-{f_{\E}} & \ar[d]^-{\t^\prime}\E^\prime\\
\V \ar[r]_-{f_{\V}}& \V^\prime &&\V \ar[r]_-{f_{\V}}& \V^\prime
}
\end{flalign}
commute.

\item[(ii)] For all vertices $v\in\V$, the induced 
maps of fibers $f_{\E} : \s^{-1}(v)\to \s^{\prime -1}(f_{\V}(v))$
and $f_{\E} : \t^{-1}(v)\to \t^{\prime -1}(f_{\V}(v))$ are bijections.
\end{itemize}
\end{itemize}
\end{defi}
\begin{rem}
Note that our definition of $\DD$-morphisms is much more restrictive than
the concept of graph refinements used in \cite{Baez}. The reason behind this is as follows:
While the assignment $U\mapsto [\Con_G(U)/\Gau(U)]$ of quotient stacks considered
in \cite{Baez} is contravariantly functorial with respect to general graph refinements $U\to U^\prime$,
the resulting stack morphisms $[\Con_G(U^\prime)/\Gau(U^\prime)]\to [\Con_G(U)/\Gau(U)]$
do {\em not} in general induce to the derived cotangent stacks \eqref{eqn:phasespace}
as a consequence of the limited functorial properties of cotangent bundles.
Our more restrictive concept of $\DD$-morphisms is not affected by this problem.
Property (ii) of Definition \ref{def:digraph} formalizes the idea
that the neighborhood structure of every vertex $v\in\V$ is preserved 
by an embedding of directed graphs, i.e.\ the edges in $U$ starting/ending
at $v$ correspond bijectively to the edges in $U^\prime$ starting/ending at $f_{\V}(v)$.
This is required for naturality of the moment map \eqref{eqn:gaugemoment}.
\end{rem}

To describe the functorial structure on $U \mapsto \mathcal{S}(U)$,
let us first observe that the assignment $U\mapsto \O(\mu_{U}^{-1}(0))_\bullet$ 
of the chain CDGAs defined in \eqref{eqn:OmuU} can be promoted to a functor
$\O(\mu_{(-)}^{-1}(0))_\bullet : \DD \to \dgCAlg_{\geq 0}$ on the category 
$\DD$ introduced in Definition \ref{def:digraph}. Concretely, 
to each morphism $f:U\to U^\prime$ in $\DD$ we assign the chain CDGA morphism
\begin{subequations}\label{eqn:fastmap}
\begin{flalign}
f_\ast:= \O(\mu_{f}^{-1}(0))_\bullet\,:\, \O(\mu_{U}^{-1}(0))_\bullet~\longrightarrow~\O(\mu_{U^\prime}^{-1}(0))_\bullet
\end{flalign}
that is defined on the generators by
\begin{flalign}
f_\ast(a_e)\,=\,a_{f_{\E}(e)}~~,\quad
f_\ast(t_e)\,=\,t_{f_{\E}(e)}~~,\quad
f_\ast(t_v)\,=\, t_{f_{\V}(v)}\quad.
\end{flalign}
\end{subequations}
Using property (ii) of Definition \ref{def:digraph},
one easily confirms that $f_\ast$ preserves the chain differentials \eqref{eqn:OmuU2}.
Next, we note that the assignment $U\mapsto \O(\Gau(U))$ of the 
gauge Hopf algebras \eqref{eqn:Ogaugegroup} is functorial too:
To each $\DD$-morphism $f:U\to U^\prime$ we assign the Hopf algebra morphism
\begin{subequations}\label{eqn:tildefastmap}
\begin{flalign}
\tilde{f}_\ast :=\O(\Gau(f))\,:\, \O(\Gau(U))~\longrightarrow~\O(\Gau(U^\prime))
\end{flalign}
that is defined on the generators by
\begin{flalign}
\tilde{f}_{\ast}(h_v)\,=\, h_{f_{\V}(v)}~~.
\end{flalign}
\end{subequations}
It is evident that \eqref{eqn:fastmap} is equivariant relative to \eqref{eqn:tildefastmap}, i.e.\
the diagram
\begin{flalign}
\xymatrix@R=2em@C=2em{
\ar[d]_-{f_\ast}\O(\mu_{U}^{-1}(0))_\bullet \ar[r]^-{\rho} & \O(\mu_{U}^{-1}(0))_\bullet\otimes \O(\Gau(U))\ar[d]^-{f_\ast\otimes \tilde{f}_\ast}\\
\O(\mu_{U^\prime}^{-1}(0))_\bullet  \ar[r]_-{\rho^\prime}&  \O(\mu_{U^\prime}^{-1}(0))_\bullet\otimes \O(\Gau(U^\prime))
}
\end{flalign}
involving the coactions \eqref{eqn:OmuU3} commutes.
This shows that $\mu^{-1}_{(-)}(0)\,,\,\Gau(-) : \DD^\op\to \mathbf{dSt}$ are strict functors
and the right action $\mu^{-1}_{(-)}(0)\times \Gau(-)\to \mu^{-1}_{(-)}(0)$ is a strict natural transformation, 
which passing to derived quotient stacks yields an $\infty$-functor
$\mathcal{S}  : \DD^{\op}\to \mathbf{dSt}$ that assigns the phase spaces \eqref{eqn:phasespace}.
\sk

To show that the canonical unshifted Poisson structures
on $\mathcal{S}(U) \cong T^\ast[\Con_G^{}(U)/\Gau(U)]$ are natural with 
respect to this functorial structure, one has to resolve
each $\mathcal{S}(U)$ in terms of stacky CDGAs (see Subsection \ref{subsec:resolution})
and then show that the degree-wise Poisson structures $\{\,\cdot\,,\,\cdot\,\}_n^{}$
from Subsection \ref{subsec:cotangentresolution}
are natural with respect to the induced functorial structure of the cosimplicial stacky CDGAs.
We shall now show this explicitly for the case $n=0$ and note that
the case $n\geq 1$ follows by the same arguments.
For $n=0$, we are dealing with the formal derived quotient stack
$T^\ast[\Con_G^{}(U)/\g(U)]\cong [\mu^{-1}_U(0)/\g(U)]$
obtained by replacing the action of the gauge group $\Gau(U)$ 
with that of the gauge Lie algebra \eqref{eqn:gaugeLiealgebra}.
This is described algebraically by the Chevalley-Eilenberg stacky CDGA
\begin{flalign}\label{eqn:latticeCE}
\CE^\bullet\big(\g(U), \O(\mu^{-1}_U(0))_\bullet\big)\,\in\,\DGdgCAlg\quad.
\end{flalign}
Let us spell this out fully explicitly. The underlying bigraded
algebra reads as
\begin{flalign}
\CE^\sharp\big(\g(U), \O(\mu^{-1}_U(0))_\sharp\big)\,=\, 
\Big(\bigotimes_{v\in\V}\Sym\,\g^{\vee[-1]}\Big)^{\sharp} \otimes \Big(\bigotimes_{v\in\V}\Sym\,\g_{[-1]}\Big)_{\sharp} 
\otimes \bigotimes_{e\in\E } \big( \Sym\,\g \otimes H\big)\quad.
\end{flalign}
Using \eqref{eqn:OmuU2}, we observe that the
chain differential acts on the generators as
\begin{flalign}
\partial(a_e)\,=\,0~~,\quad
\partial(t_e)\,=\, 0~~,\quad
\partial(t_v)\,=\, -\!\! \sum_{e\in \s^{-1}(v)}\!\! t_e + \!\! \sum_{e\in \t^{-1}(v)}\!\! t_{\und{0}\,e} \, S(t_{\und{1}})_{e}~~,\quad
\partial(\theta_v ) \,=\,0\quad.
\end{flalign}
Using also \eqref{eqn:OmuU3}, we obtain that 
the cochain (i.e.\ Chevalley-Eilenberg) differential acts on the generators as
\begin{align}
\nn \delta(a_e)\,&=\, \theta^b_{\t(e)} ~\rho^{\mathrm{R}}(t_b)(a)_e +
\theta^b_{\s(e)} ~\rho^{\mathrm{L}}(t_b)(a)_{e}\quad, &
\delta(t_e)\,&=\, \theta^b_{\s(e)} ~\rho_{\g}^{}(t_b)(t)_e\quad,\\
\delta(t_v)\,&=\,\theta^b_{v}~\rho_{\g}^{}(t_b)(t)_v\quad, &
\delta(\theta^a_v)\,&=\,-\tfrac{1}{2} f^a_{bc}\,\theta^b_v\,\theta^c_v\quad,
\end{align}
where $\rho^{\mathrm{L}/\mathrm{R}}$ denotes the $\g$-action on $H$ in 
terms of left/right invariant derivations (see \eqref{eqn:inducedLieactions2})
and $\rho_{\g}^{}$ denotes the adjoint $\g$-action on $\g$.
Concerning the canonical unshifted Poisson structure, 
we obtain from \eqref{eqn:Poissonlevel0} that the non-vanishing
Poisson brackets between the generators read for the present example as
\begin{flalign}\label{eqn:latticePoisson0}
\big\{t_{e},a_{e^\prime}\big\}_0^{} \,=\,\delta_{e e^\prime}\,\rho^{\mathrm{L}}(t)(a)_e~~,\quad
\big\{t_{e},t^\prime_{e^\prime}\big\}_0^{} \,=\,\delta_{e e^\prime}\,[t,t^\prime]_e ~~,\quad
\big\{t_v,\theta_{v^\prime}\big\}_0^{} \,=\, - \delta_{v v^\prime}\, \langle\theta,t\rangle \quad,
\end{flalign}
where $[\,\cdot\, ,\,\cdot\,]$ denotes the Lie bracket on $\g$, and
$\delta_{e e^\prime}$ and $\delta_{vv^\prime}$ are the Kronecker-deltas 
on the sets of edges $\E$ and vertices $\V$.
In other words, the Poisson bracket acts locally on both the edges and the vertices.
The Chevalley-Eilenberg stacky CDGAs in \eqref{eqn:latticeCE} are clearly
functorial with respect to $\DD$-morphisms $f:U\to U^\prime$
by applying the construction in Remark \ref{rem:CEalgebrafunctoriality} to the two
morphisms in \eqref{eqn:fastmap} and \eqref{eqn:tildefastmap}.
The unshifted Poisson structure defined in \eqref{eqn:latticePoisson0} 
is clearly natural with respect to this functorial structure.
The same holds true for the degree $n\geq 1$ stacky CDGAs in the 
cosimplicial diagram obtained by specializing
\eqref{eqn:cosimplicialstackyCDGA} to our case of interest 
and their Poisson structures determined by 
\eqref{eqn:Poissonleveln1} and \eqref{eqn:Poissonleveln2}.


\subsection{\label{subsec:quantization}Quantization}
Recall that, for each directed graph $U = (\s,\t:\E\rightrightarrows\V)\in\DD$, 
the phase space $\mathcal{S}(U)$ of our gauge theory \eqref{eqn:phasespace} is a derived cotangent stack together with
its canonical unshifted Poisson structure. Hence, our global quantization construction
in Subsection \ref{subsec:globalquantization} can be applied in order to assign, 
to each object $U\in \DD$, a dg-category
\begin{flalign}\label{eqn:AAAAU}
\AAAA(U)\,:=\, \mathrm{per}_\hbar\big(\mathcal{S}(U)\big)
\end{flalign}
over $\bbK[[\hbar]]$ that quantizes the dg-category of perfect modules over $\mathcal{S}(U)$.
From a physical point of view, the dg-category $\AAAA(U)$ should be interpreted
as a generalization of an algebra of quantum observables, which motivates our choice of notation $\AAAA$. 
Note that this generalization from algebras to categories is necessary because the derived stack 
$\mathcal{S}(U)$ is not affine and hence it is not captured by a single algebra. Readers who are not 
familiar with such concepts are referred to \cite{2AQFT}, where similar ideas are pursued
in a simpler $2$-categorical setup.
\sk

In the remainder of this subsection, we shall spell out the dg-categories
\eqref{eqn:AAAAU} fully explicitly and also describe their (pseudo-)functorial structure
$\AAAA : \DD\to \DGCat_{\bbK[[\hbar]]}$ on the category of directed graphs $\DD$
introduced in Definition \ref{def:digraph}. Concerning the first task,
let us recall the general description of $\mathrm{per}_\hbar(T^\ast[X/G])$ given 
in Definition \ref{def:quantizeddgCat}. In the present case, the algebra $A\in\CAlg$ is given
by the function algebra
\begin{flalign}
\O(\Con_G(U))\,=\,\bigotimes_{e\in\E} H
\end{flalign}
on the space of $G$-connections on the graph 
$U= (\s,\t:\E\rightrightarrows\V)\in\DD$.
The $\O(\Gau(U))$-comodule structure on $\O(\Con_G(U))$ is the one defined in the first 
equation in \eqref{eqn:gaugerhocoaction}. The datum in item ($1_{\hbar}$) of
Subsection \ref{subsec:globalquantization} then specializes as follows:
\begin{itemize}
\item[($1_{\hbar}$)] A left $\O(\Con_G(U))[[\hbar]]$-dg-module $V_\bullet$ with a compatible $\O(\Gau(U))$-coaction
$\rho_V^{} : V_\bullet\to V_\bullet\otimes \O(\Gau(U))$ of the gauge Hopf algebra \eqref{eqn:Ogaugegroup}.
\end{itemize}
To proceed with the other data, let us observe that $\Omega^1_{\O(\Con_G(U))}\cong 
\big(\bigoplus_{e\in \E} \g^\vee \big)\otimes \O(\Con_G(U))$ is a free $\O(\Con_G(U))$-module by using the dual of
the isomorphism in \eqref{tmp:isoleftinv}.
The data in items ($2_{\hbar}$) and ($3_{\hbar}$) of Subsection \ref{subsec:globalquantization} 
then specialize as follows:
\begin{itemize}
\item[($2_{\hbar}$)] An $\O(\Gau(U))$-equivariant chain map
$\nabla : \big(\bigoplus_{e\in\E} \g\big) \otimes V_\bullet\to V_\bullet$ satisfying the Leibniz rule
\begin{flalign}\label{eqn:gaugeLeibnizrule}
\nabla_{t_{e}}\big(a_{e^\prime}\cdot s\big)\,=\, \hbar\, \delta_{e e^\prime} \,\rho^{\mathrm{L}}(t)(a)_e \cdot s 
+ a_{e^\prime}\cdot \nabla_{t_{e}}(s)\quad,
\end{flalign}
for all $t_e\in \bigoplus_{e\in\E} \g$, $a_{e^\prime}\in \O(\Con_G(U))$ and $s\in V_\bullet$, where 
$\cdot$ denotes the module structure on $V_\bullet$.

\item[($3_{\hbar}$)] An $\O(\Gau(U))$-equivariant graded left 
$\O(\Con_G(U))[[\hbar]]$-module map $\Psi : \g(U)_{[-1]}\otimes V_\sharp\to V_\sharp$.
\end{itemize}
These structure maps have to satisfy the following properties:
\begin{itemize}
\item[($\text{i}_{\hbar}$)] For all $t_e, t^\prime_{e^\prime} \in \bigoplus_{e\in\E} \g$ and
$t_v,t^\prime_{v^\prime} \in \g(U)_{[-1]}$,
\begin{flalign}
\nn \nabla_{t_{e}}\circ \nabla_{t^\prime_{e^\prime}} - \nabla_{t^\prime_{e^\prime}} \circ \nabla_{t_{e}}
\,&=\,\hbar\,\delta_{e e^\prime} \,\nabla_{[t,t^\prime]_e} \quad,\\
\nn \nabla_{t_e} \circ \Psi_{t_v} \,&=\, \Psi_{t_v} \circ \nabla_{t_e} \quad,\\
\Psi_{t_v}\circ \Psi_{t^\prime_{v^\prime}} \,&=\,- \Psi_{t^\prime_{v^\prime}}\circ \Psi_{t_v}\quad,
\end{flalign}
where $[t,t^\prime]\in\g$ denotes the Lie bracket.

\item[($\text{ii}_{\hbar}$)] For all $t_v\in \g(U)_{[-1]}$,
\begin{flalign}
\nn \partial\circ \Psi_{t_v} +\Psi_{t_v}\circ \partial \,&=\, \nabla_{\mu^{\ast}_U(t_v)}+ \hbar\, \rho_V(t_v)\\
&\,=\,-\!\! \sum_{e\in \s^{-1}(v)}\!\! \nabla_{t_e} + \!\! \sum_{e\in \t^{-1}(v)}\!\!  S(t_{\und{1}})_{e}\cdot 
\nabla_{t_{\und{0}\,e}} + \hbar\, \rho_V(t_v)\quad,\label{eqn:tmpgausslaw}
\end{flalign}
where $\partial$ denotes the differential on $V_\bullet$ 
and in the second line we have used the explicit expression \eqref{eqn:gaugemoment} for the moment map.
\end{itemize}
\begin{rem}\label{rem:quantummechanics}
A triple $(V_\bullet,\nabla,\Psi)$ as described above is a (dg and equivariant) generalization of a familiar concept from quantum mechanics:
The quantum mechanical analogs of $\O(\Con_G(U))$ in ($1_{\hbar}$) and 
of $\bigoplus_{e\in\E} \g$ in ($2_{\hbar}$) are respectively 
the algebra of position operators and the space of canonical momenta, which act on $V_\bullet$ via $\nabla$. The Leibniz
rule \eqref{eqn:gaugeLeibnizrule} and the flatness condition in ($\text{i}_{\hbar}$)
enforce that this action satisfies the canonical commutation relations determined by the Poisson 
bracket \eqref{eqn:latticePoisson0}. 
Therefore, the module $V_\bullet$ plays the same role as the `space of wave functions' in quantum mechanics.
The $\O(\Gau(U))$-comodule structure on $V_\bullet$ means that the `wave functions' transform under the gauge 
symmetries of the system, which is a distinct feature of gauge theories that is typically not
present in ordinary quantum mechanics, and $\O(\Gau(U))$-equivariance of $\nabla$ and $\Psi$ 
simply enforces compatibility of these structures with gauge symmetries.
The physical interpretation
of $\g(U)_{[-1]}$ in ($3_{\hbar}$) is that of anti-ghosts (in the context of the BFV formalism)
and the map $\Psi$ provides an action of these anti-ghosts on the `wave functions'. 
Item ($\text{i}_{\hbar}$) expresses that this action is compatible with the rest of the structure
and item ($\text{ii}_{\hbar}$) implements homologically the quantum Gauss constraint
for the `wave functions'. The $\hbar$-dependent factor in \eqref{eqn:tmpgausslaw} 
connects the anti-ghosts to the $\O(\Gau(U))$-comodule structure and it 
can be interpreted as a non-infinitesimal analog of the non-trivial commutation relations
between the anti-ghost and the ghosts of the traditional BFV formalism.
\end{rem}

\begin{rem}
We would like to comment briefly on the relationship
between our quantized dg-category $\mathrm{per}_\hbar(T^\ast[X/G])$
and the approach in \cite{Pflaum}, which is based on BRST deformation quantization techniques 
that have been developed in \cite{Bordemann1,Bordemann2}. The main
difference is that these papers consider infinitesimal gauge transformations,
i.e.\ they work with the gauge Lie algebra \eqref{eqn:gaugeLiealgebra}, 
while our quantized dg-category $\mathrm{per}_\hbar(T^\ast[X/G])$
takes into account the gauge group \eqref{eqn:gaugegroup}. When translated
to the language of our paper, this means that these authors construct 
an analog (in the context of differential geometry and deformation quantization)
of our noncommutative cochain dg-algebra $\AAA^{0\bullet}_{\hbar}$, which is however
only the first component of the cosimplicial diagram in Proposition 
\ref{propo:cosimplicialquantum} that encodes the quantization of $T^\ast[X/G]$.
\end{rem}

Summing up, let us record
\begin{cor}
For each directed graph $U = (\s,\t:\E\rightrightarrows\V)\in\DD$, 
the quantized dg-category $\AAAA(U)$ in \eqref{eqn:AAAAU} admits the following explicit description:
Its objects are all triples $(V_\bullet,\nabla,\Psi)$ as described above 
that satisfy the conditions from Definition \ref{def:quantizeddgCat}.
The cochain complexes (over $\bbK[[\hbar]]$) of morphisms consist of all left $\O(\Con_G(U))[[\hbar]]$-module maps satisfying 
the compatibility conditions for the $\O(\Gau(U))$-coactions, $\nabla$ and $\Psi$
listed in Definition \ref{def:quantizeddgCat}.
\end{cor}

We conclude by spelling out rather explicitly the (pseudo-)functorial structure of the 
assignment $U\mapsto \AAAA(U)$ of the dg-categories in \eqref{eqn:AAAAU}.
This structure can be deduced from the obvious functoriality of the local quantizations
$U \mapsto \AAA_{\hbar}^{n\bullet}(U)$ (the latter are obtained by specializing Subsection 
\ref{subsec:localquantization} to our example) and passing to dg-categories 
\eqref{eqn:quantizeddgCATholim} via the induced module functor.
The resulting dg-functor $\AAAA(f) : \AAAA(U)\to \AAAA(U^\prime)$
associated to a $\DD$-morphism $f : U\to U^\prime$ (see Definition \ref{def:digraph})
admits a concrete description that is similar to our construction of the object 
$B_{\hbar\bullet}$ in the Examples \ref{ex:globalrank1module} and \ref{ex:rank1module}.
To emphasize this analogy, let us introduce the graded $\bbK[[\hbar]]$-module
\begin{flalign}\label{eqn:Bhbarf}
B_{\hbar\sharp}(f)\,:=\, \!\!\!\bigotimes_{e^\prime\in\E^\prime\setminus f_{\E}(\E)} \!\!\!\mathrm{DiffOp}_\hbar(H[[\hbar]])\, \otimes \!\!\!\bigotimes_{v^\prime\in \V^{\prime}\setminus f_{\V}(\V)} \!\!\! \big(\Sym\,\g_{[-1]}\big)_\sharp 
\end{flalign}
that is obtained by forming tensor products over the edges 
and vertices of $U^\prime$ that are not contained in the image of the morphism $f$.
Due to locality of the gauge symmetries (see e.g.\ \eqref{eqn:OmuU3}),
the graded module $B_{\hbar\sharp}(f)$ carries a canonical comodule
structure for the Hopf algebra $\O(\Gau(U^\prime))$ associated with the bigger graph $U^\prime$.
Given any object $(V_\bullet, \nabla,\Psi)\in \AAAA(U)$, we can endow
$V_\sharp$ with the $\O(\Gau(U^\prime))$-comodule structure
\begin{flalign}
\xymatrix@C=2.5em{
V_\sharp \ar[r]^-{\rho_V^{}} & V_\sharp \otimes \O(\Gau(U))\ar[r]^-{\id\otimes \tilde{f}_\ast} &V_\sharp \otimes \O(\Gau(U^\prime))
}
\end{flalign}
obtained by composition with \eqref{eqn:tildefastmap}. We define
the underlying graded $\O(\Gau(U^\prime))$-comodule of 
$\AAAA(f)(V_\bullet, \nabla,\Psi)\in \AAAA(U^\prime)$ by
\begin{flalign}\label{eqn:AAAAfobject}
\AAAA(f)\big(V_\bullet,\nabla,\Psi\big)_\sharp \,=\, B_{\hbar\sharp}(f)\otimes_{\bbK[[\hbar]]} V_\sharp
\end{flalign}
and the tensor product $\O(\Gau(U^\prime))$-coaction.
We denote elements of this comodule by symbols like
$\hat{D}\otimes \hat{\chi}\otimes s\in B_{\hbar\sharp}(f)\otimes V_\sharp$ in 
order to stress that also $B_{\hbar\sharp}(f)$ is given by a tensor product \eqref{eqn:Bhbarf}.
Whenever necessary, we shall further write $\hat{\chi}=(\hat{t}_1)_{v^\prime_1}\cdots (\hat{t}_n)_{v^\prime_n}$
in order to indicate in which factors of the tensor product of symmetric algebras \eqref{eqn:Bhbarf}
the element is living. The graded $\O(\Con_G(U^\prime))[[\hbar]]$-module structure
is then given on elements by
\begin{flalign}\label{eqn:AAAfobject2}
a_{e^\prime}\cdot\big(\hat{D}\otimes \hat{\chi}\otimes s\big)\,=\,\begin{cases}
\hat{D}\otimes\hat{\chi}\otimes (a_{e}\cdot s) &~~\text{for } e^\prime = f_{\E}(e) \in f_{\E}(\E)~~,\\
(\hat{a}_{e^\prime} \,\hat{D})\otimes\hat{\chi}\otimes s&~~\text{else}~~,
\end{cases}
\end{flalign}
where the element $e\in \E$ in the first line is unique because $f_{\E}:\E\to\E^\prime$
is injective by Definition \ref{def:digraph} 
and the second line uses the multiplication of differential operators.
The chain differential is obtained in a similar way as the one in 
\eqref{eqn:Bhbardifferential} and it reads as
\begin{flalign}
\nn &\partial\Big(\hat{D} \otimes (\hat{t}_1)_{v^\prime_1}\cdots (\hat{t}_n)_{v^\prime_n}\otimes s\Big)
\,=\, (-1)^n\,\hat{D} \otimes (\hat{t}_1)_{v^\prime_1}\cdots (\hat{t}_n)_{v^\prime_n}\otimes \partial(s)\\
\nn &\qquad\qquad\qquad + \sum_{i=1}^n (-1)^{i-1} \hat{D}\, \widehat{\mu^\ast_{U^\prime}((t_{i})_{v^\prime_i})}
 \otimes  (\hat{t}_1)_{v^\prime_1}\cdots (\check{t_i})_{v^\prime_i} \cdots (\hat{t}_n)_{v^\prime_n} \otimes s\\
 &\qquad\qquad\qquad + \sum_{i=1}^n (-1)^{i-1} \hat{D}\otimes \big[(\hat{t}_1)_{v^\prime_1}\cdots (\hat{t}_{i-1})_{v^\prime_{i-1}},(\hat{\theta}^{b})_{v_i^\prime}\big]\,\widehat{\rho_\g^{}(t_b)(t_i)}_{v^\prime_i}\cdots (\hat{t}_n)_{v^\prime_n}\otimes s\quad, \label{eqn:AAAfobject3}
\end{flalign}
where $\check{-}$ denotes omission of the corresponding factor and $[\,\cdot\,,\,\cdot\,]$ is 
the graded commutator of differential operators.
The flat connection and the $\Psi$-map on the resulting $\O(\Con_G(U^\prime))[[\hbar]]$-dg-module
$\AAAA(f)\big(V_\bullet,\nabla,\Psi\big)$ are given on elements by
\begin{flalign}\label{eqn:AAAfobject4}
\nabla_{t_{e^\prime}}\big(\hat{D}\otimes\hat{\chi}\otimes s\big)\,=\, \begin{cases}
\hat{D}\otimes \hat{\chi} \otimes\nabla_{t_e}(s)&~~\text{for }e^\prime = f_{\E}(e) \in f_{\E}(\E)~~,\\
(\hat{t}_{e^\prime} \,\hat{D})\otimes\hat{\chi}\otimes s&~~\text{else}~~,
\end{cases}
\end{flalign}
and
\begin{flalign}\label{eqn:AAAfobject5}
\Psi_{t_{v^\prime}}\big(\hat{D}\otimes\hat{\chi}\otimes s\big)\,=\, \begin{cases}
(-1)^{\vert \chi\vert} ~\hat{D}\otimes \hat{\chi}\otimes \Psi_{t_{v}}(s)&~~\text{for }v^\prime = 
f_{\V}(v) \in f_{\V}(\V)~~,\\
\hat{D}\otimes(\hat{t}_{v^\prime}\,\hat{\chi})\otimes s&~~\text{else}~~,
\end{cases}
\end{flalign}
where again the element $v\in \V$ is unique because $f_{\V}:\V\to\V^\prime$
is injective by Definition \ref{def:digraph}.
The action of the dg-functor $\AAAA(f) : \AAAA(U)\to \AAAA(U^\prime)$ on morphism complexes
is simply given by taking tensor products $\id\otimes_{\bbK[[\hbar]]} (-)$ to let
morphisms act on the second tensor factor in \eqref{eqn:AAAAfobject}.
\sk

Summing up, we obtain
\begin{propo}\label{prop:AAAApseudofunctor}
The construction above defines a pseudo-functor
$\AAAA : \DD\to \DGCat_{\bbK[[\hbar]]}$ to the category of 
dg-categories over $\bbK[[\hbar]]$. The coherences for compositions
and identities are given by the canonical coherence isomorphisms for tensor products.
\end{propo}


\subsection{\label{subsec:factorizationproducts}Prefactorization algebra structure}
The pseudo-functorial structure from Proposition \ref{prop:AAAApseudofunctor}
can be upgraded to a (weak) prefactorization algebra structure on the collection
of dg-categories $\AAAA(U)$ in \eqref{eqn:AAAAU}. As explained in \cite{2AQFT}, prefactorization algebras
with values in (dg-)categories provide a categorification of the concept of
algebraic quantum field theories, which are traditionally assumed to take values 
in algebras instead of in categories. Category-valued prefactorization algebras
also appeared before in the works of Ben-Zvi, Brochier and Jordan \cite{BZBJ,BZBJ2}.
\sk

Before we can describe our prefactorization algebra structure, we have to introduce
the relevant prefactorization operad over the category $\DD$ of directed graphs 
from Definition \ref{def:digraph}. This is a special case of the family of prefactorization operads
that are assigned to so-called {\em orthogonal categories} \cite{2AQFT}, which vastly generalize the 
original prefactorization operad on manifolds introduced by Costello and Gwilliam 
\cite{CostelloGwilliam,CostelloGwilliam2}. 
\begin{defi}
We say that a pair of $\DD$-morphisms
$f_{1} : U_1\to U^\prime \leftarrow U_2 : f_2$ to a common target is {\em orthogonal},
written as $f_1\perp f_2$, if and only if
$f_{1\E}(\E_1)\cap f_{2\E}(\E_2)\,=\,\emptyset$ and 
$f_{1\V}(\V_1)\cap f_{2\V}(\V_2) = \emptyset$.
The {\em prefactorization operad} $\P_{\ovr{\DD}}$ associated with the orthogonal
category $\ovr{\DD} := (\DD,\perp)$ is then defined as follows:
Its objects are all directed graphs $U\in \DD$ and the set of operations
from the tuple $\und{U} := (U_1,\dots,U_n)$ to $U^\prime$ is given by
\begin{flalign}
\P_{\ovr{\DD}}\big(\substack{U^\prime \\ \und{U}}\big)\,:=\,
\Big\{\und{f} := (f_1,\dots,f_n) \in \prod_{i=1}^n \Hom_{\DD}^{}(U_i,U^\prime)\,:\, f_i\perp f_j\text{ for all } i\neq j \Big\}\quad.
\end{flalign}
For the empty tuple $\und{U}=\emptyset$, we set $\P_{\ovr{\DD}}\big(\substack{U^\prime\\ \emptyset}\big) := 
\{\ast\}$ to be a singleton set.
The operadic composition is defined by compositions in $\DD$, the operadic units 
are the identity morphisms in $\DD$ and the permutation groups act via permutations
on the tuples $\und{f} = (f_1,\dots,f_n)$. See also \cite[Definition 2.5]{2AQFT} for a more detailed and 
fully explicit description of $\P_{\ovr{\DD}}$.
\end{defi}

Our desired prefactorization algebra structure on the family of dg-categories
$\AAAA(U)$ is given by a (weak) $\P_{\ovr{\DD}}$-algebra in the symmetric monoidal category
$\DGCat_{\bbK[[\hbar]]}$ of dg-categories over $\bbK[[\hbar]]$. The
symmetric monoidal structure on $\DGCat_{\bbK[[\hbar]]}$ is spelled out
explicitly in e.g.\ \cite[Section 2.3]{Keller}. In short, the tensor product
$\mathcal{C}\otimes \mathcal{D}$ of two dg-categories
$\mathcal{C}$ and $\mathcal{D}$ is the dg-category whose
objects are pairs $(c,d)$ of objects $c\in\mathcal{C}$ and $d\in\mathcal{D}$
and whose morphism complexes 
\begin{flalign}
\hom_{\mathcal{C}\otimes\mathcal{D}}^{}\big((c,d), (c^\prime,d^\prime)\big)\,:=\,
\hom_{\mathcal{C}}^{}(c,c^\prime)\otimes \hom_{\mathcal{D}}^{}(d,d^\prime)
\end{flalign}
are given by tensor products of cochain complexes. The monoidal unit is the dg-category
with a single object and morphism complex $\bbK[[\hbar]]$.
So our task is to construct, for each operation 
$\und{f}\in \P_{\ovr{\DD}}\big(\substack{U^\prime \\ \und{U}}\big)$
in the prefactorization operad, a dg-functor
\begin{flalign}\label{eqn:AAAAundf}
\AAAA(\und{f})\,:\,\bigotimes_{i=1}^n \AAAA(U_i)~\longrightarrow~\AAAA(U^\prime)\quad,
\end{flalign}
which in the special case of an empty tuple $\und{U}=\emptyset$ amounts
to specifying an object in $\AAAA(\ast)\in \AAAA(U^\prime)$. 
The construction of the dg-functors $\AAAA(\und{f})$ (and also of the objects 
$\AAAA(\ast)\in \AAAA(U^\prime)$) is given by a relatively straightforward generalization 
of our dg-functors $\AAAA(f)$ constructed at the end of Subsection \ref{subsec:quantization}.
Similarly to the graded $\bbK[[\hbar]]$-module in \eqref{eqn:Bhbarf}, let us define
\begin{flalign}
B_{\hbar\sharp}(\und{f})\, :=\, \!\!\!\bigotimes_{e^\prime\in\E^\prime\setminus \bigsqcup_{i=1}^n f_{\E}(\E_i)} \!\!\!\mathrm{DiffOp}_\hbar(H[[\hbar]])\, \otimes \!\!\!\bigotimes_{v^\prime\in \V^{\prime}\setminus \bigsqcup_{i=1}^n f_{\V}(\V_i)} \!\!\! \big(\Sym\,\g_{[-1]}\big)_\sharp \quad.
\end{flalign}
For the special case of an empty tuple, we define $B_{\hbar\sharp}(\ast)$ 
to be tensor product over all edges and vertices, in analogy to the object constructed in
Example \ref{ex:rank1module}. The action of the dg-functor \eqref{eqn:AAAAundf} on an object
$\{(V_{i \bullet},\nabla_i,\Psi_i)\}_{i=1}^n\in \bigotimes_{i=1}^n \AAAA(U_i)$ is then
defined as in \eqref{eqn:AAAAfobject} by
\begin{flalign}\label{eqn:AAAAundfobject}
\AAAA(\und{f})\big(\{(V_{i \bullet},\nabla_i,\Psi_i)\}_{i=1}^n\big)_\sharp\,:=\,
 B_{\hbar\sharp}(\und{f})\otimes_{\bbK[[\hbar]]} V_{1\sharp}\otimes_{\bbK[[\hbar]]} \cdots\otimes_{\bbK[[\hbar]]}  V_{n\sharp}\quad,
\end{flalign}
which we endow with the obvious tensor product $\O(\Gau(U^\prime))$-coaction.
For an empty tuple, this should be read as $\AAAA(\ast)_\sharp := B_{\hbar\sharp}(\ast)$,
i.e.\ there are no tensor products with modules $V_{i\sharp}$ in this case.
The $\O(\Con_{G}(U^\prime))[[\hbar]]$-dg-module structure is given by the obvious generalization
of \eqref{eqn:AAAfobject2} and \eqref{eqn:AAAfobject3}. The flat connection
and $\Psi$-map on $\AAAA(\und{f})\big(\{(V_{i \bullet},\nabla_i,\Psi_i)\}_{i=1}^n\big)$ 
are defined in complete analogy to \eqref{eqn:AAAfobject4} and \eqref{eqn:AAAfobject5}.
Note that for the empty tuple this yields precisely the distinguished object constructed
in Examples \ref{ex:globalrank1module} and  \ref{ex:rank1module}.
The action of the dg-functor $\AAAA(\und{f}) : \Motimes_{i=1}^n\AAAA(U_i)\to \AAAA(U^\prime)$ on morphism complexes
is again simply given by taking tensor products $\id\otimes_{\bbK[[\hbar]]} (-)$ to let
morphisms act on the tensor factors corresponding to the $V_{i\sharp}$'s in \eqref{eqn:AAAAundfobject}.
\sk

Using the canonical coherence isomorphisms associated with tensor products,
one easily confirms that this construction yields a weak $\P_{\ovr{\DD}}$-algebra
(or, in other words, a pseudo-multifunctor) with values in $\DGCat_{\bbK[[\hbar]]}$.
The relevant axioms are analogous to the ones that are spelled out in \cite[Definition 3.3 and Remark 3.4]{2AQFT}.
Summing up, we have shown
\begin{propo}
The construction above defines a weak prefactorization algebra (i.e.\ a pseudo-multifunctor)
$\AAAA \in \Alg^{\mathrm{weak}}_{\P_{\ovr{\DD}}}\big(\DGCat_{\bbK[[\hbar]]}\big)$ 
on the orthogonal category of directed graphs $\ovr{\DD}=(\DD,\perp)$
with values in the symmetric monoidal category of dg-categories over $\bbK[[\hbar]]$. 
\end{propo}


\section*{Acknowledgments}
The work of M.B.\ is fostered by 
the National Group of Mathematical Physics (GNFM-INdAM (IT)). 
A.S.\ gratefully acknowledges the financial support of 
the Royal Society (UK) through a Royal Society University 
Research Fellowship (URF\textbackslash R\textbackslash 211015)
and the Enhancement Awards (RGF\textbackslash EA\textbackslash 180270, 
RGF\textbackslash EA\textbackslash 201051 and RF\textbackslash ERE\textbackslash 210053).



\begin{thebibliography}{10}

\bibitem[AC21]{AnelCalaque}
M.~Anel and D.~Calaque,
``Shifted symplectic reduction of derived critical loci,'' 
{\em to appear in Advances in Theoretical and Mathematical Physics}
[arXiv:2106.06625 [math.AG]].


\bibitem[AO21]{holimdgCat}
S.~Arkhipov and S.~Orsted,
``Homotopy limits in the category of dg-categories in terms of $A_\infty$-comodules,'' 
European Journal of Mathematics {\bf 7}, 671--705 (2021) 
[arXiv:1812.03583 [math.CT]].


\bibitem[BMR14]{Riehl}
T.~Barthel, J.~P.~May and E.~Riehl, 
``Six model structures for DG-modules over DGAs: model category theory in homological action,''
New York J.\ Math.\ {\bf 20}, 1077--1159 (2014) 
[arXiv:1310.1159 [math.CT]].


\bibitem[Bae96]{Baez}
J.~C.~Baez,
``Spin network states in gauge theory,''
Adv.\ Math.\ \textbf{117}, 253--272 (1996)
[arXiv:gr-qc/9411007 [gr-qc]].


\bibitem[BBS20]{LinearYM}
M.~Benini, S.~Bruinsma and A.~Schenkel,
``Linear Yang-Mills theory as a homotopy AQFT,''
Commun.\ Math.\ Phys.\ \textbf{378}, no.\ 1, 185--218 (2020)
[arXiv:1906.00999 [math-ph]].


\bibitem[BPSW21]{2AQFT}
M.~Benini, M.~Perin, A.~Schenkel and L.~Woike,
``Categorification of algebraic quantum field theories,''
Lett.\ Math.\ Phys.\ \textbf{111}, 35 (2021)
[arXiv:2003.13713 [math-ph]].


\bibitem[BSS21]{BVgroup}
M.~Benini, P.~Safronov and A.~Schenkel,
``Classical BV formalism for group actions,''
Commun.\ Contemp.\ Math.\ {\bf 25}(1), 2150094 (2023)
[arXiv:2104.14886 [math-ph]].


\bibitem[BS19]{BSreview} 
M.~Benini and A.~Schenkel,
``Higher structures in algebraic quantum field theory,''
Fortsch.\ Phys.\  {\bf 67}, no.\ 8--9, 1910015 (2019)
[arXiv:1903.02878 [hep-th]].


\bibitem[BSW19]{BSWhomotopy}
M.~Benini, A.~Schenkel and L.~Woike,
``Homotopy theory of algebraic quantum field theories,''
Lett.\ Math.\ Phys.\ \textbf{109}, no.\ 7, 1487--1532 (2019)
[arXiv:1805.08795 [math-ph]].


\bibitem[BZBJ18a]{BZBJ}
D.~Ben-Zvi, A.~Brochier and D.~Jordan,
``Integrating quantum groups over surfaces,''
J.\ Topol.\ {\bf 11}, no.\ 4, 874--917 (2018)
[arXiv:1501.04652 [math.QA]].


\bibitem[BZBJ18b]{BZBJ2}
D.~Ben-Zvi, A.~Brochier and D.~Jordan, 
``Quantum character varieties and braided module categories,'' 
Sel.\ Math.\ New Ser.\ {\bf 24}, 4711--4748 (2018)
[arXiv:1606.04769 [math.QA]].


\bibitem[BHP07]{Bordemann2}
M.~Bordemann, H.-C.~Herbig and M.~J.~Pflaum,
``A homological approach to singular reduction in deformation quantization,'' 
in: D.~Ch{\'e}niot, N.~Dutertre, C.~Murolo, D.~Trotman and A.~Pichon (eds.),
{\it Singularity Theory: Proceedings of the 2005 Marseille Singularity School and Conference}, 
World Scientific Publication (2007)
[arXiv:math-ph/0603078].


\bibitem[BHW00]{Bordemann1}
M.~Bordemann, H.-C.~Herbig and S.~Waldmann,
``BRST Cohomology and Phase Space Reduction in Deformation Quantisation,''
Commun.\ Math.\ Phys.\ {\bf 210}, 107--144 (2000)
[arXiv:math/9901015[math.QA]].


\bibitem[Cal19]{CalaqueTangent}
D.~Calaque, 
``Shifted cotangent stacks are shifted symplectic,''
Annales de la facult{\'e} des sciences de Toulouse {\bf 28}, no.\ 1, 67--90 (2019)
[arXiv:1612.08101 [math.AG]].


\bibitem[CPTVV17]{DAG2}
D.~Calaque, T.~Pantev, B.~To{\"e}n, M.~Vaqui{\'e} and G.~Vezzosi,
``Shifted Poisson structures and deformation quantization,''
J.\ Topol.\ {\bf 10}, no.\ 2, 483 (2017)
[arXiv:1506.03699 [math.AG]].


\bibitem[CG17]{CostelloGwilliam}
K.~Costello and O.~Gwilliam,
{\it Factorization algebras in quantum field theory: Volume 1},
New Mathematical Monographs,
Cambridge University Press, Cambridge (2017).


\bibitem[CG21]{CostelloGwilliam2}
K.~Costello and O.~Gwilliam,
{\it Factorization algebras in quantum field theory: Volume 2},
New Mathematical Monographs,
Cambridge University Press, Cambridge (2021).


\bibitem[DG13]{DrinfeldGaitsgoryFinAlgStacks}
V.~Drinfeld and D.~Gaitsgory,
``On Some Finiteness Questions for Algebraic Stacks,''
Geom.\ Funct.\ Anal.\ {\bf 23}, 149--294 (2013)
[arXiv:1108.5351 [math.AG]].
 
 
\bibitem[GR11]{GaitsgoryRozenblyumCrystals}
D.~Gaitsgory and N.~Rozenblyum,
``Crystals and $D$-modules,''
arXiv:1111.2087 [math.AG].


\bibitem[GR17]{GaitsgoryRozenblyumBook}
D.~Gaitsgory and N.~Rozenblyum, 
{\it A study in derived algebraic geometry I \& II}, 
Mathematical Surveys and Monographs {\bf 221}, 
American Mathematical Society, Providence, RI (2017).


\bibitem[HS98]{HirschowitzSimpson}
A.~Hirschowitz and C.~Simpson,
``Descente pour les $n$-champs (Descent for $n$-stacks),''
arXiv:math/9807049 [math.AG].


\bibitem[Kel06]{Keller}
B.~Keller, 
``On differential graded categories,'' 
International Congress of Mathematicians Vol.\ II, 151--190, 
Eur.\ Math.\ Soc., Z\"urich (2006)
[arXiv:math/0601185 [math.KT]].


\bibitem[PTVV13]{DAG}
T.~Pantev, B.~To{\"e}n, M.~Vaqui{\'e} and G.~Vezzosi,
``Shifted symplectic structures,''
Publ.\ Math.\ Inst.\ Hautes {\'E}tudes Sci.\ {\bf 117}, 271 (2013)
[arXiv:1111.3209 [math.AG]]. 


\bibitem[PRS21]{Pflaum}
M.~J.~Pflaum, G.~Rudolph and M.~Schmidt,
``Deformation quantization and homological reduction of a lattice gauge model,''
Commun.\ Math.\ Phys.\ \textbf{382}, no.\ 2, 1061--1109 (2021)
[arXiv:1912.12819 [math-ph]].


\bibitem[Pos11]{positselski}
L.~Positselski,
``Two kinds of derived categories, Koszul duality, and comodule-contramodule correspondence,''
Mem.\ Amer.\ Math.\ Soc.\ {\bf 212}, no.\ 996 (2011)
[arXiv:0905.2621 [math.CT]].


\bibitem[Pri17]{PridhamPoisson}
J.~P.~Pridham,
``Shifted Poisson and symplectic structures on derived $N$-stacks,''
Journal of Topology {\bf 10},  178--210 (2017)
[arXiv:1504.01940 [math.AG]].


\bibitem[Pri18a]{PridhamUnshifted}
J.~P.~Pridham, 
``Deformation quantisation for unshifted symplectic structures on derived Artin stacks,'' 
Sel.\ Math.\ New Ser.\ {\bf 24}, 3027--3059 (2018)
[arXiv:1604.04458 [math.AG]]. 


\bibitem[Pri18b]{PridhamOverview}
J.~P.~Pridham, 
``An outline of shifted Poisson structures 
and deformation quantisation in derived differential geometry,''
arXiv:1804.07622 [math.DG].


\bibitem[Saf17]{SafronovImplosion}
P.~Safronov,
``Symplectic implosion and the Grothendieck-Springer resolution,''
Transform.\ Groups {\bf 22},  767--792 (2017)
[arXiv:1411.2962 [math.AG]].


\bibitem[Sim96]{simpsonhodgefil}
C.~Simpson
``The Hodge filtration on nonabelian cohomology,''
arXiv:alg-geom/9604005.
  

\bibitem[Tab05]{tabuadaMCdgcat}
G.~Tabuada, 
``Une structure de cat{\'e}gorie de mod{\`e}les de Quillen sur la cat{\'e}gorie des dg-cat{\'e}gories,''
C.\ R.\ Acad.\ Sci.\ Paris {\bf 340}, 15--19  (2005)
[arXiv:math/0407338 [math.KT]].


\bibitem[TV08]{hag2}
B.~To{\"e}n and G.~Vezzosi,
``Homotopical algebraic geometry II: Geometric stacks and applications,'' 
Mem.\ Amer.\ Math.\ Soc.\ {\bf 193}, no.\ 902 (2008)
[arXiv:math/0404373 [math.AG]].
 

\bibitem[Toe14a]{Toen}
B.~To\"en, 
``Derived algebraic geometry,'' 
EMS Surv.\ Math.\ Sci.\ {\bf 1}, no.\ 2, 153--240 (2014)
[arXiv:1401.1044 [math.AG]].


\bibitem[Toe14b]{ToenQuant}
B.\ To\"en, 
``Derived algebraic geometry and deformation quantization,'' 
Proceedings of the International Congress of Mathematicians, Seoul (2014) 
[arXiv:1403.6995 [math.AG]].


\bibitem[Yeu21]{Yeung} 
WK.~Yeung,
``Shifted symplectic and Poisson structures on global quotients,'' 
arXiv:2103.09491 [math.AG].
 

\end{thebibliography}
\end{document}